\documentclass[10pt]{article}

\usepackage[T1]{fontenc}
\usepackage[utf8]{inputenc}
\usepackage{lmodern}
\usepackage[a4paper,margin=0.85in]{geometry}
\usepackage{amsmath,amssymb,amsthm,mathtools}
\usepackage{mathrsfs}
\usepackage{booktabs}
\usepackage{array}
\usepackage{microtype}
\usepackage[colorlinks=true,
            linkcolor=blue,
            citecolor=blue,
            urlcolor=blue]{hyperref}
\usepackage[nameinlink,noabbrev]{cleveref}

\newtheorem{theorem}{Theorem}[section]
\newtheorem{proposition}[theorem]{Proposition}
\newtheorem{lemma}[theorem]{Lemma}
\newtheorem{corollary}[theorem]{Corollary}
\theoremstyle{definition}
\newtheorem{definition}[theorem]{Definition}
\newtheorem{example}[theorem]{Example}

\DeclareMathOperator{\Herm}{Herm}
\DeclareMathOperator{\rank}{rank}
\DeclareMathOperator{\Tr}{Tr}
\DeclareMathOperator{\spanR}{span_{\mathbb R}}
\DeclareMathOperator{\spanC}{span_{\mathbb C}}
\DeclareMathOperator{\Lie}{Lie}
\DeclareMathOperator{\erfc}{erfc}
\DeclareMathOperator{\Ad}{Ad}
\DeclareMathOperator{\Mat}{Mat}

\newcommand{\R}{\mathbb R}
\newcommand{\C}{\mathbb C}
\newcommand{\A}{\mathcal A}
\newcommand{\Hh}{\mathcal H}
\newcommand{\cR}{\mathscr R}
\newcommand{\cE}{\mathscr E}
\newcommand{\cB}{\mathscr B}
\newcommand{\cV}{\mathscr V}

\title{\bfseries Behavioral Memory under Symmetry in One-Way Quantum Automata}
\author{Zeyu Chen}
\date{}

\begin{document}
\maketitle

\begin{abstract}
Under compact symmetry, observable behavior reduces to an invariant operator algebra, but its dimension is not yet classical memory: some coordinates are dynamically frozen, some invisible to threshold tests, and some already classical.  We develop an operator-algebraic theory that separates these effects through three filters.  For one automaton, behavior is the Hilbert--Schmidt pairing between prefix-reachable states and suffix-observable effects, whose rank equals the real Hankel rank without controllability or observability assumptions.  Maximizing this invariant over a symmetry-constrained dynamical class gives a structural capacity controlled by the symmetry commutant: its center stores isotypic populations frozen by reversible dynamics, its traceless multiplicity blocks carry movable noncommutative coordinates, dissipation removes the unary spectral loss inside those blocks, and covariant mobility releases relative populations subject to component conservation.  Operational realization then determines which surviving coordinates force probabilistic states.  For a fixed nontrivial invariant readout, full mobility gives an exact dichotomy in worst-case state cost: a commutative invariant algebra costs exactly its dimension, whereas a noncommutative multiplicity block raises the unrestricted cost by exactly one state.  Thus noncommutativity has a one-state worst-case classical price.  The known four-letter quadratic-plus-one law at trivial symmetry is the fully mobile endpoint of this principle.  Schur--Weyl duality further shows that different preserved symmetries on the same tensor-power Hilbert space can change the worst memory scale from polynomial to exponential, while fixed-weight modules give an exact Catalan law at half filling, with structural capacity equal to the Catalan count minus its central-sector correction.
\end{abstract}

\section{Introduction}

Finite automata make memory a discrete computational resource.  In one-way quantum automata the evolving memory is a finite-dimensional operator state, and under compact symmetry its observable behavior can be reduced to an invariant operator algebra rather than the full matrix space.  A smaller operator space, however, is not yet a smaller classical memory: some surviving coordinates are frozen by the dynamics, some movable coordinates are invisible to a fixed threshold readout, and some are already classical stochastic degrees of freedom.  Measure-once quantum finite automata isolate reversible evolution \cite{MooreCrutchfield,BrodskyPippenger}; one-way general quantum finite automata assign a quantum channel to each letter; and measure-many automata interleave the updates with halting measurements \cite{LiEtAl1gQFA,KondacsWatrous}.  Across these models we ask one question: under compact symmetry, which operator degrees of freedom survive as distinguishable word behavior, and which of them force states in every probabilistic finite automaton recognizing the same strict-cutpoint language?

The standard linear theory identifies the rank of a word function's Hankel matrix with the dimension of its minimal weighted-automaton representation \cite{BerstelReutenauer,KieferEtAl}; classical conversion then turns such a real representation into a probabilistic automaton \cite{Rabin,Paz,Turakainen}.  Prepare--test and sign-rank constructions established the ambient \(\Theta(N^2)\) strict-cutpoint scale for reversible automata \cite{ChenWuSimulation}, while dynamic shattering gives the exact unsymmetrized general-channel cost \(N^2+1\), already over four letters, together with the codimension-one stochastic embedding used here \cite{ChenWuQuadratic}.  These results determine the scale when the full operator space is available.  They do not explain what replaces that space after symmetry, why reversible and dissipative models retain different coordinates, or when a continuous operator dimension becomes an additional classical threshold state.

The paper's central answer is that behavioral memory is produced by three logically distinct filters, and the strongest general-channel endpoint is governed by a sharp algebraic boundary.  For fully mobile covariant channels with a nontrivial invariant readout, a commutative symmetry commutant is already a classical stochastic memory: the worst strict-cutpoint cost equals the dimension of that invariant algebra, even over a binary alphabet.  Once the commutant has a noncommutative multiplicity block, the unrestricted-alphabet cost is exactly one state larger, with at most five letters and four when the readout is noncentral.  The binary cost remains within the corresponding one-state interval.  Thus commutativity is the exact full-mobility boundary between ``already classical'' invariant memory and memory that incurs a one-state stochastic overhead.  The familiar four-letter value \(N^2+1\) at trivial symmetry is the endpoint where the invariant algebra is the whole matrix algebra, rather than an isolated quadratic phenomenon.

Three filters explain both this dichotomy and the models in which it does not collapse to an equality.  Instance geometry identifies the exact continuation space of one automaton by quotienting prefix-reachable states by every direction invisible to suffix effects; the Hilbert--Schmidt pairing on the resulting spaces has rank equal to the real Hankel rank without controllability or observability assumptions.  Structural capacity then asks which coordinates a symmetry-constrained dynamics can write.  In the compact-group and open-system setting \cite{FultonHarris,MarvianSpekkens,BucaProsen,AlbertJiang,BaumgartnerNarnhofer,DeGrootEtAl,CirstoiuEtAl}, the center of the commutant stores isotypic populations frozen by commuting unitaries, while its traceless multiplicity blocks carry the reversible directions; charge-zero dissipation removes the unary spectral loss inside those blocks, and covariant mobility releases relative central populations subject to component conservation.  Operational realization is the independent final filter: finite strict-cutpoint shattering converts visible continuous directions into sign obstructions, while recurrent Markov-limit centering can force one further probabilistic state beyond what finite sign-rank alone detects.  Together, symmetry, dynamics, and readout do more than reduce dimension: symmetry fixes the invariant algebra, the dynamics determines which coordinates can move, and threshold geometry decides which movable directions become unavoidable probabilistic states.

The distinction matters away from full mobility.  For a prescribed reversible readout, the relevant visible dimension is the orbit of the accepting projector rather than the whole commutator space, giving the binary interval \(2+\kappa\leq\operatorname{SC}\leq\mathsf B+1\).  For component-conserving channels with structural capacity \(M\), strict-cutpoint cost lies in \([M,M+1]\), and a persistent noncommutative phase attains the upper endpoint over unrestricted alphabets.  Intermediate halting changes the structural input by promoting traces of active nonhalting corners, but the same interval survives; a scalar profile attains the lower endpoint and shows that structural dimension alone does not force the dynamic extra state.  These are not exceptions to the framework: they identify which of its three filters is binding.

The structural law has consequences that are visible before any individual automaton is constructed.  Schur--Weyl duality places two mutual commutants on the same tensor-power Hilbert space and makes the contrast extreme: preserving permutation symmetry gives a polynomial worst-case memory scale at fixed local dimension, whereas preserving the collective-unitary symmetry gives an exponential one.  The ambient quantum system is the same; the preserved symmetry changes the available memory scale.  Fixed-weight permutation modules sharpen this representation-theoretic picture further: at half filling their squared-sector-dimension sum is exactly Catalan, so the structural capacity is that Catalan count minus the central-sector correction, with a critical window describing how this scale is approached.  The accepting-projector orbit and its intrinsic Fisher rank similarly quantify the task-visible part of reversible capacity, while subgroup branching measures how symmetry release restores hidden directions.  Adjacent notions remain distinct: common-character Kraus operators form a proper boundary subclass of component-conserving channels, and group actions on nominal alphabets act on a different object altogether \cite{BojanczykKlinLasota}.

The resulting capability is a symmetry-aware memory calculus rather than a collection of model-specific simulation bounds.  Given a compact symmetry, a dynamical class, and a readout profile, the theory first computes the exact instance invariant, then the largest operator space the dynamics can activate, and finally the strict-cutpoint obstruction that converts part of that space into classical states.  In the fully mobile case the commutant dimension and its commutativity already decide the exact unrestricted state cost.

The scope is strict-cutpoint language recognition by real PFAs.  The simulator preserves the threshold language rather than numerical acceptance probabilities word by word; isolated-cutpoint, bounded-error, and hybrid quantum--classical models therefore require different invariants.  Qualitative unbounded-error equivalences \cite{YakaryilmazSay} and contextuality-based bounded-error promise separations \cite{PrakashContextuality} concern different comparison classes.  Within the present scope, reversible evolution, dissipative channels, and intermediate halting probe the same principle under progressively richer dynamics: symmetry specifies the operator algebra, the dynamics selects its movable part, and threshold geometry prices the resulting classical memory.

\section{Instance geometry: the exact reachable--observable invariant}

Before symmetry can be priced, one automaton already has an exact memory that can be much smaller than its state space.  Its acceptance function is bilinear in a prefix-generated state and a suffix-generated effect, so the relevant quantity is the rank of that interaction rather than the ambient operator dimension.  Compact word closure makes this rank an exact reachable--observable pairing and resolves the paired space into irreducible operator channels.

\subsection{Model and sector reduction}

\begin{definition}\label{def:model}
A sector-preserving measure-once one-way quantum finite automaton is a
tuple
\[
  \A=(\Hh,\rho_0,\{U_a\}_{a\in\Sigma},P,\tau),
\]
where \(\Sigma\) is finite,
\[
  \Hh=\bigoplus_{\alpha\in\Lambda}W_\alpha,
  \qquad
  U_a=\bigoplus_{\alpha\in\Lambda}U_{a,\alpha},
  \qquad
  P=\bigoplus_{\alpha\in\Lambda}P_\alpha,
\]
and write \(D_\alpha=\dim W_\alpha\).  Here \(\rho_0\) is a density
operator, each \(U_a\) is unitary, \(P\) is an orthogonal projector, and
\(\tau\in\R\) is a strict cutpoint.  The standard pure-state model is
the special case \(\rho_0=|\psi_0\rangle\langle\psi_0|\).
For \(w=a_1\cdots a_m\), set
\[
  U_w=U_{a_m}\cdots U_{a_1}
\]
and
\[
  f_\A(w)
  =
  \Tr(PU_w\rho_0U_w^\dagger).
\]
The recognized language is
\[
  L_{\A,\tau}
  =
  \{w\in\Sigma^*:f_\A(w)>\tau\}.
\]
\end{definition}

Let \(\mathcal D\) be dephasing across the fixed sectors and define
\[
  \bar\rho_0
  =
  \mathcal D(\rho_0).
\]
Since every \(U_w\) and \(P\) is block diagonal,
\[
  f_\A(w)=\Tr(PU_w\bar\rho_0U_w^\dagger).
\]
Cross-sector coherences in the initial state never enter the acceptance
function.  Allowing a mixed \(\rho_0\) is also necessary for the reduced
multiplicity-space description below, because Haar twirling a pure
physical state may produce mixed multiplicity blocks.  All extremal
unitary capacities proved later are nevertheless attained by pure
physical inputs.

For any automaton considered below and every prefix \(x\), let \(f_{\A,x}(y)=f_\A(xy)\) be its continuation behavior and define
\[
  \mathscr M_\A
  =
  \spanR\{f_{\A,x}:x\in\Sigma^*\}.
\]
We call \(\mathscr M_\A\) the linear behavioral-memory space of the instance: two prefix states are identified exactly when every suffix gives them the same acceptance value.  The acceptance Hankel matrix is
\[
  H_\A(x,y)=f_\A(xy),
  \qquad x,y\in\Sigma^*,
\]
and its real rank is denoted
\[
  \beta(\A)=\rank_\R H_\A.
\]
Thus \(\dim\mathscr M_\A=\beta(\A)\).  This linear memory is distinct from the operational strict-cutpoint memory \(\operatorname{sc}_{\mathrm{PFA},\mathbb R}(\A,\tau)\), the minimum number of states in a real PFA recognizing \(L_{\A,\tau}\).  Operational lower bounds arise only when a finite or recurrent construction turns this linear dimension into strict threshold distinctions.
For a cutpoint \(\tau\), we also use the centered Hankel rank
\[
  \beta_\tau(\A)
  =
  \rank_\R\bigl(f_\A(xy)-\tau\bigr)_{x,y\in\Sigma^*}.
\]
Plainly \(\beta_\tau(\A)\leq\beta(\A)+1\); it can be no larger than
\(\beta(\A)\) when the constant series is already contained in a
minimal behavior space.
By the standard Hankel theorem for real weighted automata,
\(\beta(\A)\) is the minimum dimension of a real linear
representation of \(f_\A\) \cite{BerstelReutenauer}.

The dephasing step identifies the operator space in which the acceptance
function lives.  It does not yet determine the effective dimension,
because the words may generate only a proper subset of the available
states and effects.  That dependence is captured by the compact closure
of the word dynamics.

\subsection{The reachable--observable pairing}

Let
\[
  \Gamma=\{U_w:w\in\Sigma^*\},
  \qquad
  G=\overline{\Gamma}.
\]

\begin{lemma}\label{lem:compact-closure}
The compact closure \(G\) is a subgroup of
\(\prod_\alpha U(D_\alpha)\).
\end{lemma}

\begin{proof}
It is a closed subsemigroup of a compact group.  For \(g\in G\), the
closure of \(\{g^n:n\geq1\}\) contains the identity.  Hence a
subsequence \(g^{n_j}\) converges to the identity, and
\(g^{n_j-1}\to g^{-1}\).  Closedness gives \(g^{-1}\in G\).
\end{proof}

Define the real orbit spans
\[
  \cR
  =
  \spanR\{g\bar\rho_0g^\dagger:g\in G\},
\]
\[
  \cE
  =
  \spanR\{g^\dagger Pg:g\in G\}.
\]

The forward--backward factorization of a Hankel matrix and its
minimality criterion are standard in weighted-automaton theory
\cite{KieferEtAl}.  The specialization needed here replaces the
word-generated spaces by compact orbit spans without assuming
controllability or observability.

\begin{theorem}[Exact reachable--observable invariant]\label{thm:exact-rank}
For every sector-preserving measure-once automaton, its linear behavioral-memory dimension satisfies
\[
  \dim\mathscr M_\A
  =
  \beta(\A)
  =
  \rank
  \left(
  T:\cR\longrightarrow\cE^*
  \right),
\]
where
\[
  T(X)(Y)=\Tr(YX).
\]
Equivalently,
\[
  \beta(\A)
  =
  \dim\cR-\dim(\cR\cap\cE^\perp)
  =
  \dim\cE-\dim(\cE\cap\cR^\perp).
\]
\end{theorem}

\begin{proof}
For each prefix \(x\) and suffix \(y\), write
\[
  \rho_x=U_x\bar\rho_0U_x^\dagger,
  \qquad
  E_y=U_y^\dagger P U_y.
\]
Since \(U_{xy}=U_yU_x\),
\[
  H_\A(x,y)=\Tr(E_y\rho_x).
\]
The row indexed by \(x\) is therefore \(T(\rho_x)\) restricted to the
word-generated effect set.  The word semigroup $\Gamma$ is dense in
$G$, and the conjugation maps
$g\mapsto g\bar\rho_0g^\dagger$ and $g\mapsto g^\dagger Pg$ are
continuous.  Thus the closures of the word-generated state and effect
sets are the corresponding $G$-orbits; finite-dimensional linear spans
are closed, so these sets span \(\cR\) and \(\cE\).  Hence
the row space of \(H_\A\) is exactly \(T(\cR)\).  The remaining
identities are rank--nullity.
\end{proof}

This theorem makes the instance-geometric filter exact without a controllability or observability assumption: behavioral memory is reachability only after quotienting out every direction annihilated by all suffix effects.  The pairing formula consequently separates the two ways in which symmetry can reduce memory---the dynamics may fail to reach an operator direction, or the suffix measurement orbit may fail to observe it---and the representation-theoretic resolution below makes that separation explicit inside each irreducible operator mode.

\subsection{Representation-theoretic resolution}

Let
\[
  \cB=\bigoplus_{\alpha\in\Lambda}\Herm(W_\alpha)
\]
with the conjugation action
\[
  \alpha(g)X=gXg^\dagger.
\]
Complexify this real unitary representation and decompose it as
\[
  \cB_\C
  \cong
  \bigoplus_{\mu\in\widehat G_{\cB}}
  V_\mu\otimes\mathcal N_\mu,
\]
where \(V_\mu\) is irreducible,
\(d_\mu=\dim_\C V_\mu\), and \(\mathcal N_\mu\) is its multiplicity space.
Let \(\rho_\mu\) and \(P_\mu\) denote the corresponding components of
\(\bar\rho_0\) and \(P\).  Define
\[
  C_\mu
  =
  \Tr_{\mathcal N_\mu}
  \bigl(
  |\rho_\mu\rangle\langle P_\mu|
  \bigr)
  \in\operatorname{End}(V_\mu).
\]

\begin{theorem}[Isotypic rank formula]\label{thm:isotypic}
The exact behavior rank is
\[
  \beta(\A)
  =
  \sum_{\mu\in\widehat G_{\cB}}
  d_\mu\,\rank C_\mu.
\]
\end{theorem}

\begin{proof}
Choose an orthonormal basis \(\{m_u\}\) of \(\mathcal N_\mu\) and write
\[
  \rho_\mu=\sum_u r_{\mu,u}\otimes m_u,
  \qquad
  P_\mu=\sum_u p_{\mu,u}\otimes m_u.
\]
The partial-trace convention in the definition gives
\[
  C_\mu=\sum_u r_{\mu,u}p_{\mu,u}^\dagger.
\]
The contribution of this isotypic component to the acceptance
function is
\[
  f_\mu(g)
  =
  \sum_u
  p_{\mu,u}^\dagger
  \pi_\mu(g)
  r_{\mu,u}
  =
  \Tr\bigl(\pi_\mu(g)C_\mu\bigr),
\]
where $\pi_\mu$ acts by left multiplication on $V_\mu$.

With the right-translation convention $R_hf(g)=f(gh)$, the coefficient
matrix is replaced by \(\pi_\mu(h)C_\mu\).  By irreducibility and
Burnside's theorem,
\[
  \spanC\{\pi_\mu(h):h\in G\}
  =
  \operatorname{End}(V_\mu).
\]
Therefore the translated coefficient space is
\[
  \{AC_\mu:A\in\operatorname{End}(V_\mu)\},
\]
whose complex dimension is \(d_\mu\rank C_\mu\).
Matrix-coefficient spaces of inequivalent irreducible
representations are linearly independent by the Peter--Weyl theorem.
Summing over \(\mu\) gives the complexified Hankel rank.  Complexifying
a real matrix preserves its largest nonzero minor and therefore its
rank, so this is the original real Hankel rank.
\end{proof}

The rank of \(C_\mu\) is the dimension of the effective irreducible-direction space that remains after the excitation and observation tensors are contracted over multiplicity coordinates.  It need not equal the number of nonzero multiplicity components.

The isotypic formula is an exact statement about the acceptance
function.  To translate it into a probabilistic state bound, we use a
codimension-one stochastic embedding of a real linear representation.

\subsection{From Hankel rank to probabilistic states}

A real probabilistic finite automaton is allowed real stochastic
transition probabilities.  More generally, for an ordered subfield
\(\mathbb F\subseteq\R\), an \(\mathbb F\)-PFA has all stochastic data in \(\mathbb F\).

We write a \(k\)-dimensional real linear representation as
\((u,\{A_a\}_{a\in\Sigma},v)\), with row vector \(u\), column vector
\(v\), and series value \(uA_wv\).  The following embedding improves
the general quantitative conversion needed here and, unlike a generic
positive--negative splitting, preserves the alphabet.

\begin{theorem}[Codimension-one stochastic embedding]
\label{thm:codimension-one-pfa}
Let \(\mathbb F\subseteq\R\) be an ordered subfield.
Let \(k\geq1\).  Every \(k\)-dimensional \(\mathbb F\)-linear representation at
cutpoint zero has the same strict-cutpoint language as an
alphabet-preserving \((k+1)\)-state \(\mathbb F\)-PFA.  Consequently,
every cutpoint \(\tau\in\mathbb F\) costs at most \(k+2\)
probabilistic states.

The latter bound improves to \(k+1\) whenever the same
\(k\)-dimensional representation contains a normalized constant mode:
either there is a column \(t\) with
\[
  A_at=t\quad(a\in\Sigma),\qquad ut=1,
\]
or there is a row \(q\) with
\[
  qA_a=q\quad(a\in\Sigma),\qquad qv=1.
\]
\end{theorem}

\begin{proof}
First consider cutpoint zero.  Put \(N=k+1\), let
\(\mathbf 1\in\mathbb F^N\) be the all-ones column, and set
\(J=\mathbf 1\mathbf 1^{\mathsf T}\).  If \(v=0\), the represented
language is empty and needs only one probabilistic state, so suppose
\(v\ne0\).  Choose an isomorphism
\[
  L:\mathbb F^k\longrightarrow \mathbf 1^\perp
  \quad\text{such that}\quad
  Lv=e_1-\frac1N\mathbf 1,
\]
and define
\[
  R=L^{-1}\!\left(I-\frac1N J\right).
\]
Then
\[
  RL=I_k,\qquad LR=I-\frac1N J,\qquad
  R\mathbf 1=0,\qquad \mathbf 1^{\mathsf T}L=0,\qquad Re_1=v.
\]
For a \(k\times k\) matrix \(M\), write \(B(M)=LMR\).  This map is
multiplicative and every \(B(M)\) has zero row and column sums:
\[
  B(M)B(M')=B(MM'),\qquad
  B(M)\mathbf 1=0,\qquad \mathbf 1^{\mathsf T}B(M)=0.
\]
Because the alphabet is finite, one may choose
\(\varepsilon\in\mathbb F\), \(\varepsilon>0\), sufficiently small that
\[
  P_a=\frac1N J+\varepsilon B(A_a)
\]
is entrywise nonnegative for every letter.  Each \(P_a\) is therefore
doubly stochastic.  Likewise, for sufficiently small
\(\delta\in\mathbb F\), \(\delta>0\),
\[
  \pi=\frac1N\mathbf 1^{\mathsf T}+\delta uR
\]
is a probability row vector.  Take state \(1\) as the sole accepting
state.  Since \(JB(M)=B(M)J=0\), multiplicativity gives
\[
  P_w=\frac1N J+\varepsilon^{|w|}B(A_w)
\]
for every word, including the empty word because \(B(I_k)=I-J/N\).
Hence
\[
  \pi P_we_1
  =\frac1N+\delta\varepsilon^{|w|}uA_wv.
\]
Thus comparison with the PFA cutpoint \(1/N\) reproduces exactly the
zero-cutpoint language.

For a general cutpoint \(\tau\in\mathbb F\), augment the representation
by the constant coordinate
\[
  u'=(u,-\tau),\qquad A'_a=A_a\oplus[1],\qquad
  v'=\binom{v}{1}.
\]
It has dimension \(k+1\) and zero-cutpoint value \(uA_wv-\tau\), so
the first part gives \(k+2\) probabilistic states.  If a normalized
right constant mode exists, replace \(v\) by \(v-\tau t\); if a
normalized left constant mode exists, replace \(u\) by \(u-\tau q\).
Either operation centers the cutpoint inside the original
\(k\)-dimensional representation, and the first part then gives
\(k+1\) states.
\end{proof}

The zero-cutpoint mechanism is the codimension-one construction of
\cite{ChenWuQuadratic}.  We record it here over an ordered field and
isolate the normalized constant-mode centering that lets the
symmetry-reduced and accumulator representations below retain the same
one-state overhead.  It should still be distinguished from the older
qualitative conversion of Turakainen \cite{Turakainen}.

\begin{corollary}\label{cor:pfa-upper}
Every sector-preserving automaton satisfies
\[
  \operatorname{sc}_{\mathrm{PFA},\R}(\A,\tau)
  \leq
  \beta_\tau(\A)+1
  \leq
  \beta(\A)+2.
\]
If a minimal \(\beta(\A)\)-dimensional representation of \(f_\A\)
contains a normalized left or right constant mode, then
\[
  \beta_\tau(\A)\leq\beta(\A),
  \qquad
  \operatorname{sc}_{\mathrm{PFA},\R}(\A,\tau)
  \leq\beta(\A)+1.
\]
If the representation and cutpoint lie in an ordered subfield
\(\mathbb F\), the corresponding statement holds for
\(\mathbb F\)-PFAs.
\end{corollary}

\begin{proof}
The Hankel theorem applied to the centered series supplies a real GFA
of dimension $\beta_\tau(\A)$ at cutpoint zero.  If this rank is zero,
the strict language is empty and one state suffices.  Otherwise apply
\cref{thm:codimension-one-pfa}.  The constant-mode clause gives the
stated refinement when one starts from a minimal representation of
$f_\A$.  The construction stays inside $\mathbb F$ when the data and
cutpoint do.
\end{proof}

The PFA reproduces the threshold language.  It is not asserted to
reproduce the numerical acceptance probability word by word.

This completes the instance-geometric layer: \(\beta(\A)\) is exact, while the stochastic embedding turns it only into an upper bridge for strict-cutpoint memory.  The next layer asks for the largest pairing rank permitted by symmetry and dynamics before a particular automaton is fixed.

\section{Structural capacity: the center--commutator split}

An exact invariant for one automaton does not yet say how much memory a symmetric dynamical class can support.  The structural question is which invariant operator coordinates the dynamics can ever write.  For compact symmetry the answer is encoded by the commutant: Haar reduction exposes the multiplicity blocks, while the center--commutator split separates isotypic populations that reversible conjugation leaves read-only from traceless coordinates that it can move.

\subsection{Haar reduction to multiplicity spaces}

Let $K$ be a compact group with a finite-dimensional unitary
representation $\pi:K\to U(\Hh)$.  A measure-once automaton is
$K$-equivariant when
\[
  [U_a,\pi(k)]=0,
  \qquad
  [P,\pi(k)]=0
\]
for every $a\in\Sigma$ and $k\in K$.  Choose an isotypic decomposition
\[
  \Hh
  \cong
  \bigoplus_{\lambda\in\widehat K}
  \mathcal M_\lambda\otimes V_\lambda,
\]
where the $V_\lambda$ are inequivalent irreducible $K$-modules and
$\mathcal M_\lambda$ are their multiplicity spaces.  Write
\[
  m_\lambda=\dim \mathcal M_\lambda,
  \qquad
  d_\lambda=\dim V_\lambda.
\]
In this decomposition,
\[
  \pi(k)
  =
  \bigoplus_\lambda
  I_{\mathcal M_\lambda}\otimes\pi_\lambda(k).
\]

Define the Haar twirl
\[
  \mathcal T_K(X)
  =
  \int_K\pi(k)X\pi(k)^\dagger\,dk.
\]

\begin{proposition}[Standard multiplicity-space reduction]\label{prop:twirl-reduction}
Let $\A$ be a $K$-equivariant measure-once automaton with initial state
$\rho_0$.  Set
\(\widetilde\rho_0=\mathcal T_K(\rho_0)\).  Then for every word $w$,
\[
  \Tr(PU_w\rho_0U_w^\dagger)
  =
  \Tr\!\left(PU_w\widetilde\rho_0U_w^\dagger\right).
\]
Moreover, there exist positive semidefinite
$\sigma_\lambda\in\Herm(\mathcal M_\lambda)$, unitaries
$W_{a,\lambda}\in U(\mathcal M_\lambda)$, and orthogonal projectors
$Q_\lambda$ on $\mathcal M_\lambda$ such that
\[
  \widetilde\rho_0
  =
  \bigoplus_\lambda
  \sigma_\lambda\otimes\frac{I_{V_\lambda}}{d_\lambda},
\]
\[
  U_a
  =
  \bigoplus_\lambda
  W_{a,\lambda}\otimes I_{V_\lambda},
  \qquad
  P
  =
  \bigoplus_\lambda
  Q_\lambda\otimes I_{V_\lambda},
\]
and hence
\[
  f_\A(w)
  =
  \sum_\lambda
  \Tr\!\left(
    Q_\lambda W_{w,\lambda}
    \sigma_\lambda W_{w,\lambda}^\dagger
  \right).
\]
\end{proposition}

\begin{proof}
Since $U_w$ and $P$ commute with $\pi(k)$,
\begin{align*}
  \Tr\!\left(PU_w\widetilde\rho_0U_w^\dagger\right)
  &=
  \int_K
  \Tr\!\left(
    PU_w\pi(k)\rho_0\pi(k)^\dagger U_w^\dagger
  \right)\,dk\\
  &=
  \int_K
  \Tr\!\left(
    \pi(k)^\dagger PU_w\pi(k)\rho_0U_w^\dagger
  \right)\,dk\\
  &=
  \Tr(PU_w\rho_0U_w^\dagger).
\end{align*}
Schur's lemma gives the displayed forms of the twirled state and every
operator in the commutant.  Taking the trace over each irreducible
factor $V_\lambda$ yields the reduced acceptance formula.
\end{proof}

The reduction itself is standard compact-group representation theory \cite{FultonHarris}.  Its role here is to identify the operator space on which the automaton acts; the new step is to connect the Hermitian noncommutative part of that space to real Hankel rank and finite-alphabet probabilistic state complexity.

Thus every $K$-equivariant automaton is behaviorally identical to a
sector-preserving automaton on the multiplicity-space direct sum
$\bigoplus_\lambda \mathcal M_\lambda$.  The irreducible dimensions
$d_\lambda$ affect how the physical Hilbert space is assembled, but the
word-dependent acceptance behavior is carried by the multiplicities
$m_\lambda$.

\subsection{The noncommutative commutant law}

Let
\[
  \mathcal C_K
  =
  \operatorname{End}_K(\Hh)
  =
  \{X\in\operatorname{End}(\Hh):
      X\pi(k)=\pi(k)X\ \text{for all }k\in K\}.
\]
For a finite-dimensional algebra $\mathcal C$, write
\[
  [\mathcal C,\mathcal C]
  =
  \spanC\{XY-YX:X,Y\in\mathcal C\}
\]
for the linear commutator space, and define
\[
  \nu_K
  =
  \dim_\C[\mathcal C_K,\mathcal C_K].
\]
The Wedderburn decomposition supplied by the isotypic representation is
\[
  \mathcal C_K
  \cong
  \bigoplus_\lambda \Mat_{m_\lambda}(\C),
\]
so
\[
  [\mathcal C_K,\mathcal C_K]
  \cong
  \bigoplus_\lambda\mathfrak{sl}(m_\lambda,\C)
\]
and
\[
  \nu_K
  =
  \sum_\lambda(m_\lambda^2-1)
  =
  \dim_\C\mathcal C_K-
  \dim_\C Z(\mathcal C_K).
\]

\begin{lemma}[Hermitian bridge]\label{lem:hermitian-bridge}
For the finite-dimensional $C^*$-algebra $\mathcal C_K$,
\[
  [\mathcal C_K,\mathcal C_K]
  \cong
  \bigoplus_\lambda\mathfrak{sl}(m_\lambda,\C),
\]
and
\[
  [\mathcal C_K,\mathcal C_K]\cap\Herm(\Hh)
  \cong
  \bigoplus_\lambda\Herm_0(\mathcal M_\lambda).
\]
Consequently,
\[
  \dim_\R\bigl([\mathcal C_K,\mathcal C_K]\cap\Herm(\Hh)\bigr)
  =
  \dim_\C[\mathcal C_K,\mathcal C_K].
\]
\end{lemma}

\begin{proof}
In each block, the algebraic commutator space is the traceless complex
matrix algebra.  Every traceless complex matrix has a unique
decomposition $X=A+iB$ with $A,B$ traceless Hermitian.  Thus
$\mathfrak{sl}(m,\C)$ is the complexification of $\Herm_0(m)$, and the
real dimension of its Hermitian part equals its complex dimension.
Taking direct sums proves the claim.
\end{proof}

\begin{theorem}[Noncommutative commutant law]\label{thm:commutant-law}
Every $K$-equivariant measure-once automaton satisfies
\[
  \beta(\A)\leq 1+\nu_K.
\]
If, in every block with $m_\lambda\geq2$, the reduced initial state
$\sigma_\lambda$ has a nonzero traceless component, the accepting
projector $Q_\lambda$ is nontrivial, the central state and effect have
nonzero pairing,
\[
  \sum_\lambda
  \frac{\Tr(\sigma_\lambda)\rank(Q_\lambda)}{m_\lambda}
  >0,
\]
and the compact word-closure
contains
\[
  \prod_{\lambda:m_\lambda\geq2}SU(m_\lambda)
\]
acting independently on the multiplicity spaces, then
\[
  \beta(\A)=1+\nu_K.
\]
Consequently,
\[
  \max_{\A\ \text{$K$-equivariant}}\beta(\A)
  =
  1+\dim_\C[\mathcal C_K,\mathcal C_K].
\]
\end{theorem}

\begin{proof}
By \cref{prop:twirl-reduction}, each block trace
$\Tr\sigma_\lambda$ is invariant under every word.  All central
components therefore contribute to the acceptance function through one
word-independent scalar sequence.  Every word-dependent direction lies
in
\[
  \bigoplus_\lambda\Herm_0(\mathcal M_\lambda),
\]
whose real dimension is $\nu_K$ by \cref{lem:hermitian-bridge}.  \Cref{thm:exact-rank} gives the upper bound.

Under independent $SU(m_\lambda)$ control, the conjugacy orbit of every
nonzero traceless Hermitian operator spans
$\Herm_0(\mathcal M_\lambda)$.  Applying this fact to the traceless parts of
$\sigma_\lambda$ and $Q_\lambda$ shows that the reachable and observable
spans project onto all traceless blocks.  The central-overlap condition
supplies one additional paired constant direction.  The
Hilbert--Schmidt pairing therefore has rank $1+\nu_K$.

For the maximum, choose nonzero traceless state components and
nontrivial accepting projectors in every block with $m_\lambda\geq2$,
and choose a nonzero central pairing.  \Cref{thm:binary-dense} supplies
the required independent compact closure with two letters.  If
$\nu_K=0$, taking $P=I$ gives the constant function one and rank one.
Thus the upper bound is attained in every case.
\end{proof}

An isotypic block with \(m_\lambda=1\) has no traceless multiplicity
direction.  It therefore contributes zero both to \(\nu_K\) and to the
prepare--test/readout-orbit summands, and it is omitted from the
nontrivial \(SU(m_\lambda)\) control factors without any exceptional
case in the formulas.

The theorem identifies the reversible part of symmetric memory by what the dynamics can write.  The center of $\mathcal C_K$ stores isotypic populations that the accepting projector may read but commuting unitaries cannot change; the linear commutator space is exactly the direct sum of traceless multiplicity blocks and contains every word-dependent direction.  Within this reversible class, symmetry does not merely reduce the operator space: the central coordinates are read-only, while the writable part has dimension $\dim_\C[\mathcal C_K,\mathcal C_K]$.

After symmetry reduction, the multiplicity spaces form an ordinary
sector profile.  The exact Hankel rank still depends on the transition
group, initial state, and accepting projector.  The total operator
capacity allowed by these reduced sectors is controlled by traceless
Hermitian blocks, whereas the part visible to the chosen readout is
controlled by the orbit of the accepting projector.

\subsection{The universal sector cap}

For each sector set
\[
  p_\alpha=\Tr\bar\rho_{0,\alpha},
  \qquad
  r_\alpha=\rank P_\alpha,
  \qquad
  q_\alpha=D_\alpha-r_\alpha,
\]
and define the traceless parts
\[
  \rho_\alpha^\circ
  =
  \bar\rho_{0,\alpha}
  -\frac{p_\alpha}{D_\alpha}I_\alpha,
  \qquad
  P_\alpha^\circ
  =
  P_\alpha-\frac{r_\alpha}{D_\alpha}I_\alpha.
\]
A sector is jointly active when
\[
  \rho_\alpha^\circ\neq0
  \qquad\text{and}\qquad
  P_\alpha^\circ\neq0.
\]
Write \(\Lambda_{\mathrm{act}}\) for the jointly active set and put
\[
  \cV_0
  =
  \bigoplus_{\alpha\in\Lambda_{\mathrm{act}}}
  \Herm_0(W_\alpha).
\]
This definition remains valid after Haar reduction, when a positive-
weight multiplicity block can be maximally mixed and hence have no
traceless excitation.

Define the central parts
\[
  c_\rho
  =
  \bigoplus_{\alpha\in\Lambda}
  \frac{p_\alpha}{D_\alpha}I_\alpha,
  \qquad
  c_P
  =
  \bigoplus_{\alpha\in\Lambda}
  \frac{r_\alpha}{D_\alpha}I_\alpha,
\]
and let
\[
  \chi=\Tr(c_\rho c_P),
  \qquad
  \delta=\begin{cases}1,&\chi>0,\\0,&\chi=0.\end{cases}
\]
Let \(\Pi_0\) denote the Hilbert--Schmidt projection from the direct
sum of all traceless sector spaces onto \(\cV_0\).  If
$X=\sum_i a_i g_i\bar\rho_0g_i^\dagger$, then
$\Tr X=\sum_i a_i$ and its $\alpha$-block trace is
$(\Tr X)p_\alpha$.  Hence every \(X\in\cR\) satisfies
\[
  X=(\Tr X)c_\rho+X^\circ
\]
with blockwise traceless \(X^\circ\).  Likewise, if \(P\neq0\), every
\(Y\in\cE\) has the form
\[
  Y=s(Y)c_P+Y^\circ,
  \qquad
  s(Y)=\frac{\Tr Y}{\Tr P},
\]
with blockwise traceless \(Y^\circ\); set \(s(Y)=0\) when \(P=0\).
Define the behavior-relevant images
\[
  \widehat\cR
  =
  \left\{
    \delta(\Tr X)c_\rho+\Pi_0X^\circ:X\in\cR
  \right\},
\]
\[
  \widehat\cE
  =
  \left\{
    \delta s(Y)c_P+\Pi_0Y^\circ:Y\in\cE
  \right\}.
\]

\begin{theorem}[Universal sector cap and exact saturation criterion]
\label{thm:sector-cap}
For every sector-preserving automaton,
\[
  \beta(\A)
  =
  \rank
  \left(
    \langle\cdot,\cdot\rangle_{\mathrm{HS}}
    \big|_{\widehat\cR\times\widehat\cE}
  \right)
  \leq
  \delta+
  \sum_{\alpha\in\Lambda_{\mathrm{act}}}(D_\alpha^2-1).
\]
Equality holds if and only if
\[
  \widehat\cR
  =
  \delta\operatorname{span}_{\R}\{c_\rho\}
  \oplus\cV_0
\]
and
\[
  \widehat\cE
  =
  \delta\operatorname{span}_{\R}\{c_P\}
  \oplus\cV_0.
\]
\end{theorem}

\begin{proof}
Central and traceless sector operators are orthogonal.  In every sector
outside \(\Lambda_{\mathrm{act}}\), either the state orbit has no
traceless component or the effect orbit has none.  Moreover, the
central contribution to the pairing is
\[
  (\Tr X)s(Y)\Tr(c_\rho c_P),
\]
which vanishes exactly when \(\delta=0\).  Hence
\[
  \Tr(YX)
  =
  \left\langle
    \delta(\Tr X)c_\rho+\Pi_0X^\circ,
    \delta s(Y)c_P+\Pi_0Y^\circ
  \right\rangle_{\mathrm{HS}}
\]
for all \(X\in\cR\) and \(Y\in\cE\).  The exact pairing formula in
\cref{thm:exact-rank} proves the first equality.

The ambient pairing between
\[
  \delta\operatorname{span}_{\R}\{c_\rho\}\oplus\cV_0
  \quad\text{and}\quad
  \delta\operatorname{span}_{\R}\{c_P\}\oplus\cV_0
\]
is nondegenerate: when \(\delta=1\), its central coefficient is
\(\chi>0\), and on \(\cV_0\) it is the Hilbert--Schmidt inner
product.  Both ambient spaces have dimension
\[
  \delta+
  \sum_{\alpha\in\Lambda_{\mathrm{act}}}(D_\alpha^2-1).
\]
This gives the upper bound.  A restriction of a nondegenerate pairing
has full ambient rank if and only if both of its argument spaces equal
their respective ambient spaces, which proves the saturation
criterion.
\end{proof}

For a prescribed nonempty active profile, write
\[
  \mathsf B(\mathbf D)
  =
  1+\sum_{\alpha\in\Lambda_{\mathrm{act}}}(D_\alpha^2-1).
\]
Joint activity implies \(\delta=1\), so this is the corresponding
profile cap.  If the active set is empty, the exact instance cap is
instead \(\delta\in\{0,1\}\).

\begin{corollary}\label{cor:full-sector-control}
If
\[
  \prod_{\alpha\in\Lambda_{\mathrm{act}}} SU(D_\alpha)
  \subseteq G
\]
acts independently on the active sectors, then
\[
  \beta(\A)
  =
  \delta+
  \sum_{\alpha\in\Lambda_{\mathrm{act}}}(D_\alpha^2-1).
\]
In particular, a nonempty active profile attains
\(\mathsf B(\mathbf D)\).
\end{corollary}

\begin{proof}
The conjugacy orbit of every nonzero traceless Hermitian operator under
\(SU(D)\) spans \(\Herm_0(D)\).  Indeed, its span is a nonzero invariant
subspace of the adjoint representation, and \(\mathfrak{su}(D)\) is
simple.  Independent sector control isolates the active blocks one at a
time.  Haar averaging over their product supplies the central coordinate
when \(\delta=1\); the average belongs to the orbit span because it lies
in its closure and finite-dimensional linear subspaces are closed.  Thus
both behavior-relevant images in
\cref{thm:sector-cap} fill their ambient spaces.  The empty-active-set
case reduces directly to the central pairing of rank \(\delta\).
\end{proof}

Full independent control is not necessary.

\begin{example}[Saturation with locked sectors]
Take two two-dimensional sectors and let one copy of \(SU(2)\) act
diagonally on both.  On the two traceless Bloch spaces, the operator
representation is the three-dimensional adjoint representation with
multiplicity two.  Choose the state components along Bloch directions
\(z\) and \(x\), and choose the accepting-projector components along
the same two directions.  The isotypic contraction is
\[
  C=zz^\dagger+xx^\dagger,
\]
which has rank two.  The isotypic formula gives
\[
  \beta=1+3\cdot2=7.
\]
This equals
\[
  1+2(2^2-1)=7,
\]
although the control group is the diagonal \(SU(2)\), not
\(SU(2)\times SU(2)\).
\end{example}

The quantity \(\mathsf B(\mathbf D)\) measures the largest blockwise behavior
space compatible with the fixed sector traces.  Saturation requires
both reachability and observability, which is why locked controls can
still attain the cap when their multiplicity channels span the relevant
operator modes.  Lower bounds for probabilistic automata, however, must
ultimately be read through the accepting measurement.  This leads to a
second, generally smaller capacity.

\section{Operational realization: finite and dynamic shattering}

A movable operator coordinate is a continuous dimension, whereas a probabilistic automaton counts states.  The operational question is when one forces the other.  Finite Jacobian shattering turns visible coordinates into strict signs, the codimension-one embedding turns linear memory into stochastic states, and recurrent Markov-limit centering can force one additional state that finite sign-rank does not expose.

\subsection{Differential sign witnesses}

For \(d\geq1\) and \(\sigma\in\{\pm1\}\), let \(H_d^\sigma\) be the
\((d+1)\times2^d\) sign matrix whose column indexed by
\(\eta\in\{\pm1\}^d\) is
\[
  (\eta_1,\ldots,\eta_d,\sigma)^{\mathsf T}.
\]

\begin{lemma}[Complete and affine complete-sign matrices]
\label{lem:affine-complete-sign}
The complete sign matrix \(C_d(j,\eta)=\eta_j\) has sign-rank \(d\),
whereas
\[
  \operatorname{signrank}(H_d^\sigma)=d+1.
\]
\end{lemma}

\begin{proof}
The displayed matrices give the upper bounds.  For \(C_d\), a
nonzero vector in the left kernel of a putative rank-below-\(d\)
sign realization is contradicted by choosing the column whose signs
agree with its nonzero coordinates.

For \(H_d^\sigma\), suppose a sign-equivalent real matrix \(G\) had
rank at most \(d\), and choose a nonzero vector \(c=(u,c_*)\) in its
left kernel.  If \(c_*\sigma>0\), choose
\(\eta_i=\operatorname{sign}(u_i)\) on nonzero coordinates; every
nonzero term in \(c^{\mathsf T}G_{\cdot,\eta}\) is positive.  If
\(c_*\sigma<0\), choose the opposite signs and every nonzero term is
negative.  When \(c_*=0\), use the first choice.  Each case
contradicts \(c^{\mathsf T}G=0\).
\end{proof}

\begin{lemma}[PFA cutpoint rank]\label{lem:pfa-cutpoint-rank}
Let $F$ be a finite prefix--suffix matrix of an $s$-state PFA at
cutpoint $\tau$:
\[
  F(x,y)=f(xy)-\tau.
\]
Then $\rank F\leq s$.  Consequently, every finite strict-cutpoint sign
matrix realized by the PFA has sign-rank at most $s$.
\end{lemma}

\begin{proof}
Let $\pi$ be the initial row distribution, let $P_w$ be the row-stochastic
transition matrix of $w$, and let $e$ be the accepting column.  Since
$\pi P_x\mathbf 1=1$,
\[
  f(xy)-\tau
  =
  (\pi P_x)\bigl(P_ye-\tau\mathbf 1\bigr).
\]
This factors every finite centered prefix--suffix matrix through
$\mathbb R^s$.  If a rejecting entry equals the cutpoint, increase the
cutpoint on the finite matrix by less than its smallest positive
margin; all accepting signs remain positive and every other sign
becomes negative.  The same factorization applies at the shifted
cutpoint, proving the sign-rank claim.
\end{proof}

\begin{lemma}[Dynamic affine lift]
\label{lem:dynamic-affine-lift}
Let a strict-cutpoint language be witnessed by prefixes
\(x_1,\ldots,x_d,x_*\), suffixes
\(y_\eta\) indexed by \(\eta\in\{\pm1\}^d\), a carrier prefix
\(z\), and a letter \(a\).  Suppose
\begin{align*}
 f(x_i y_\eta)&>\tau &&\text{if }\eta_i=+1,\\
 f(x_i y_\eta)&\leq\tau &&\text{if }\eta_i=-1,\\
 f(x_*y_\eta)&>\tau &&\text{for every }\eta,
\end{align*}
and suppose that, for every \(q\geq1\) and every \(\eta\), each of
the two sets
\[
 \{n:f(za^{qn}y_\eta)>\tau\},\qquad
 \{n:f(za^{qn}y_\eta)\leq\tau\}
\]
is infinite.  Then every real PFA recognizing the language has at
least \(d+2\) states.
\end{lemma}

\begin{proof}
Let an \(s\)-state PFA recognize the language, let \(S\) be its
transition matrix for \(a\), and put
\(c_\eta=P_{y_\eta}e-\tau\mathbf1\).  There is an integer \(q\geq1\),
divisible by all periods of the recurrent classes of \(S\), such that
\[
 \pi P_zS^{qn}\longrightarrow\omega
\]
for a probability row \(\omega\).  The two-sided recurrence assumption
and convergence force \(\omega c_\eta=0\) for every \(\eta\).

Set \(r_i=\pi P_{x_i}-\omega\) and
\(r_*=\pi P_{x_*}-\omega\).  These \(d+1\) rows lie in the zero-sum
hyperplane of \(\mathbb R^s\).  They are linearly independent.  Indeed,
if \(\sum_i u_ir_i+u_*r_*=0\) and \(u_*\ne0\), choose
\(\eta_i=\operatorname{sign}(u_iu_*)\) whenever \(u_i\ne0\).
After pairing with \(c_\eta\), every term has the weak sign of
\(u_*\), while the anchor term is strict, a contradiction.  If
\(u_*=0\), use \(\eta_i=\operatorname{sign}(u_i)\) when some
\(u_i>0\), and the opposite choice otherwise.  Again all terms have
one weak sign and at least one is strict.  Thus
\(d+1\leq s-1\), proving \(s\geq d+2\).
\end{proof}

The extra state in this lemma is not finite sign-rank: it is created by
centering the prefix distributions at a Markov limit.  The one-sided
positive anchor is essential; a negative anchor could lie exactly at
the cutpoint after centering.

\begin{theorem}[Tangent witness criterion]\label{thm:tangent-witness}
Suppose there are \(g_1,\ldots,g_m\in G\) and
\(X_1,\ldots,X_m\in\mathfrak g\) such that, with
\[
  \rho_j=g_j\bar\rho_0g_j^\dagger,
\]
one has
\[
  \Tr(P\rho_j)=\tau
\]
and the Jacobian
\[
  J_{j\ell}
  =
  \Tr\bigl(\rho_j[P,X_\ell]\bigr)
\]
is nonsingular.  Then finite prefix and suffix sets realize an affine
complete-sign matrix \(H_m^\sigma\) for some
\(\sigma\in\{\pm1\}\).
Every real PFA whose strict-cutpoint decisions agree with those of \(\A\) on these concatenations has
at least \(m+1\) states.
\end{theorem}

\begin{proof}
Define
\[
  F_j(\theta)
  =
  \Tr
  \left[
  P
  e^{\sum_\ell\theta_\ell X_\ell}
  \rho_j
  e^{-\sum_\ell\theta_\ell X_\ell}
  \right]
  -\tau.
\]
Then \(F(0)=0\) and \(DF(0)=J\).  The inverse function theorem gives
a neighborhood of the origin in the image of \(F\).  Hence for every
\(\eta\in\{\pm1\}^m\), a sufficiently small positive multiple of
\(\eta\) is realized by a group element \(h_\eta\).  The target
multiple may be chosen uniformly small, so all \(h_\eta\) may be
taken arbitrarily close to the identity.

Because \(J\) is nonsingular, at least one derivative
\(J_{j\ell}\) is nonzero.  Moving \(\rho_j\) a sufficiently small
distance along the corresponding reachable group direction gives a
reachable anchor state \(\rho_*\) with
\(\Tr(P\rho_*)-\tau\) strictly nonzero.  By taking the
\(h_\eta\)'s still closer to the identity, continuity makes the sign
of
\[
  \Tr(P h_\eta\rho_*h_\eta^\dagger)-\tau
\]
independent of \(\eta\); call it \(\sigma\).  The \(m\) cutpoint rows
together with this anchor row therefore form \(H_m^\sigma\).
Density of the word semigroup and continuity allow all finitely many
tests, preparations, and the anchor to be approximated by words
without changing any sign.

By \cref{lem:pfa-cutpoint-rank}, the finite sign matrix has sign-rank at
most $s$.  Since \cref{lem:affine-complete-sign} gives sign-rank
$m+1$, one has \(s\geq m+1\).
\end{proof}

This criterion is independent of the particular sector decomposition:
any control group that provides a full-rank cutpoint Jacobian yields an
\((m+1)\)-state obstruction.  To obtain uniform profile-level bounds, we
must realize the relevant group motions with a fixed finite alphabet.

\subsection{Two generators for independent sector control}

The first task is to generate the block dynamics themselves.  We use
frequency separation to isolate every adjacent matrix edge in every
sector, and then invoke local dense generation in the resulting compact
semisimple group.  Two infinitesimal generators suffice for arbitrary
active sector dimensions.

Assign globally distinct nonnegative integers \(m_{\alpha,j}\) to all
basis vectors and define
\[
  H_0
  =
  \bigoplus_\alpha
  \left[
  \sum_{j=1}^{D_\alpha}
  2^{m_{\alpha,j}}|j\rangle\langle j|
  -
  \frac{\sum_j2^{m_{\alpha,j}}}{D_\alpha}I_\alpha
  \right],
\]
\[
  H_1
  =
  \bigoplus_\alpha
  \sum_{j=1}^{D_\alpha-1}
  \left(
  |j\rangle\langle j+1|
  +
  |j+1\rangle\langle j|
  \right).
\]

\begin{lemma}[Binary Lie generation]\label{lem:binary-lie}
The two anti-Hermitian matrices \(iH_0,iH_1\) generate
\[
  \Lie(iH_0,iH_1)
  =
  \bigoplus_\alpha\mathfrak{su}(D_\alpha).
\]
\end{lemma}

\begin{proof}
See \cref{app:binary-lie-proof}.
\end{proof}

We use the standard conventions \(SU(1)=\{1\}\) and
\(\mathfrak{su}(1)=\{0\}\); one-dimensional sectors therefore require
no control and contribute no Lie-algebra directions.

A theorem of Breuillard and Gelander provides an identity
neighborhood in every connected semisimple real Lie group such that
near-identity elements generate a dense subgroup whenever their
logarithms generate the Lie algebra \cite{BreuillardGelander}.

\begin{theorem}[Binary dense sector control]\label{thm:binary-dense}
For every finite active profile \(\mathbf D\), there exist two
sector-preserving unitaries \(U_0,U_1\) such that
\[
  \overline{\langle U_0,U_1\rangle}
  =
  \prod_\alpha SU(D_\alpha).
\]
\end{theorem}

\begin{proof}
If every \(D_\alpha=1\), both sides of the asserted closure are the
trivial group.  Otherwise, discard the one-dimensional factors and
apply the following argument to the factors with \(D_\alpha\geq2\).
Take
\[
  U_0=e^{i\varepsilon H_0},
  \qquad
  U_1=e^{i\varepsilon H_1}
\]
with \(\varepsilon>0\) sufficiently small that both lie in the
Breuillard--Gelander neighborhood.  Their logarithms generate the
full direct-sum Lie algebra by \cref{lem:binary-lie}, so the generated
group is dense.  The parameter may moreover be chosen outside a
countable exceptional set.  In particular, for any prescribed
nonzero spectral gap \(\Delta\) of \(H_0\), we may and do require
\(\varepsilon\Delta/(2\pi)\notin\mathbb Q\).

The closure of the positive-word semigroup generated by
\(U_0,U_1\) is a closed subsemigroup of a compact group and hence a
group.  It therefore equals the same dense group closure even though
input words do not contain formal inverse letters.
\end{proof}

This two-letter statement is an abstract global-control result inside
the symmetry commutant.  It does not assert that the same pair is
generated by geometrically local or hardware-native gates.  Symmetry
and locality can impose additional restrictions on realizable
unitaries \cite{MarvianLocality}; under such restrictions the exact
instance rank in \cref{thm:exact-rank} and the tangent criterion in
\cref{thm:tangent-witness} remain applicable, but the saturated profile
law need not be attainable.

Dense binary control supplies the finite alphabet, but the lower bound
also requires a family of states placed exactly at the cutpoint and a
controlled perturbation that realizes every sign pattern.  The next
construction implements this prepare--test geometry with an explicit
uniform margin.

\subsection{A binary prepare--test construction}

Fix nontrivial ranks \(r_\alpha\) and
\(q_\alpha=D_\alpha-r_\alpha\).  In every sector choose
\[
  W_\alpha=A_\alpha\oplus B_\alpha
\]
with bases
\[
  \{e_{\alpha,a}\}_{a=1}^{r_\alpha},
  \qquad
  \{f_{\alpha,b}\}_{b=1}^{q_\alpha},
\]
and let \(P_\alpha\) project onto \(A_\alpha\).

Set
\[
  \kappa
  =
  \sum_\alpha2r_\alpha q_\alpha.
\]
For every accepting--rejecting pair define
\[
  |\phi^R_{\alpha ab}\rangle
  =
  \frac{|e_{\alpha,a}\rangle+|f_{\alpha,b}\rangle}{\sqrt2},
\]
\[
  |\phi^I_{\alpha ab}\rangle
  =
  \frac{|e_{\alpha,a}\rangle+i|f_{\alpha,b}\rangle}{\sqrt2}.
\]

\begin{theorem}[Binary prepare--test lower bound]\label{thm:prepare-test}
For every active profile \((\mathbf D,\mathbf r)\), there exists a
binary sector-preserving measure-once automaton such that every real PFA recognizing the same strict-cutpoint language has at least
\[
  2+\kappa(\mathbf D,\mathbf r)
  =
  2+\sum_\alpha2r_\alpha q_\alpha
\]
states.
\end{theorem}

\begin{proof}
Use the two dense generators from \cref{thm:binary-dense}.  Let \(s\)
be the number of active sectors and set \(p=1/s\).  Choose a pure
initial state whose sector-dephased state has weight \(p\) on one
fixed basis vector in every sector.

For a coordinate
\[
  j=(\alpha,a,b,R)
  \quad\text{or}\quad
  j=(\alpha,a,b,I),
\]
independent sector control gives a target prefix unitary that maps the
initial component in sector \(\alpha\) to
\(|\phi^R_{\alpha ab}\rangle\) or
\(|\phi^I_{\alpha ab}\rangle\), and maps every other sector component
to a rejecting basis vector.  At the identity suffix, its acceptance
probability is
\[
  \tau=\frac p2.
\]

For a sign vector
\(\eta\in\{\pm1\}^{\kappa}\), define
\[
  (Z_{\alpha,\eta})_{ab}
  =
  \eta^R_{\alpha ab}
  -
  i\eta^I_{\alpha ab}
\]
and
\[
  X_{\alpha,\eta}
  =
  \begin{pmatrix}
    0 & Z_{\alpha,\eta}\\
    -Z_{\alpha,\eta}^\dagger & 0
  \end{pmatrix}.
\]
Let
\[
  W_\eta(t)
  =
  \bigoplus_\alpha e^{tX_{\alpha,\eta}}.
\]
For the prefix corresponding to coordinate \(j\), write
\(F_{j,\eta}(t)\) for the resulting acceptance probability.  Direct
differentiation gives
\[
  F_{j,\eta}(0)=\tau,
  \qquad
  F'_{j,\eta}(0)=p\eta_j.
\]
For the target sector this follows by evaluating the corresponding
two-dimensional accepting--rejecting block.  Every nontarget sector
was prepared in a rejecting basis vector, so its first derivative
vanishes because $P_\alpha$ annihilates that vector.

Moreover,
\[
  \|Z_{\alpha,\eta}\|^2
  \leq
  \|Z_{\alpha,\eta}\|_F^2
  =
  2r_\alpha q_\alpha.
\]
The double-commutator bound
\[
  \|[X,[X,P]]\|
  \leq4\|X\|^2
\]
must be applied blockwise.  If $\rho_{j,\alpha}(t)$ is the evolved
prefix block in sector $\alpha$, then
$\|\rho_{j,\alpha}(t)\|_1=p$.  Hence
\begin{align*}
  |F''_{j,\eta}(t)|
  &\leq
  \sum_\alpha
  \|\rho_{j,\alpha}(t)\|_1
  \|[X_{\alpha,\eta},[X_{\alpha,\eta},P_\alpha]]\|\\
  &\leq
  4p\sum_\alpha\|X_{\alpha,\eta}\|^2
  \leq4p\sum_\alpha2r_\alpha q_\alpha
  =4p\kappa.
\end{align*}
Choose
\[
  t=\frac{1}{8\kappa}.
\]
Taylor's theorem yields
\[
  \eta_j\bigl(F_{j,\eta}(t)-\tau\bigr)
  \geq
  \frac{p}{8\kappa}
  -
  \frac{p}{32\kappa}
  =
  \frac{3p}{32\kappa}.
\]
In particular, all signs are stable with the conservative margin
\[
  \gamma=\frac{p}{16\kappa}.
\]

Add one anchor preparation that maps the initial component in every
active sector to an accepting basis vector.  For a unit vector
\(e\in\operatorname{ran}P_\alpha\), the elementary estimate
\[
  \|(I-P_\alpha)e^{tX_{\alpha,\eta}}e\|^2
  \leq
  2t^2\|X_{\alpha,\eta}\|^2
\]
holds here because \(t\|X_{\alpha,\eta}\|\leq1/8\).  Consequently the
anchor acceptance under any target suffix is at least
\[
  1-2pt^2\sum_\alpha\|X_{\alpha,\eta}\|^2
  \geq
  1-2p\kappa t^2
  =
  1-\frac{p}{32\kappa}
  >\tau.
\]
Thus the anchor supplies a common positive row.

It remains to retain information beyond this finite sign matrix.  Fix
one coordinate plane \((e_{\alpha_0,a_0},f_{\alpha_0,b_0})\), and use
the choice in \cref{thm:binary-dense} for which the relative phase
\(\theta/(2\pi)\) of \(U_0\) on this plane is irrational.  Add a
carrier preparation whose \(\alpha_0\)-component is
\((e_{\alpha_0,a_0}+f_{\alpha_0,b_0})/\sqrt2\), with all other active
components in rejecting eigenvectors of \(H_0\).  Along the carrier
orbit \(U_0^k\), direct first-order expansion gives
\[
 f(za^ky_\eta)-\tau
 =pt\bigl(\eta^R_{\alpha_0a_0b_0}\cos(k\theta)
          \pm\eta^I_{\alpha_0a_0b_0}\sin(k\theta)\bigr)
   +R_{\eta,k},
 \]
where the sign convention in the sine term is immaterial and the
blockwise double-commutator estimate above gives the uniform bound
\( |R_{\eta,k}|\leq2p\kappa t^2\).  More explicitly, the leading
sinusoid has amplitude \(p\sqrt2,t\).  Our choice
\(t=1/(8\kappa)\) therefore makes its amplitude strictly larger than
the remainder.

There are finitely many target prefixes, the anchor, the carrier, and
suffixes.  Dense generation allows each to be approximated by a binary
word closely enough to preserve both the static margin and the strict
amplitude inequality.  Irrationality of \(\theta/(2\pi)\) implies that
for every \(q\geq1\) the subsequence \(k=qn\) is dense on the phase
circle.  Hence every approximated suffix is crossed infinitely often
on both sides of the cutpoint along that subsequence.  The static rows
give the one-sided \(H_\kappa^{(+1)}\) hypotheses, and
\cref{lem:dynamic-affine-lift} now gives the lower bound
\(\kappa+2\).
\end{proof}

If there are $s$ active sectors, assign weight $1/(2s-1)$ to each and
weight $(s-1)/(2s-1)$ to a one-dimensional always-accepting spectator.
The common cutpoint is then $1/2$, and the active-sector construction is
unchanged apart from this uniform rescaling.

The dynamic lift already supplies the extra state for every nontrivial
rank profile.  For rank-one readouts, curvature supplies an independent
purely finite certificate: it enlarges the sign matrix itself, without
using a Markov limit or an infinite carrier orbit.

\begin{theorem}[Curved rank-one readout witness]
\label{thm:curved-witness}
Suppose every active sector has $D_\alpha\geq2$ and accepting rank
$r_\alpha=1$.  There exists a binary sector-preserving measure-once
automaton with this profile such that every real PFA recognizing the
same strict-cutpoint language has at least
\[
  \kappa(\mathbf D,\mathbf 1)+2
  =
  2+\sum_\alpha2(D_\alpha-1)
\]
states.
\end{theorem}

\begin{proof}
Let $L$ be the number of active sectors, set $p_\alpha=1/L$, and fix
$s\in(1/2,1)$.  In each sector choose a unit vector $\psi_\alpha$ and
a real-orthonormal basis
\[
  \chi_{\alpha,1},\ldots,\chi_{\alpha,d_\alpha}
  \quad\text{of}\quad
  \psi_\alpha^\perp,
  \qquad
  d_\alpha=2(D_\alpha-1),
\]
where orthogonality is taken for the inner product
$\operatorname{Re}\langle\cdot,\cdot\rangle$.  Put
$\kappa=\sum_\alpha d_\alpha$ and take the initial reduced state
\[
  \rho_{\rm rad}
  =
  \bigoplus_\alpha
  p_\alpha|\psi_\alpha\rangle\langle\psi_\alpha|
\]
at cutpoint $\tau=s$.  The dense binary controls from
\cref{thm:binary-dense} approximate independent special-unitary
motions in all active sectors.

For $\varepsilon>0$, let $\rho_{\alpha,i}$ replace $\psi_\alpha$ by
\[
  \cos\varepsilon\,\psi_\alpha
  +\sin\varepsilon\,\chi_{\alpha,i}
\]
in sector $\alpha$ and retain $\psi_\mu$ in every other sector.  Let
$\rho_*$ carry $\chi_{\alpha,1}$ in every sector.  These states and
$\rho_{\rm rad}$ lie in the closure of the prefix orbit.

Fix a sign vector
$\eta=(\eta_0,(\eta_{\alpha,i}))\in\{\pm1\}^{\kappa+1}$ and set
\[
  w_{\alpha,\eta}
  =
  \frac1{\sqrt{d_\alpha}}
  \sum_{i=1}^{d_\alpha}
  \eta_{\alpha,i}\chi_{\alpha,i},
  \qquad
  s_{\alpha,\eta}=s+\eta_0\varepsilon^3.
\]
For sufficiently small $\varepsilon$, the vector
\[
  u_{\alpha,\eta}
  =
  \sqrt{s_{\alpha,\eta}}\,\psi_\alpha
  +
  \sqrt{1-s_{\alpha,\eta}}\,w_{\alpha,\eta}
\]
is unit, and the rank-one effect
$E_\eta=\bigoplus_\alpha
|u_{\alpha,\eta}\rangle\langle u_{\alpha,\eta}|$
lies in the closure of the suffix orbit.  Since
$\operatorname{Re}\langle\chi_{\alpha,i},w_{\alpha,\eta}\rangle
=\eta_{\alpha,i}/\sqrt{d_\alpha}$, one has
\[
  \Tr(E_\eta\rho_{\rm rad})-\tau
  =
  \eta_0\varepsilon^3
\]
and
\[
  \Tr(E_\eta\rho_{\alpha,i})-\tau
  =
  \eta_0\varepsilon^3+T_{\alpha,i}+R_{\alpha,i},
\]
where
\[
  T_{\alpha,i}
  =
  2p_\alpha\sin\varepsilon\cos\varepsilon
  \sqrt{s_{\alpha,\eta}(1-s_{\alpha,\eta})}
  \frac{\eta_{\alpha,i}}{\sqrt{d_\alpha}},
  \qquad
  |R_{\alpha,i}|\leq p_\alpha\sin^2\varepsilon.
\]
There is a constant $c_0>0$ such that
$|T_{\alpha,i}|\geq c_0\varepsilon$ for every $\alpha,i,\eta$ once
$\varepsilon$ is small.  Choose it still smaller so that
$\varepsilon^3+\max_\alpha p_\alpha\sin^2\varepsilon<c_0\varepsilon$.
The radial row then has sign $\eta_0$, and every coordinate row has
sign $\eta_{\alpha,i}$.

For the anchor,
\[
  \Tr(E_\eta\rho_*)
  =
  \sum_\alpha
  p_\alpha(1-s_{\alpha,\eta})
  |\langle\chi_{\alpha,1},w_{\alpha,\eta}\rangle|^2
  \leq
  1-\sum_\alpha p_\alpha s_{\alpha,\eta}.
\]
Hence
$\Tr(E_\eta\rho_*)-\tau\leq1-2s+\varepsilon^3<0$ uniformly in
$\eta$.  The $\kappa+2$ preparation rows therefore form
$H_{\kappa+1}^{(-1)}$ against the $2^{\kappa+1}$ target effects.
All margins are strict, so density replaces the targets by finite
binary words without changing a sign.  By
\cref{lem:affine-complete-sign,lem:pfa-cutpoint-rank}, every equivalent
real PFA has at least $\kappa+2$ states.
\end{proof}

\subsection{Exact alphabet-capacity phase diagram}

The binary construction proves that two symbols attain the full
noncommutative capacity.  A single symbol has a more rigid spectral
structure, and its exact loss can be computed block by block.  Define
\[
  B_q(K;\Hh)
  =
  \sup\{\beta(\A):
      \A\text{ is $K$-equivariant and }|\Sigma|\leq q\},
\]
and let
\[
  \varrho_K
  =
  \sum_\lambda(m_\lambda-1).
\]
This is the sum of the ranks of the semisimple factors of the
commutant.

\begin{lemma}[Hankel rank of exponential sequences]\label{lem:exponential-hankel}
Let $z_1,\ldots,z_s$ be distinct nonzero complex numbers and let
$c_1,\ldots,c_s$ be nonzero.  The sequence
$h_t=\sum_{r=1}^s c_r z_r^t$ has complex Hankel rank $s$.  If $h_t$ is
real for all $t$, its real Hankel rank is also $s$.
\end{lemma}

\begin{proof}
The Hankel matrix factors through the Vandermonde vectors
$(1,z_r,z_r^2,\ldots)$.  Every finite $s\times s$ Vandermonde minor is
nonsingular, giving rank $s$ over $\C$.  Complexification preserves the
rank of a real matrix, which proves the final statement.
\end{proof}

\begin{theorem}[Exact alphabet capacity]\label{thm:alphabet-capacity}
For every finite-dimensional compact symmetry representation,
\[
  B_1(K;\Hh)
  =
  1+\sum_\lambda m_\lambda(m_\lambda-1)
  =
  1+\nu_K-\varrho_K,
\]
whereas
\[
  B_q(K;\Hh)=1+\nu_K,
  \qquad q\geq2.
\]
Hence
\[
  B_2(K;\Hh)-B_1(K;\Hh)
  =
  \varrho_K.
\]
\end{theorem}

\begin{proof}
For a unary automaton, write the reduced transition as
\[
  U=\bigoplus_\lambda U_\lambda
\]
and let $e^{i\theta_{\lambda,j}}$ be the eigenvalues of
$U_\lambda$.  The conjugation operator $\Ad_U$ is normal and acts on a
matrix unit $E_{jk}$ by the eigenvalue
\[
  e^{i(\theta_{\lambda,j}-\theta_{\lambda,k})}.
\]
A cyclic subspace of a normal operator contains at most one direction
from each distinct eigenspace.  Every diagonal matrix unit belongs to
the common eigenvalue $1$, so all block-diagonal directions together
contribute at most one constant dimension.  The off-diagonal ordered
pairs contribute at most
$\sum_\lambda m_\lambda(m_\lambda-1)$ dimensions.  This gives the unary
upper bound.

To attain it, choose the eigenphases so that all nonzero differences
$\theta_{\lambda,j}-\theta_{\lambda,k}$ are globally distinct modulo
$2\pi$ and avoid $\pi$.  Only finitely many affine hyperplanes are
excluded, so such a choice exists.  For every block choose
\[
  |\psi_\lambda\rangle
  =
  \frac{1}{\sqrt{m_\lambda}}
  \sum_{j=1}^{m_\lambda}
  e^{i\alpha_{\lambda j}}|j\rangle
\]
with arbitrary phases, assign it a positive isotypic weight, and choose
$Q_\lambda=|\phi_\lambda\rangle\langle\phi_\lambda|$ with every
coordinate of $|\phi_\lambda\rangle$ nonzero.  The state and effect
then have no zero matrix entries in the eigenbasis.  These reduced pure states arise from a global pure input by choosing product vectors in the isotypic blocks and superposing them with the prescribed positive weights.  The unary acceptance
sequence then has the form
\[
  f_\A(a^t)
  =
  c_0+
  \sum_{\lambda}\sum_{j\neq k}
  c_{\lambda jk}
  e^{it(\theta_{\lambda,j}-\theta_{\lambda,k})},
\]
with every $c_{\lambda jk}\neq0$.  By \cref{lem:exponential-hankel}, the Hankel rank attains the upper bound.

For $q\geq2$, \cref{thm:binary-dense} gives a word-closure
containing $\prod_\lambda SU(m_\lambda)$.  \Cref{thm:commutant-law} then gives rank $1+\nu_K$, and no automaton can exceed this value.
\end{proof}

A unary transition resolves generic nonzero Bohr frequencies but cannot
separate the traceless diagonal directions contained in the zero
frequency.  A second noncommuting transition releases exactly the
$\sum_\lambda(m_\lambda-1)$ Cartan directions.  The threshold from one
to two symbols is therefore an exact noncommutative alphabet transition,
not merely a limitation of a particular construction.

Central phases and central populations play different roles here.  A
central commutant unitary is scalar on each isotypic block, so its
conjugation acts trivially on invariant states and effects.  Missing
relative central phase gates can obstruct synthesis of the full
invariant unitary group, but they do not reduce the measure-once
behavior capacity above.  By contrast, the component laws for
covariant channels concern relative central populations; those are
genuine observable coordinates once dissipation is allowed to move
weight between labels.

\subsection{Tight binary profile laws}

Let
\[
  \operatorname{SC}^{(2)}_{\R}
  (\mathbf D,\mathbf r)
\]
be the supremum, over binary automata with the indicated active sector
dimensions and accepting ranks, of the minimum number of states in an
equivalent real PFA.

\begin{theorem}[Fixed-profile binary bounds]\label{thm:profile-bounds}
For every active profile,
\[
  2+\kappa(\mathbf D,\mathbf r)
  \leq
  \operatorname{SC}^{(2)}_{\R}
  (\mathbf D,\mathbf r)
  \leq
  \mathsf B(\mathbf D)+1.
\]
\end{theorem}

\begin{proof}
The lower bound is \cref{thm:prepare-test}.  For rank-one profiles,
\cref{thm:curved-witness} gives the same numerical bound through a
finite sign-rank certificate.  The upper bound follows from
the centered representation on the profile space: its trace coordinate
is the normalized constant mode, so
\(\beta_\tau(\A)\leq\mathsf B(\mathbf D)\).  Apply
\cref{thm:codimension-one-pfa}.
\end{proof}

For fixed \(D\),
\[
  \max_{1\leq r\leq D-1}2r(D-r)
  =
  \left\lfloor\frac{D^2}{2}\right\rfloor.
\]

\begin{corollary}[Tight worst-rank law]\label{cor:worst-profile}
Let
\[
  \operatorname{SC}^{(2)}_{\R}(\mathbf D)
  =
  \max_{\mathbf r}
  \operatorname{SC}^{(2)}_{\R}
  (\mathbf D,\mathbf r).
\]
Then
\[
  2+\sum_\alpha
  \left\lfloor\frac{D_\alpha^2}{2}\right\rfloor
  \leq
  \operatorname{SC}^{(2)}_{\R}(\mathbf D)
  \leq
  2+\sum_\alpha(D_\alpha^2-1).
\]
In particular,
\[
  \operatorname{SC}^{(2)}_{\R}(\mathbf D)
  =
  \Theta
  \left(
    1+\sum_\alpha(D_\alpha^2-1)
  \right)
\]
uniformly over finite profiles.
\end{corollary}

\begin{proof}
The lower bound uses balanced accepting ranks, the affine anchor, and
the dynamic lift.
Since
\[
  \left\lfloor\frac{D^2}{2}\right\rfloor
  \geq
  \frac{D^2-1}{2},
\]
the lower bound is at least
\[
  \frac{\mathsf B(\mathbf D)-1}{2}.
\]
The upper bound is linear in \(\mathsf B(\mathbf D)\).
\end{proof}

Thus the finite-alphabet problem closes at the same scale as the
sectorwise operator capacity: balanced readouts expose a constant
fraction of \(\mathsf B(\mathbf D)\), and the codimension-one stochastic
embedding leaves only one additional probabilistic state.  The pair
\((\mathsf B,\kappa)\) also quantifies how this cost changes when the
symmetry constraints themselves are relaxed.

\begin{corollary}[Binary probabilistic state law under compact symmetry]\label{cor:compact-pfa}
Let $\operatorname{SC}^{(2)}_K(\Hh)$ denote the supremum, over binary
$K$-equivariant measure-once automata on $\Hh$, of the minimum number of
states in an equivalent real probabilistic finite automaton under a
strict cutpoint.  Then
if \(\nu_K>0\),
\[
  2+\sum_\lambda
  \left\lfloor\frac{m_\lambda^2}{2}\right\rfloor
  \leq
  \operatorname{SC}^{(2)}_K(\Hh)
  \leq
  \nu_K+2.
\]
If \(\nu_K=0\), then \(\operatorname{SC}^{(2)}_K(\Hh)=1\).
In particular,
\[
  \operatorname{SC}^{(2)}_K(\Hh)
  =
  \Theta(1+\nu_K).
\]
\end{corollary}

\begin{proof}
Use balanced accepting ranks
$r_\lambda=\lfloor m_\lambda/2\rfloor$ in the binary prepare--test
construction.  Then
\[
  2+\sum_\lambda
  2r_\lambda(m_\lambda-r_\lambda)
  =
  2+\sum_\lambda\left\lfloor\frac{m_\lambda^2}{2}\right\rfloor.
\]
The upper bound follows from the constant-normalized behavior
representation of dimension at most $1+\nu_K$ and
\cref{thm:codimension-one-pfa}.  When \(\nu_K=0\), every invariant
measure-once acceptance function is constant, so one state is both
necessary and sufficient.
\end{proof}

The noncommutative commutant dimension therefore controls both the
largest exact Hankel behavior and, up to universal constants, the worst
strict-cutpoint probabilistic state cost.

\section{Dissipative release: channels and central mobility}

Reversible dynamics freeze the center; dissipation can make central populations writable.  This is the structural change that produces the general-channel dichotomy.  A mobility partition records exactly which sums of isotypic populations remain conserved, so the capacity interpolates between reversible read-only center coordinates and full mobility, where the whole invariant algebra becomes behavioral memory.

\subsection{Model and exact channel rank}

\begin{definition}\label{def:covariant-1gqfa}
A symmetry-compatible $K$-covariant measure-once one-way general quantum
finite automaton is a tuple
\[
  \A=(\Hh,\rho_0,\{\Phi_a\}_{a\in\Sigma},P,\tau),
\]
where $\rho_0$ is a density operator, $P\in\mathcal C_K$ is an
orthogonal projector, $\tau\in\R$, and every
\[
  \Phi_a:\operatorname{End}(\Hh)\longrightarrow
  \operatorname{End}(\Hh)
\]
is completely positive and trace preserving.  Each symbol channel
satisfies
\[
  \Phi_a\circ\Ad_{\pi(k)}
  =
  \Ad_{\pi(k)}\circ\Phi_a
\]
for every $a\in\Sigma$ and $k\in K$.  For
$w=a_1\cdots a_m$, set
\[
  \Phi_w=\Phi_{a_m}\circ\cdots\circ\Phi_{a_1}
\]
and
\[
  f_\A(w)=\Tr\!\left(P\Phi_w(\rho_0)\right).
\]
The recognized language is
\[
  L_{\A,\tau}
  =
  \{w\in\Sigma^*:f_\A(w)>\tau\}.
\]
\end{definition}

This is the measure-once general-channel model of
\cite{LiEtAl1gQFA}, with covariance imposed separately on every input
symbol and with a symmetry-compatible final readout.  The covariance
condition is strictly weaker than requiring every Kraus operator of
every channel to belong to $\mathcal C_K$.

Set
\[
  \widetilde\rho_0=\mathcal T_K(\rho_0).
\]
Covariance implies
\(
  \Phi_a\mathcal T_K=\mathcal T_K\Phi_a
\), and $P\in\mathcal C_K$ gives
\[
  f_\A(w)
  =
  \Tr\!\left(P\Phi_w(\widetilde\rho_0)\right).
\]
Thus the physical initial state need not itself be invariant.  Define
\[
  \cR_\Phi
  =
  \spanR\{\Phi_x(\widetilde\rho_0):x\in\Sigma^*\},
\]
\[
  \cE_\Phi
  =
  \spanR\{\Phi_y^*(P):y\in\Sigma^*\}.
\]

From this point onward, a sector-preserving channel construction means
its multiplicity-space realization after Haar reduction; the physical
map is always the corresponding \(K\)-covariant extension unless a
stronger Kraus condition is stated explicitly.

\begin{theorem}[Exact channel behavioral-memory invariant]
\label{thm:channel-exact-rank}
For every automaton in \cref{def:covariant-1gqfa},
\[
  \dim\mathscr M_\A
  =
  \beta(\A)
  =
  \rank_\R
  \left(
    \langle\cdot,\cdot\rangle_{\mathrm{HS}}
    \big|_{\cR_\Phi\times\cE_\Phi}
  \right).
\]
Moreover,
\[
  \cR_\Phi,\cE_\Phi
  \subseteq\Herm(\mathcal C_K).
\]
\end{theorem}

\begin{proof}
For all $x,y\in\Sigma^*$,
\[
  H_\A(x,y)
  =
  \Tr\!\left(P\Phi_y\Phi_x(\widetilde\rho_0)\right)
  =
  \Tr\!\left(\Phi_y^*(P)\Phi_x(\widetilde\rho_0)\right).
\]
The Hankel matrix is therefore the matrix of the Hilbert--Schmidt
pairing on the displayed generating families.  Passing to either real
span does not change its rank.

The twirled state belongs to $\mathcal C_K$, and covariance preserves
the fixed-point algebra.  The adjoint of a covariant channel is
covariant, so every $\Phi_y^*(P)$ also belongs to $\mathcal C_K$.
Hermiticity preservation proves the inclusions.
\end{proof}

Only one invariant boundary is necessary.  With $P\in\mathcal C_K$, an
arbitrary initial state may be twirled as above.  Conversely, if
$\rho_0\in\mathcal C_K$, every reachable state is invariant and an
arbitrary final effect may be replaced behaviorally by its Haar twirl.
If neither boundary is invariant, covariance alone need not confine the
behavior to $\Herm(\mathcal C_K)$.

\subsection{Mobility components and the channel commutant law}

Let $Z_\lambda$ denote the central projection onto the
$\lambda$-isotypic summand.  Choose a partition
\[
  \Pi=\{C_1,\ldots,C_c\}
\]
of the occurring isotypic labels and define
\[
  Z_C=\sum_{\lambda\in C}Z_\lambda,
  \qquad
  D_C=\dim_\C(Z_C\mathcal C_K)
  =\sum_{\lambda\in C}m_\lambda^2.
\]
It is useful to separate its noncommutative part:
\[
  \nu_C
  =\dim_\C[Z_C\mathcal C_K,Z_C\mathcal C_K]
  =\sum_{\lambda\in C}(m_\lambda^2-1).
\]
Thus
\[
  D_K:=\dim_\C\mathcal C_K=\sum_{C\in\Pi}D_C.
\]

\begin{definition}\label{def:component-conservation}
A $K$-covariant channel $\Phi$ is $\Pi$-conservative when
\[
  \Phi^*(Z_C)=Z_C
\]
for every $C\in\Pi$.  An automaton is $\Pi$-conservative when every
symbol channel is $\Pi$-conservative.
\end{definition}

This condition is equivalent to exact preservation of every component
weight:
\[
  \Tr\!\left(Z_C\Phi(\rho)\right)=\Tr(Z_C\rho)
\]
for every density operator $\rho$.  Define
\[
  M_\Pi=D_K-c+1
  =1+\nu_K+\bigl(\dim_\C Z(\mathcal C_K)-c\bigr).
\]
For $q\geq1$, let $B_q^{\mathrm{1g}}(K,\Pi;\Hh)$ be the supremum of
$\beta(\A)$ over all symmetry-compatible, $\Pi$-conservative
$K$-covariant automata on $\Hh$ with $|\Sigma|\leq q$.

\begin{theorem}[Component-conserving channel capacity]
\label{thm:component-channel-capacity}
Every symmetry-compatible, $\Pi$-conservative $K$-covariant automaton
satisfies
\[
  \beta(\A)\leq M_\Pi.
\]
For every $q\geq1$,
\[
  B_q^{\mathrm{1g}}(K,\Pi;\Hh)=M_\Pi.
\]
In particular, one input symbol already attains the full capacity.
\end{theorem}

\begin{proof}
Let
\[
  p_C=\Tr(Z_C\widetilde\rho_0).
\]
Every reachable state lies in
\[
  \mathcal S_{\mathbf p}
  =
  \{X\in\Herm(\mathcal C_K):
    \Tr(Z_CX)=p_C\ \text{for all }C\in\Pi\}.
\]
The $c$ component-trace functionals are linearly independent, so this
affine space has real dimension $D_K-c$.  Its linear span has dimension
at most $D_K-c+1=M_\Pi$.  The upper bound follows from
\cref{thm:channel-exact-rank}.

We construct a unary automaton attaining equality.  Set
\[
  \Hh_C=Z_C\Hh,
  \qquad
  n_C=\dim\Hh_C,
  \qquad
  \omega_C=\frac{Z_C}{n_C},
\]
and
\[
  \mathcal V_{C,0}
  =
  \{X\in\Herm(Z_C\mathcal C_K):\Tr X=0\}.
\]
Then $\dim_\R\mathcal V_{C,0}=D_C-1$.  If $D_C\geq2$, the
finite-dimensional $C^*$-algebra $Z_C\mathcal C_K$ contains an
orthogonal projection
\[
  0<Q_C<Z_C.
\]
The functional $X\mapsto\Tr(Q_CX)$ is nonzero on
$\mathcal V_{C,0}$; otherwise $Q_C$ would be Hilbert--Schmidt
orthogonal to every trace-zero element of the algebra and hence would
be a scalar multiple of $Z_C$, contrary to $0<Q_C<Z_C$.  We may
therefore choose a basis
\[
  F_{C,1},\ldots,F_{C,D_C-1}
\]
such that
\[
  \Tr(Q_CF_{C,j})\neq0
\]
for every $j$.  Choose globally pairwise distinct nonzero real numbers
$\lambda_{C,j}$ and define a Hermiticity-preserving linear map by
\[
  L_C(\omega_C)=0,
  \qquad
  L_C(F_{C,j})=\lambda_{C,j}F_{C,j},
\]
extended complex linearly to $Z_C\mathcal C_K$.  Let
\[
  \Psi_C=L_C\circ
  \left.\mathcal T_K\right|_{\operatorname{End}(\Hh_C)}.
\]
Every value of $\Psi_C$ has trace zero.
When $D_C=1$, set $L_C=0$ and $\Psi_C=0$.

The replacement channel
\[
  \mathcal R_{\omega_C}(X)=\Tr(X)\omega_C
\]
has positive-definite Choi matrix
\[
  J(\mathcal R_{\omega_C})=I_{\Hh_C}\otimes\omega_C.
\]
Since the number of components is finite, a sufficiently small common
$\varepsilon>0$ makes
\[
  \Phi_C=\mathcal R_{\omega_C}+\varepsilon\Psi_C
\]
completely positive.  It is trace preserving because $\Psi_C$ is trace
annihilating.  Both summands are $K$-covariant: the twirl removes the
input action, and the output of $\Psi_C$ is invariant.  Define
\[
  \Phi(X)=\sum_{C\in\Pi}\Phi_C(Z_CXZ_C).
\]
This is a completely positive trace-preserving $K$-covariant channel;
it kills intercomponent coherences and satisfies
$\Phi^*(Z_C)=Z_C$ for every component.

Choose positive $p_C$ with $\sum_Cp_C=1$.  For sufficiently small
$\delta_C>0$,
\[
  \eta_C
  =
  \omega_C+\delta_C\sum_{j=1}^{D_C-1}F_{C,j}
\]
is a density operator on $\Hh_C$.  Set
\[
  \rho_0=\bigoplus_{C\in\Pi}p_C\eta_C,
  \qquad
  P=\sum_{C:D_C\geq2}Q_C.
\]
Use the unary alphabet $\Sigma=\{a\}$ with $\Phi_a=\Phi$.  After
decreasing $\varepsilon$ if necessary, the numbers
\[
  1,
  \qquad
  z_{C,j}=\varepsilon\lambda_{C,j}
\]
are nonzero and pairwise distinct.  Direct induction gives
\[
  \Phi^t(\rho_0)
  =
  \bigoplus_{C\in\Pi}p_C
  \left(
    \omega_C+
    \delta_C\sum_{j=1}^{D_C-1}z_{C,j}^{\,t}F_{C,j}
  \right).
\]
Consequently,
\[
  f_\A(a^t)
  =
  c_0+
  \sum_{C\in\Pi}\sum_{j=1}^{D_C-1}
  c_{C,j}z_{C,j}^{\,t},
\]
where
\[
  c_0=\sum_{C:D_C\geq2}p_C\Tr(Q_C\omega_C),
  \qquad
  c_{C,j}=p_C\delta_C\Tr(Q_CF_{C,j}).
\]
When some $D_C\geq2$, all these coefficients are nonzero.  The
Vandermonde factorization in \cref{lem:exponential-hankel} yields
\[
  \beta(\A)
  =
  1+\sum_{C\in\Pi}(D_C-1)
  =D_K-c+1.
\]
If every $D_C=1$, then $M_\Pi=1$; taking $P=I$ gives the constant
function one and rank one.
\end{proof}

Fix now an invariant accepting projector \(P\in\mathcal C_K\), and
write
\[
  P_C=Z_CP,
  \qquad
  \mathcal A_{\Pi,P}
  =
  \{C\in\Pi:0<P_C<Z_C\}.
\]
For \(P\ne0\), define
\[
  M_{\Pi,P}
  =
  1+\sum_{C\in\mathcal A_{\Pi,P}}(D_C-1).
\]
Let \(B_{q,\Pi,P}^{\mathrm{1g}}(K;\Hh)\) denote the preceding channel
capacity with this readout held fixed.

\begin{theorem}[Prescribed-readout component capacity]
\label{thm:fixed-readout-component-capacity}
If \(P=0\), then
\(B_{q,\Pi,P}^{\mathrm{1g}}(K;\Hh)=0\).  If \(P\ne0\), then for every
\(q\geq1\),
\[
  B_{q,\Pi,P}^{\mathrm{1g}}(K;\Hh)=M_{\Pi,P}.
\]
Thus a component contributes its \(D_C-1\) trace-zero directions if
and only if the prescribed readout is nontrivial on that component;
all component baselines together contribute only one constant mode.
\end{theorem}

\begin{proof}
For a \(\Pi\)-conservative channel,
\(\Phi^*(Z_C)=Z_C\).  Because $\Phi^*$ is unital and $Z_C$ is a
projection, equality holds in the Schwarz inequality for $Z_C$; hence
$Z_C$ belongs to the multiplicative domain of $\Phi^*$ and
\[
  \Phi^*(Z_CYZ_C)=Z_C\Phi^*(Y)Z_C.
\]
Therefore \(P_C=0\) remains zero under every suffix, whereas
\(P_C=Z_C\) remains \(Z_C\).  Such components contribute no
nonconstant observable direction.  On a component in
\(\mathcal A_{\Pi,P}\), the suffix effect can pair with at most the
\((D_C-1)\)-dimensional trace-zero corner.  The fixed component
weights supply at most one shared constant series.  This proves the
upper bound \(M_{\Pi,P}\), and \(P=0\) gives the zero series.

For unary saturation, repeat the Choi-interior construction in the
proof of \cref{thm:component-channel-capacity}, but only on the active
components and with \(Q_C=P_C\).  The functional
\[
  \ell_C(X)=\Tr(P_CX)
\]
is nonzero on
\(V_C=\{X\in\Herm(Z_C\mathcal C_K):\Tr X=0\}\): otherwise \(P_C\)
would be a scalar multiple of \(Z_C\), impossible for a nontrivial
projection.  A basis of \(V_C\) can therefore be chosen so that
\(\ell_C\) is nonzero on every basis vector.  Assign globally distinct
nonzero contraction eigenvalues to these basis vectors, use replacement
dynamics on the inactive components, and choose positive initial
component weights.  The acceptance sequence then consists of one
nonzero constant mode together with exactly
\(\sum_{C\in\mathcal A_{\Pi,P}}(D_C-1)\) distinct nonconstant modes.
\Cref{lem:exponential-hankel} gives rank \(M_{\Pi,P}\).  If the active
set is empty and \(P\ne0\), positive initial weight on a component
where \(P_C=Z_C\) gives a nonzero constant series of rank one.
\end{proof}

The formula has two natural endpoints.  If all isotypic labels form one
mobility component, then
\[
  B_q^{\mathrm{1g}}(K,\{\Lambda\};\Hh)=D_K.
\]
If every component is a singleton, then
\[
  M_\Pi
  =D_K-\dim_\C Z(\mathcal C_K)+1
  =1+\nu_K.
\]

Writing \(z_K=\dim_\C Z(\mathcal C_K)\), the component law has the
structural form
\[
  M_\Pi=1+\nu_K+(z_K-|\Pi|).
\]
Thus each independent merge of conserved central components releases
one relative population coordinate in addition to the noncommutative
directions already available to reversible dynamics.  More generally,
if \(H\subseteq K\) and \(\Pi\) is a mobility partition of the
\(H\)-isotypic labels, then
\[
  M_\Pi(H)-B_2^{\mathrm{MO}}(K;\Hh)
  =
  (\nu_H-\nu_K)+(z_H-|\Pi|).
\]
The first term is algebraic symmetry release; the second is dissipative
release of central populations.

\subsection{Kraus-wise charge-zero symmetry}

A stronger symmetry convention requires the environment to carry no
symmetry charge.

\begin{definition}\label{def:kraus-zero}
A channel $\Phi$ is Kraus-wise charge zero when it admits a Kraus
representation
\[
  \Phi(X)=\sum_rA_rXA_r^\dagger
\]
such that
\[
  A_r\in\mathcal C_K
\]
for every $r$.
\end{definition}

Every Kraus-wise charge-zero channel is \(K\)-covariant, but the converse does not hold.  The term ``charge zero'' also excludes the adjacent convention in which every Kraus operator carries the same nontrivial character \cite{BucaProsen,DeGrootEtAl}.

\begin{proposition}[Common-character boundary]\label{prop:common-character}
Fix a unitary character \(\chi:K\to U(1)\), and suppose a trace-preserving channel has Kraus operators satisfying
\[
  A_r\pi(k)=\chi(k)\pi(k)A_r
  \qquad(k\in K)
\]
for every \(r\).  Let \(T_\chi\) be the induced permutation of the occurring isotypic labels, defined by
\[
  V_{T_\chi(\lambda)}\cong\chi^{-1}\otimes V_\lambda,
\]
and let \(\Pi_\chi\) be its orbit partition.  Then
\[
  \Phi^*(Z_\lambda)=Z_{T_\chi^{-1}(\lambda)},
  \qquad
  \Phi^*(Z_C)=Z_C\quad(C\in\Pi_\chi),
\]
and every symmetry-compatible automaton whose symbol channels carry this same character satisfies
\[
  \beta(\A)\leq D_K-|\Pi_\chi|+1.
\]
The upper bound need not be attained within the common-character subclass.
\end{proposition}

The proposition keeps the structural hierarchy sharp.  Common-character dynamics deterministically twist the isotypic labels and therefore conserve character orbits, whereas a general \(\Pi_\chi\)-conservative covariant channel may mix populations arbitrarily inside each orbit.  The proof and a four-label example with strict inequality are given in \cref{app:common-character-boundary}.

\begin{corollary}[Charge-zero channel capacity]
\label{cor:kraus-zero-capacity}
Let $B_{q,\mathrm{K0}}^{\mathrm{1g}}(K;\Hh)$ denote the maximal Hankel
rank over symmetry-compatible automata whose symbol channels are
Kraus-wise charge zero and whose alphabet has at most $q$ symbols.  Then
\[
  B_{q,\mathrm{K0}}^{\mathrm{1g}}(K;\Hh)
  =1+\nu_K
\]
for every $q\geq1$.
\end{corollary}

\begin{proof}
For every isotypic central projection $Z_\lambda$,
\[
  \Phi^*(Z_\lambda)
  =
  \sum_rA_r^\dagger Z_\lambda A_r
  =
  Z_\lambda\sum_rA_r^\dagger A_r
  =Z_\lambda.
\]
Thus all singleton components are conserved, and
\cref{thm:component-channel-capacity} gives the upper bound
$1+\nu_K$.

For saturation, work on each multiplicity space
$\mathcal M_\lambda$.  When $m_\lambda\geq2$, choose a nontrivial
projector $Q_\lambda$ and a basis
\[
  F_{\lambda,1},\ldots,F_{\lambda,m_\lambda^2-1}
\]
of $\Herm_0(\mathcal M_\lambda)$ with
$\Tr(Q_\lambda F_{\lambda,j})\neq0$.  Choose globally pairwise
distinct nonzero real numbers $\mu_{\lambda,j}$ and define
\[
  L_\lambda(I_{\mathcal M_\lambda}/m_\lambda)=0,
  \qquad
  L_\lambda(F_{\lambda,j})=\mu_{\lambda,j}F_{\lambda,j}.
\]
For a sufficiently small common $\varepsilon>0$, the Choi-interior
perturbation
\[
  \phi_\lambda
  =\mathcal R_{I/m_\lambda}+\varepsilon L_\lambda
\]
is completely positive and trace preserving.  Decreasing
$\varepsilon$ if necessary makes
$\{1\}\cup\{z_{\lambda,j}=\varepsilon\mu_{\lambda,j}\}_{\lambda,j}$
nonzero and globally pairwise distinct.  If
\[
  \phi_\lambda(X)
  =\sum_rB_{\lambda,r}XB_{\lambda,r}^\dagger,
\]
use the physical Kraus operators
\[
  A_{\lambda,r}
  =B_{\lambda,r}\otimes I_{V_\lambda}
\]
on the $\lambda$-isotypic summand, extended by zero elsewhere.  Their
union over $\lambda$ is trace preserving on $\Hh$, and every one lies
in $\mathcal C_K$.  For $m_\lambda=1$, use the identity channel on that
isotypic summand.

Choose positive isotypic weights $p_\lambda$ summing to one and small
$\delta_\lambda>0$ such that
\[
  \tau_\lambda
  =
  \frac{I_{\mathcal M_\lambda}}{m_\lambda}
  +\delta_\lambda
  \sum_{j=1}^{m_\lambda^2-1}F_{\lambda,j}
\]
is positive definite, with the sum empty when $m_\lambda=1$.  Take
\[
  \rho_0
  =
  \bigoplus_\lambda
  p_\lambda\tau_\lambda\otimes
  \frac{I_{V_\lambda}}{d_\lambda}
\]
and
\[
  P
  =
  \bigoplus_{\lambda:m_\lambda\geq2}
  Q_\lambda\otimes I_{V_\lambda}.
\]
If some $m_\lambda\geq2$, the unary acceptance sequence contains one
nonzero constant mode and all
\[
  \sum_\lambda(m_\lambda^2-1)=\nu_K
\]
distinct nonconstant modes, so its Hankel rank is $1+\nu_K$.  If every
$m_\lambda=1$, take $P=I$ to obtain the constant rank-one behavior.
\end{proof}

The smallest example separating the two channel notions is a qubit
with
\[
  K=U(1),
  \qquad
  \pi(\theta)=\operatorname{diag}(1,e^{i\theta}).
\]
Here
\[
  \mathcal C_K=\C\oplus\C,
  \qquad
  D_K=2,
  \qquad
  \nu_K=0.
\]
Kraus-wise charge-zero channels have diagonal Kraus operators and
preserve the two populations separately, so their symmetry-compatible
capacity is one.  In the ordered basis $(|0\rangle,|1\rangle)$, let
\[
  T=
  \begin{pmatrix}
    1-\alpha&\beta\\
    \alpha&1-\beta
  \end{pmatrix},
  \qquad
  0<\alpha,\beta,
  \quad
  \alpha+\beta<1,
\]
and use Kraus operators
\[
  A_{ji}=\sqrt{T_{ji}}\,|j\rangle\langle i|.
\]
Each $A_{ji}$ has a definite $U(1)$ charge, so the channel is
$U(1)$-covariant although its off-diagonal Kraus operators do not lie
in $\mathcal C_K$.  With $\rho_0=|0\rangle\langle0|$ and
$P=|0\rangle\langle0|$,
\[
  f_\A(a^t)
  =q+(1-q)(1-\alpha-\beta)^t,
  \qquad
  q=\frac{\beta}{\alpha+\beta},
\]
and hence $\beta(\A)=2=D_K$.  The released direction is the relative
population of the two isotypic sectors.

\begin{example}[Necessity of an invariant boundary]
\label{ex:boundary-sharpness}
The invariant-boundary hypothesis is essential.  For the same $U(1)$
representation, compose a generalized amplitude-damping channel with a
nontrivial rotation about the symmetry axis.  Its Bloch action may be
chosen as
\[
  (x+iy,z)
  \longmapsto
  \left(
    \sqrt\gamma e^{i\varphi}(x+iy),
    \gamma z+(1-\gamma)z_*
  \right),
\]
where $0<\gamma<1$ and $\varphi\notin\pi\mathbb Z$.  This channel is
$U(1)$-covariant.  For generic noninvariant initial and final rank-one projectors, the
acceptance sequence has the form
\[
  f_\A(a^t)
  =
  c_0+c_z\gamma^t
  +c_+\bigl(\sqrt\gamma e^{i\varphi}\bigr)^t
  +\overline{c_+}
   \bigl(\sqrt\gamma e^{-i\varphi}\bigr)^t,
\]
with $c_zc_+\neq0$.  The four modes are distinct, so
\cref{lem:exponential-hankel} gives Hankel rank four although $D_K=2$.
Thus channel covariance alone does not imply the commutant bound when
both boundaries break the symmetry.
\end{example}

\subsection{Strict-cutpoint probabilistic state cost}

The unary construction above determines numerical Hankel capacity.  A
probabilistic automaton need only reproduce the threshold language, so a
separate finite sign witness is required for a classical state lower
bound.  Let
\[
  \operatorname{SC}_{\Pi,\R}^{\mathrm{1g}}(\Hh)
\]
be the supremum, over arbitrary finite alphabets and all
symmetry-compatible, $\Pi$-conservative $K$-covariant automata on
$\Hh$, of the minimum number of states in a real PFA recognizing the
same strict-cutpoint language.

\begin{lemma}[Block-cyclic impulse system]
\label{lem:block-cyclic-impulse}
Let
\[
  V=\bigoplus_{j=1}^sV_j,
  \qquad
  d_j=\dim V_j\geq1,
  \qquad
  d=\sum_{j=1}^sd_j,
\]
and let \(\ell_j\in V_j^*\) be nonzero.  There are
\(A_j\in\operatorname{End}_{\mathbb R}(V_j)\) and \(v_j\in V_j\)
such that, for
\[
  A=\bigoplus_jA_j,
  \qquad
  v=\bigoplus_jv_j,
  \qquad
  \ell=\bigoplus_j\ell_j,
\]
the impulse sequence \(h_t=\ell(A^tv)\) satisfies
\[
  h_0=\cdots=h_{d-2}=0,
  \qquad
  h_{d-1}=1.
\]
\end{lemma}

\begin{proof}
Choose distinct real numbers \(\theta_1,\ldots,\theta_d\), partitioned
into sets \(I_j\) of sizes \(d_j\).  The Vandermonde system has a
unique solution \(w_1,\ldots,w_d\) to
\[
  \sum_{r=1}^dw_r\theta_r^t
  =
  \begin{cases}
    0,&0\leq t<d-1,\\
    1,&t=d-1.
  \end{cases}
\]
Every \(w_r\) is nonzero; explicitly,
\(w_r=\prod_{q\neq r}(\theta_r-\theta_q)^{-1}\).
On \(\mathbb R^{I_j}\), take the diagonal map with eigenvalues
\((\theta_r)_{r\in I_j}\), the vector with all coordinates one, and
the covector \((w_r)_{r\in I_j}\).  The last covector is nonzero, so an
isomorphism from \(V_j\) to \(\mathbb R^{I_j}\) can be chosen to carry
\(\ell_j\) to it.  Transporting the diagonal map and vector through
these isomorphisms gives the claimed block-diagonal system.
\end{proof}

Let
\(\operatorname{SC}_{\Pi,\mathbb R}^{\mathrm{1g},(2)}(\Hh)\)
denote the same worst-case state cost as
\(\operatorname{SC}_{\Pi,\mathbb R}^{\mathrm{1g}}(\Hh)\), restricted
to a binary input alphabet.

For a prescribed invariant projector \(P\), define
\(\operatorname{SC}_{\Pi,P,\mathbb R}^{\mathrm{1g},(2)}(\Hh)\) and
\(\operatorname{SC}_{\Pi,P,\mathbb R}^{\mathrm{1g}}(\Hh)\) analogously,
with and without the binary restriction.

\begin{theorem}[Binary prescribed-readout component law]
\label{thm:fixed-readout-component-pfa}
Suppose \(P\ne0\) and \(\mathcal A_{\Pi,P}\ne\varnothing\).  Then
\[
  M_{\Pi,P}
  \leq
  \operatorname{SC}_{\Pi,P,\mathbb R}^{\mathrm{1g},(2)}(\Hh)
  \leq
  \operatorname{SC}_{\Pi,P,\mathbb R}^{\mathrm{1g}}(\Hh)
  \leq
  M_{\Pi,P}+1.
\]
If \(P=0\) or \(\mathcal A_{\Pi,P}=\varnothing\), both state costs
equal one.  Thus every readout-visible component direction contributes
to an exact binary sign-rank obstruction, and stochastic positivity
costs at most one further state.
\end{theorem}

\begin{proof}
The upper bound follows from
\cref{thm:fixed-readout-component-capacity}: the component
representation contains the shared trace coordinate as a normalized
constant mode, so its centered rank is at most \(M_{\Pi,P}\).  Apply
\cref{thm:codimension-one-pfa}.  If \(P=0\), the language is empty.  If
the active set is empty, each \(P_C\) is zero or \(Z_C\), and component
conservation makes the acceptance probability constant; one state is
then exact.

Assume \(d=M_{\Pi,P}-1>0\).  For every component put
\[
  n_C=\Tr Z_C,
  \qquad
  \omega_C=\frac{Z_C}{n_C}.
\]
For each \(C\in\mathcal A_{\Pi,P}\), let
\[
  V_C=\{X\in\Herm(Z_C\mathcal C_K):\Tr X=0\}.
\]
Choose positive initial weights \(p_C\) on all components, summing to
one.  The functional \(\ell_C(X)=\Tr(P_CX)\) is nonzero on \(V_C\).
Apply
\cref{lem:block-cyclic-impulse} to the spaces \(V_C\) and these
functionals, obtaining \(A=\bigoplus_CA_C\) and
\(v=\bigoplus_Cv_C\).  Set \(A_C=v_C=0\) on every inactive component.
Extend every \(A_C\) complex linearly from its Hermitian domain.

For \(b\in\{0,1\}\), define on component \(C\)
\begin{align*}
  \phi_{C,b}(X)
  ={}&
  \Tr(X)\omega_C\\
  &+\varepsilon A_C
  \left(
    \mathcal T_K(X)-\Tr(X)\omega_C
  \right)
  +\varepsilon b\,\frac{\Tr(X)}{p_C}v_C,
\end{align*}
where \(A_C=v_C=0\) on inactive components.  The perturbation of the
replacement channel is Hermiticity preserving, trace annihilating,
and \(K\)-covariant.  Since the replacement channel has a
positive-definite Choi matrix, both \(\phi_{C,0}\) and
\(\phi_{C,1}\) are completely positive and trace preserving for all
sufficiently small \(\varepsilon>0\).  The direct-sum channels
\[
  \Phi_b(X)=\sum_{C\in\Pi}\phi_{C,b}(Z_CXZ_C)
\]
are therefore \(K\)-covariant and \(\Pi\)-conservative.

Start from \(\rho_0=\sum_Cp_C\omega_C\) and use the prescribed
projector \(P\).  On the direct sum of the active trace-zero
coordinates, the two symbols act as
\[
  x\longmapsto\varepsilon Ax+\varepsilon b v.
\]
For a prefix \(w=b_1\cdots b_d\) and suffix \(0^k\),
\begin{equation}
  f_\A(w0^k)-\Tr(P\rho_0)
  =
  \sum_{t=1}^d
  \varepsilon^{k+d-t+1}
  h_{k+d-t}b_t,
  \label{eq:block-cyclic-readout}
\end{equation}
where \(h_j=\ell(A^jv)\).  For \(0\leq k<d\), all terms with
\(t>k+1\) vanish and the term \(t=k+1\) equals
\(\varepsilon^db_{k+1}\).  The remaining terms involve only
\(h_d,\ldots,h_{2d-2}\).  Shrinking \(\varepsilon\) if necessary
gives, uniformly over all prefixes,
\[
  \left|
    \varepsilon^{-d}
    \bigl(f_\A(w0^k)-\Tr(P\rho_0)\bigr)
    -b_{k+1}
  \right|<\frac14.
\]
The same choice makes the normalized deviation for the suffix \(0^d\)
smaller than \(1/4\) in absolute value.  At the cutpoint
\[
  \tau=\Tr(P\rho_0)+\frac{\varepsilon^d}{2},
\]
the suffixes \(0^0,0^1,\ldots,0^{d-1}\) read the signs
\(2b_1-1,\ldots,2b_d-1\), while \(0^d\) is negative for every
prefix.  The resulting finite sign matrix is the transpose of
\(H_d^{(-1)}\), whose sign-rank is \(d+1=M_{\Pi,P}\).  Hence every real
PFA recognizing the same strict-cutpoint language has at least
\(M_{\Pi,P}\) states by \cref{lem:pfa-cutpoint-rank}.
\end{proof}

\begin{theorem}[Dynamic prescribed-readout component law]
\label{thm:dynamic-fixed-readout-component-pfa}
Suppose \(P\ne0\), \(\mathcal A_{\Pi,P}\ne\varnothing\), and some
active component \(C_*\in\mathcal A_{\Pi,P}\) has
\(\nu_{C_*}>0\).  Then the unrestricted-alphabet state cost is the
upper endpoint:
\[
  \operatorname{SC}_{\Pi,P,\mathbb R}^{\mathrm{1g}}(\Hh)
  =M_{\Pi,P}+1.
\]
Writing \(s=|\mathcal A_{\Pi,P}|\), a witness needs at most
\(2s+3\) letters.  If \(P_{C_*}=PZ_{C_*}\) is noncentral in
\(Z_{C_*}\mathcal C_K\), at most \(2s+2\) letters suffice.
\end{theorem}

\begin{proof}
The upper bound is already contained in
\cref{thm:fixed-readout-component-pfa}.  We adapt the dynamic
shattering construction of \cite{ChenWuQuadratic} to the fixed-weight
symmetry-reduced slice; the details below identify the two points at
which the component constraint matters.

For \(C\in\mathcal A_{\Pi,P}\), put
\[
 V_C=\{X\in\Herm(Z_C\mathcal C_K):\Tr X=0\},
 \qquad d_C=\dim V_C=D_C-1,
\]
and choose positive component weights \(p_C\).  Relative to
\(\omega_C=Z_C/\Tr Z_C\), the reachable affine slice is the direct
product of the \(p_C\omega_C+V_C\), of total affine dimension
\[
 d=\sum_{C\in\mathcal A_{\Pi,P}}d_C=M_{\Pi,P}-1.
\]
We repeatedly use the following Choi-interior fact.  If \(L_C\) is
Hermiticity preserving and trace annihilating and \(v_C\in V_C\),
then, for all sufficiently small \(\epsilon>0\),
\begin{equation}
 \Phi_{C,L,v}(X)
 =\Tr(X)\omega_C
  +\epsilon L_C\!\left(\mathcal T_K(X)-\Tr(X)\omega_C\right)
  +\epsilon\Tr(X)v_C
 \label{eq:choi-interior-affine}
\end{equation}
is completely positive, trace preserving, and \(K\)-covariant.  This
follows because the replacement term has positive-definite Choi
matrix, while the perturbation is trace annihilating.  Taking the
direct sum over components preserves every \(Z_C\).  Consequently any
finite list of sufficiently contracted blockwise real affine maps on
the displayed slice can be implemented by \(\Pi\)-conservative
channels.

For each active component, the four-adic orthant-coding lemma of
\cite{ChenWuQuadratic} supplies \(R_C\in O(d_C)\) and
\(u_C\in\mathbb R^{d_C}\) such that, for every
\(\eta_C\in\{\pm1\}^{d_C}\), some \(N_C(\eta_C)\geq0\) satisfies
\[
 \operatorname{sign}\bigl(R_C^{N_C(\eta_C)}u_C\bigr)=\eta_C.
\]
The functional \(\zeta_C(X)=\Tr(P_CX)\) is nonzero on \(V_C\).
Since \(\nu_{C_*}>0\), choose traceless Hermitian \(X,Y\) in a
two-dimensional multiplicity corner and a commutant unitary \(U\) for
which \(\Ad_U\) rotates \(\spanR\{X,Y\}\) through an angle
\(\theta\) with \(\theta/(2\pi)\notin\mathbb Q\).  If
\(P_{C_*}\) is noncentral, the corner basis may be chosen so that
\(\zeta_{C_*}\) is nonzero on both \(X\) and \(Y\).  If it is central,
choose an invertible real map \(G:V_{C_*}\to V_{C_*}\) such that
\(\zeta_{C_*}(GX)\) and \(\zeta_{C_*}(GY)\) are nonzero, implement its
contracted version as one decoder letter \(g\) by
\eqref{eq:choi-interior-affine}, and set
\(\widetilde\zeta_{C_*}=\zeta_{C_*}\circ G\).  Otherwise set
\(\widetilde\zeta_C=\zeta_C\).
On the other active components, \(g\) acts by a positive identity
contraction, so it preserves all previously chosen signs.

After a rotation and common scaling in the first coding plane, choose
\(T_{C_*}:\mathbb R^{d_{C_*}}\to V_{C_*}\) with
\[
 T_{C_*}e_1=X,\qquad T_{C_*}e_2=Y,
 \qquad
 \widetilde\zeta_{C_*}(T_{C_*}x)=u_{C_*}^{\mathsf T}x.
\]
For every other component choose an isomorphism \(T_C\) with
\(\zeta_C(T_Cx)=u_C^{\mathsf T}x\).  Let \(J_C\) be the nilpotent
shift on this coordinate space.  Equation
\eqref{eq:choi-interior-affine} implements, up to positive scalar
contractions, the reduced maps
\(T_CJ_CT_C^{-1}\) and \(T_CR_C^{\mathsf T}T_C^{-1}\).
Use them simultaneously as a prefix letter \(p\) and one tester
\(t_C\) per component.  For every \(C\ne C_*\), a selector \(s_C\)
resets all perturbations and injects \(T_Ce_1\) in component \(C\).
Starting with a small \(T_{C_*}e_1\) perturbation, the prefixes
\[
 p^j\quad(C=C_*),
 \qquad
 s_Cp^j\quad(C\ne C_*),
 \qquad 0\leq j<d_C,
\]
supply one row for each of the \(d\) affine coordinates, while a
sufficiently high power of \(p\) supplies the baseline anchor.
Independent powers of the \(t_C\)'s realize every prescribed
coordinate sign vector; in the central case each test word ends with
\(g\).  The finitely many deviations are nonzero;
choose a cutpoint \(\tau\) a positive distance \(\delta\) above the
baseline, smaller than all their absolute margins.

Choose \(\gamma>0\) so small that
\[
 H_C=\tau(1+\gamma)Z_C-\gamma P_C
 \qquad(C\in\mathcal A_{\Pi,P})
\]
satisfies \(0<H_C<Z_C\).  Measuring \(\{H_C,Z_C-H_C\}\) inside each
component and preparing a normalized state in \(P_C\) or
\(Z_C-P_C\), respectively, defines a \(K\)-covariant,
\(\Pi\)-conservative channel \(h\) for which, on the fixed-weight
slice,
\[
 f(wh)-\tau=-\gamma\bigl(f(w)-\tau\bigr).
\]
Append \(h\) and reverse all requested coordinate signs whenever the
baseline anchor must be made positive.  This gives the one-sided
\(H_d^{(+1)}\) relations with a common positive anchor.

It remains to use the clock already built into the first coding plane.
Let \(a\) be the unitary channel of \(U\) and take a carrier with a
sufficiently small \(X\)-perturbation.

For every four-adic test suffix the centered carrier value is then a
sinusoid whose amplitude is the positive contraction factor times the
norm of the first-plane projection of
\(R_{C_*}^{N_{C_*}}u_{C_*}\), hence is nonzero.  Because there are only
finitely many suffixes, decrease \(\delta\) so it is smaller than every
such amplitude as well as every static margin.  Irrationality makes
each subsequence \(qn\theta\) dense modulo \(2\pi\), and \(h\) only
reflects the centered sinusoid.  Thus every test crosses the cutpoint
in both directions infinitely often along every arithmetic
subsequence.  The dynamic affine lift,
\cref{lem:dynamic-affine-lift}, gives
the lower bound \(d+2=M_{\Pi,P}+1\).  The letters are \(a,p,h\), the
\(s\) testers, the \(s-1\) selectors, and, only in the central-readout
case, \(g\), giving the stated counts.
\end{proof}

\begin{corollary}[Binary component-conserving state law]
\label{thm:component-pfa}
With \(M_\Pi=D_K-|\Pi|+1\),
\[
  M_\Pi
  \leq
  \operatorname{SC}_{\Pi,\mathbb R}^{\mathrm{1g},(2)}(\Hh)
  \leq
  \operatorname{SC}_{\Pi,\mathbb R}^{\mathrm{1g}}(\Hh)
  \leq
  M_\Pi+1.
\]
If \(M_\Pi=1\), both state costs equal one.
If \(M_\Pi>1\) and some component has \(\nu_C>0\), then
\[
  \operatorname{SC}_{\Pi,\mathbb R}^{\mathrm{1g}}(\Hh)
  =M_\Pi+1.
\]
\end{corollary}

\begin{proof}
When \(M_\Pi>1\), choose \(P_C\) nontrivial on every component with
\(D_C>1\).  Then \(M_{\Pi,P}=M_\Pi\), and
\cref{thm:fixed-readout-component-pfa} gives the lower bound.  The
upper bound holds for every readout by
\cref{thm:component-channel-capacity,thm:codimension-one-pfa}.  If
\(M_\Pi=1\), all component algebras are one-dimensional and every
acceptance function is constant.
When some \(\nu_C>0\), choose the preceding readout so that this
component is active and apply
\cref{thm:dynamic-fixed-readout-component-pfa}; its lower bound meets
the universal upper bound.
\end{proof}

The singleton partition admits the same obstruction without charged
Kraus operators.  Write
\[
  P=\bigoplus_\lambda Q_\lambda\otimes I_{V_\lambda},
  \qquad
  \Lambda_P=\{\lambda:0<Q_\lambda<I_{\mathcal M_\lambda}\},
\]
and, for \(P\ne0\), set
\[
  M_{\mathrm{K0},P}
  =
  1+\sum_{\lambda\in\Lambda_P}(m_\lambda^2-1).
\]
Write \(\operatorname{SC}_{\mathrm{K0},P,\mathbb R}^{\mathrm{1g},(2)}\)
for the worst binary fixed-\(P\) cost within the Kraus-wise charge-zero
class; omit the superscript \((2)\) for unrestricted finite alphabets,
and omit \(P\) when maximizing over invariant readouts.

\begin{corollary}[Charge-zero state law]
\label{cor:kraus-zero-pfa}
If \(P\ne0\) and \(\Lambda_P\ne\varnothing\), then
\[
  M_{\mathrm{K0},P}
  \leq
  \operatorname{SC}_{\mathrm{K0},P,\mathbb R}^{\mathrm{1g},(2)}(\Hh)
  \leq
  M_{\mathrm{K0},P}+1.
\]
Moreover,
\[
  \operatorname{SC}_{\mathrm{K0},P,\mathbb R}^{\mathrm{1g}}(\Hh)
  =M_{\mathrm{K0},P}+1.
\]
If \(P=0\) or \(\Lambda_P=\varnothing\), the state cost is one.
Maximizing over invariant readouts gives, for \(\nu_K>0\),
\[
  1+\nu_K
  \leq
  \operatorname{SC}_{\mathrm{K0},\mathbb R}^{\mathrm{1g},(2)}(\Hh)
  \leq
  \nu_K+2,
\]
and
\[
  \operatorname{SC}_{\mathrm{K0},\mathbb R}^{\mathrm{1g}}(\Hh)
  =\nu_K+2,
\]
whereas \(\nu_K=0\) gives state cost one.
\end{corollary}

\begin{proof}
Use the singleton components in the construction of
\cref{thm:fixed-readout-component-pfa}, but perform each Choi-interior
perturbation directly on \(\operatorname{End}(\mathcal M_\lambda)\).
If its two multiplicity channels have Kraus operators
\(B_{\lambda,b,r}\), lift them physically as
\[
  B_{\lambda,b,r}\otimes I_{V_\lambda}.
\]
Every lifted Kraus operator lies in \(\mathcal C_K\), so the binary
witness is Kraus-wise charge zero and realizes
\(H_{M_{\mathrm{K0},P}-1}^{(-1)}\).  The upper bound follows from the
fixed-readout singleton capacity and
\cref{thm:codimension-one-pfa}.  Choosing every \(Q_\lambda\)
nontrivial when \(m_\lambda\geq2\) yields the worst-readout formula.
For singleton components, every channel in the dynamic construction of
\cref{thm:dynamic-fixed-readout-component-pfa} can likewise be
implemented on \(\mathcal M_\lambda\) and lifted with Kraus operators
\(B\otimes I_{V_\lambda}\).  It is therefore Kraus-wise charge zero
and gives the two unrestricted-alphabet equalities.
When no such multiplicity block exists, charge-zero dynamics preserve
all invariant populations and every invariant-readout behavior is
constant.
\end{proof}

For one full mobility component, the prescribed-readout theorem gives
the following uniform law.

\begin{theorem}[Full-mobility noncommutativity dichotomy]
\label{thm:binary-channel-shift}
Assume $D_K\geq2$ and fix a nontrivial invariant orthogonal projector
$0<P<I$.  Let
\[
\operatorname{SC}_{K,P,\R}^{\mathrm{1g},(2)}(\Hh)
\]
be the worst strict-cutpoint real-PFA state cost over binary
$K$-covariant general-channel automata on $\Hh$ with accepting
projector $P$, without an additional component-conservation constraint.
Use \(\operatorname{SC}_{K,P,\R}^{\mathrm{1g}}(\Hh)\) for arbitrary
finite alphabets.
Then
\[
  D_K
  \leq
  \operatorname{SC}_{K,P,\R}^{\mathrm{1g},(2)}(\Hh)
  \leq
  D_K+1.
\]
If \(\nu_K>0\), then
\[
  \operatorname{SC}_{K,P,\R}^{\mathrm{1g}}(\Hh)=D_K+1;
\]
a witness uses at most five letters, or four when \(P\) is noncentral
in \(\mathcal C_K\).  If \(\nu_K=0\), then instead
\[
  \operatorname{SC}_{K,P,\R}^{\mathrm{1g},(2)}(\Hh)
  =
  \operatorname{SC}_{K,P,\R}^{\mathrm{1g}}(\Hh)
  =D_K.
\]
\end{theorem}

\begin{proof}
For the one-component partition, the nontrivial projector \(P\) is
active and \(M_{\Pi,P}=D_K\).  Apply
\cref{thm:fixed-readout-component-pfa}.  If \(\nu_K>0\),
\cref{thm:dynamic-fixed-readout-component-pfa} gives the matching
unrestricted lower bound and the alphabet counts with \(s=1\).

If \(\nu_K=0\), then
\(\mathcal C_K\cong\mathbb C^{D_K}\).  On invariant states every
covariant channel is a stochastic map of the \(D_K\) minimal central
populations, and every invariant projector reads a subset of those
populations.  Thus the reduced automaton is itself a \(D_K\)-state
PFA, while the binary lower bound above is \(D_K\).
\end{proof}

The two branches are qualitatively different worst-case memory laws.  When $\mathcal C_K$ is commutative, every symmetry-reduced state is a probability vector over minimal central sectors and every covariant channel acts stochastically on those coordinates, so the full-mobility class has worst state cost $D_K$.  For the fixed nontrivial invariant readout of the theorem, a noncommutative multiplicity block raises the unrestricted worst-case cost to $D_K+1$.  The universal upper bound makes this increase exact: for the full-mobility class with a fixed nontrivial invariant readout, noncommutativity has a one-state classical price.

\begin{corollary}[Fixed-rank general-channel law]
\label{cor:binary-fixed-rank-channel}
For $N\geq2$ and $1\leq r\leq N-1$, let
$\operatorname{SC}_{\R}^{\mathrm{1g},(2)}(N,r)$ be the worst
strict-cutpoint real-PFA state cost over binary $N$-dimensional
general-channel automata whose accepting projector has rank $r$.  Then
\[
  N^2
  \leq
  \operatorname{SC}_{\R}^{\mathrm{1g},(2)}(N,r)
  \leq
  N^2+1.
\]
Without the binary restriction,
\[
  \operatorname{SC}_{\R}^{\mathrm{1g}}(N,r)=N^2+1,
\]
and four input letters suffice.  Thus every nontrivial fixed rank,
including rank one, retains the exact quadratic-plus-one state cost.
\end{corollary}

\begin{proof}
Apply \cref{thm:binary-channel-shift} to the trivial symmetry group.
Then $D_K=N^2$, and every rank-$r$ projector with
$1\leq r\leq N-1$ is nontrivial and noncentral.  The noncommutative branch of
that theorem gives the unrestricted equality and the four-letter
count.
\end{proof}

For \(r=1\), the unrestricted equality recovers the quadratic state
law of \cite{ChenWuQuadratic}.  The extension to every
\(1\leq r<N\) uses only that the trace-zero part of the prescribed
readout is nonzero: the Choi-interior tester can align any such
functional with the four-adic dynamic coordinates, so its rank does
not alter the obstruction.

The block-cyclic binary witness also quantifies its finite tests.  For
a fixed readout, prefixes and suffixes have length at most
\(M_{\Pi,P}-1\), and every tested value is separated from the cutpoint
by at least \(\varepsilon^{M_{\Pi,P}-1}/4\).  The component,
charge-zero, and full-mobility corollaries inherit this linear word
length.  These are strict-cutpoint margins and may decrease with
dimension.

\section{Representation-theoretic consequence: Schur--Weyl inversion}

A structural law should predict genuinely different regimes, not merely re-express one quadratic bound.  Schur--Weyl duality provides a sharp test: two mutual commutants act on the same tensor-power Hilbert space, yet the center--commutator law assigns them polynomial and exponential worst-case memory scales.  The contrast isolates symmetry, rather than ambient quantum dimension, as the source of the change.

\subsection{Permutation-equivariant automata}

Let
\[
  \Hh_{n,d}=(\C^d)^{\otimes n}.
\]
Schur--Weyl duality gives
\cite{FultonHarris}
\[
  \Hh_{n,d}
  \cong
  \bigoplus_{\substack{\lambda\vdash n\\\ell(\lambda)\leq d}}
  \mathcal U_\lambda^{(d)}\otimes\mathcal V_\lambda,
\]
where $\mathcal U_\lambda^{(d)}$ is an irreducible $U(d)$-module and
$\mathcal V_\lambda$ is the Specht module of $S_n$.  Write
\[
  m_\lambda^{(d)}=\dim\mathcal U_\lambda^{(d)},
  \qquad
  f^\lambda=\dim\mathcal V_\lambda,
\]
and let
\[
  p_d(n)
  =
  \#\{\lambda\vdash n:\ell(\lambda)\leq d\}.
\]

If all transitions and the accepting projector commute with the
permutation action, then they act on the $U(d)$ factors.  The commutant
is
\[
  \operatorname{End}_{S_n}(\Hh_{n,d})
  \cong
  \bigoplus_{\ell(\lambda)\leq d}
  \operatorname{End}(\mathcal U_\lambda^{(d)}).
\]
Its dimension has the exact form
\[
  \sum_{\ell(\lambda)\leq d}
  \bigl(m_\lambda^{(d)}\bigr)^2
  =
  \binom{n+d^2-1}{d^2-1}.
\]
Indeed,
$\operatorname{End}(\Hh_{n,d})\cong
\operatorname{End}(\C^d)^{\otimes n}$, and conjugation by a permutation
permutes these $n$ tensor factors.  The fixed subspace is
$\operatorname{Sym}^n(\operatorname{End}(\C^d))$, whose dimension is
the displayed binomial coefficient.

\begin{theorem}[Permutation-equivariant state capacity]\label{thm:schur-permutation}
The maximal Hankel rank of an $S_n$-equivariant automaton on
$\Hh_{n,d}$ is
\[
  \mathsf B_{S_n}(n,d)
  =
  1+
  \binom{n+d^2-1}{d^2-1}
  -p_d(n).
\]
For fixed $d\geq2$,
\[
  \mathsf B_{S_n}(n,d)=\Theta(n^{d^2-1}),
\]
and the worst binary strict-cutpoint probabilistic simulation cost has
the same order.
\end{theorem}

\begin{proof}
The center of the direct-sum commutant has one scalar coordinate for
each allowed partition, hence dimension $p_d(n)$.  \Cref{thm:commutant-law} gives the exact formula.  Since $p_d(n)=O(n^{d-1})$ for fixed $d$, it is
lower order than the binomial coefficient.  The binary probabilistic
state law then gives the final assertion.
\end{proof}

This state-count result is distinct from time-complexity simulation
results, but it is consistent with the polynomial Schur-basis size that
underlies efficient classical algorithms for sufficiently restrictive
permutation symmetry \cite{AnschuetzEtAl}.

\subsection{Collective-unitary-equivariant automata}

If all transitions and the accepting projector commute instead with
$U^{\otimes n}$ for every $U\in U(d)$, they act on the Specht factors.
The commutant is
\[
  \operatorname{End}_{U(d)}(\Hh_{n,d})
  \cong
  \bigoplus_{\ell(\lambda)\leq d}
  \operatorname{End}(\mathcal V_\lambda).
\]
Set
\[
  A_d(n)
  =
  \sum_{\substack{\lambda\vdash n\\\ell(\lambda)\leq d}}
  (f^\lambda)^2.
\]
By the Robinson--Schensted correspondence, $A_d(n)$ also counts
permutations whose longest decreasing subsequence has length at most
$d$ \cite{StanleyEC2}.  Regev's fixed-strip asymptotics give
\cite{RegevStrips}
\[
  A_d(n)
  =
  \Theta\!\left(
    d^{2n}n^{-(d^2-1)/2}
  \right)
\]
for fixed $d\geq2$.

\begin{theorem}[Collective-unitary state capacity]\label{thm:schur-collective}
The maximal Hankel rank of a collective-$U(d)$-equivariant automaton is
\[
  \mathsf B_{U(d)}(n,d)
  =
  1+A_d(n)-p_d(n).
\]
For fixed $d\geq2$,
\[
  \mathsf B_{U(d)}(n,d)
  =
  \Theta\!\left(
    d^{2n}n^{-(d^2-1)/2}
  \right),
\]
and the worst binary strict-cutpoint probabilistic simulation cost has
the same order.
\end{theorem}

\begin{proof}
The commutant dimension is $A_d(n)$ and its center again has dimension
$p_d(n)$.  The exact formula follows from \cref{thm:commutant-law}.  The center
term is polynomial in $n$ and therefore negligible relative to Regev's
asymptotic expression.  Apply the binary probabilistic state law.
\end{proof}

\subsection{The inversion on a common Hilbert space}

Combining the two sides gives
\[
  \operatorname{SC}^{(2)}_{S_n}(n,d)
  =
  \Theta(n^{d^2-1})
\]
and
\[
  \operatorname{SC}^{(2)}_{U(d)}(n,d)
  =
  \Theta\!\left(
    d^{2n}n^{-(d^2-1)/2}
  \right).
\]
Therefore
\[
  \frac{\operatorname{SC}^{(2)}_{U(d)}(n,d)}
       {\operatorname{SC}^{(2)}_{S_n}(n,d)}
  =
  \Theta\!\left(
    \frac{d^{2n}}{n^{\frac32(d^2-1)}}
  \right).
\]
The physical Hilbert space, tensor power, and ambient quantum dimension are identical on the two sides.  The two automaton classes are the distinct, generally non-nested classes associated with the two members of the same Schur--Weyl dual pair, and each displayed scale is a maximum over its own class.  The polynomial--exponential inversion therefore isolates which symmetry is preserved, rather than ambient Hilbert-space dimension, as the source of the change in worst-case behavioral memory.

For qubits, $p_2(n)=\lfloor n/2\rfloor+1$, and the formulas become
\[
  \mathsf B_{S_n}(n,2)
  =
  \binom{n+3}{3}-\left\lfloor\frac n2\right\rfloor,
  \qquad
  \operatorname{SC}^{(2)}_{S_n}(n,2)=\Theta(n^3),
\]
whereas
\[
  \mathsf B_{U(2)}(n,2)
  =
  C_n-\left\lfloor\frac n2\right\rfloor,
  \qquad
  \operatorname{SC}^{(2)}_{U(2)}(n,2)
  =
  \Theta\!\left(\frac{4^n}{n^{3/2}}\right).
\]
Equivalently, indexing one internal state of a worst-case simulator requires $3\log_2 n+O(1)$ bits on the permutation-equivariant side and $2n-\frac32\log_2 n+O(1)$ bits on the collective-unitary side.  Thus the state-count inversion becomes a logarithmic-versus-linear separation in the classical memory bits needed to index the simulator states.

\section{Readout geometry and symmetry release}

Structural capacity counts every movable operator coordinate, but operational realization must pass through a fixed accepting boundary.  For reversible conjugation the orbit of that projector identifies the first-order visible component, its Fisher geometry gives the same intrinsic dimension, and subgroup branching quantifies how both total and visible memory change when symmetry is released.

\subsection{The accepting-orbit capacity}

Let \(G\) be a compact Lie group with Lie algebra \(\mathfrak g\),
represented unitarily on \(\Hh\).  We use anti-Hermitian Lie algebra
elements.

\begin{definition}
The accepting-orbit capacity is
\[
  \kappa_G(P)
  =
  \rank
  \left(
  \mathfrak g\longrightarrow\Herm(\Hh),
  \quad X\longmapsto[X,P]
  \right).
\]
\end{definition}

Let
\[
  \mathfrak g_P
  =
  \{X\in\mathfrak g:[X,P]=0\}.
\]

\begin{proposition}\label{prop:orbit-capacity}
The capacity is the orbit dimension:
\[
  \kappa_G(P)
  =
  \dim\mathfrak g-\dim\mathfrak g_P
  =
  \dim(G\cdot P).
\]
Under independent full sector control,
\[
  \kappa_G(P)
  =
  \kappa(\mathbf D,\mathbf r)
  :=
  \sum_\alpha2r_\alpha q_\alpha.
\]
\end{proposition}

\begin{proof}
The differential of the conjugation orbit map at the identity is
\(X\mapsto[X,P]\), whose kernel is \(\mathfrak g_P\).  Rank--nullity
gives the first identity.  In a \(D_\alpha\)-dimensional sector the
stabilizer of a rank-\(r_\alpha\) projector is
\[
  S\bigl(U(r_\alpha)\times U(q_\alpha)\bigr),
\]
so the orbit dimension is
\[
  (D_\alpha^2-1)
  -
  (r_\alpha^2+q_\alpha^2-1)
  =
  2r_\alpha q_\alpha.
\]
The dimensions add over independently controlled sectors.
\end{proof}

Recall the fixed-readout decomposition
\[
  P=\bigoplus_\lambda Q_\lambda\otimes I_{V_\lambda},
  \qquad
  r_\lambda=\rank Q_\lambda,
  \qquad
  \Lambda_P=\{\lambda:0<r_\lambda<m_\lambda\}.
\]
Set
\[
  \nu_P=\sum_{\lambda\in\Lambda_P}(m_\lambda^2-1),
  \qquad
  \varrho_P=\sum_{\lambda\in\Lambda_P}(m_\lambda-1),
  \qquad
  \delta_P=\mathbf 1_{\{P\ne0\}}.
\]
For the maximal \(K\)-commuting reversible class, write
\[
  G_K=\prod_\lambda SU(\mathcal M_\lambda).
\]

\begin{theorem}[Prescribed-readout reversible capacity]
\label{thm:fixed-readout-reversible-capacity}
Let \(B_{q,P}^{\mathrm{MO}}(K;\Hh)\) be the maximal Hankel rank of a
\(K\)-commuting measure-once automaton whose accepting projector is
the prescribed \(P\), and let
\(\operatorname{SC}_{K,P,\mathbb R}^{\mathrm{MO},(2)}(\Hh)\) be the
corresponding worst binary strict-cutpoint real-PFA cost.  Then
\[
  B_{1,P}^{\mathrm{MO}}(K;\Hh)
  =
  \delta_P+\nu_P-\varrho_P
  =
  \delta_P+
  \sum_{\lambda\in\Lambda_P}m_\lambda(m_\lambda-1),
\]
whereas, for every \(q\geq2\),
\[
  B_{q,P}^{\mathrm{MO}}(K;\Hh)=\delta_P+\nu_P.
\]
If \(\Lambda_P\ne\varnothing\), the binary strict-cutpoint cost obeys
\[
  2+\kappa_{G_K}(P)
  \leq
  \operatorname{SC}_{K,P,\mathbb R}^{\mathrm{MO},(2)}(\Hh)
  \leq
  \nu_P+2.
\]
In particular, one active block with $m_\lambda=2$ and $r_\lambda=1$
satisfies
\[
  4=2+\kappa_{G_K}(P)
  \leq
  \operatorname{SC}_{K,P,\mathbb R}^{\mathrm{MO},(2)}(\Hh)
  \leq
  5=2+\nu_P.
\]
If \(\Lambda_P=\varnothing\), every reversible behavior with this
readout is constant and the state cost is one.
\end{theorem}

\begin{proof}
Blocks with \(Q_\lambda=0\) or \(I\) contribute only to the shared
constant series.  In a partially split block, unary conjugation has at
most \(m_\lambda(m_\lambda-1)\) nonzero Bohr-frequency modes, since all
diagonal directions lie in the common zero-frequency eigenspace.  A
generic eigenbasis makes every matrix entry of the fixed
\(Q_\lambda\) nonzero, and generic eigenphases make all nonzero
frequency differences globally distinct.  Choosing a state with all
corresponding entries nonzero and applying
\cref{lem:exponential-hankel} attains the unary formula.

With two letters, dense independent \(SU(m_\lambda)\) control spans
the full trace-zero Hermitian block for every
\(\lambda\in\Lambda_P\), and no other block is observable.  This gives
the binary capacity.  The lower state bound is the affine-anchored
prepare--test construction, for which
\(\kappa_{G_K}(P)=\sum_{\lambda\in\Lambda_P}2r_\lambda(m_\lambda-r_\lambda)\).
The dynamic affine lift supplies the additional state for every
nontrivial rank profile.  When all active ranks are one,
\cref{thm:curved-witness} gives the same lower bound through an
independent finite curved certificate.
The upper bound follows from the constant-normalized capacity
representation and \cref{thm:codimension-one-pfa}.  When the active set
is empty, \(P\) is central and hence fixed by every commuting-unitary
conjugation.
\end{proof}

The orbit dimension counts the independent ways in which the readout can
move under the available control group.  The same tangent space also has
an intrinsic statistical interpretation, so the readout capacity can be
identified without reference to a particular circuit parametrization.

\subsection{Intrinsic Fisher geometry of the readout orbit}

Let
\[
  R_P=\Tr P,
  \qquad
  \sigma_P=\frac{P}{R_P}.
\]
Assume throughout this subsection that $P\neq0$.
We use the symmetric-logarithmic-derivative convention
\[
  g_\rho^{\mathrm{SLD}}(A,B)
  =
  2
  \sum_{j,k:p_j+p_k>0}
  \frac{\operatorname{Re}(A_{jk}B_{kj})}{p_j+p_k}
\]
for \(\rho=\sum_jp_j|j\rangle\langle j|\)
\cite{BraunsteinCaves,PetzGhinea}.

\begin{theorem}[Orbit capacity equals intrinsic Fisher rank]\label{thm:fisher-orbit}
For tangent projector variations
\[
  \dot P_X=[X,P],
  \qquad
  \dot P_Y=[Y,P],
\]
one has
\[
  g_{\sigma_P}^{\mathrm{SLD}}
  \left(
  \frac{\dot P_X}{R_P},
  \frac{\dot P_Y}{R_P}
  \right)
  =
  \frac{2}{R_P}
  \Tr(\dot P_X\dot P_Y).
\]
Consequently,
\[
  \rank
  g_{\sigma_P}^{\mathrm{SLD}}
  =
  \kappa_G(P).
\]
In a Hilbert--Schmidt orthonormal tangent frame, every nonzero
eigenvalue is \(2/R_P\).
\end{theorem}

\begin{proof}
See \cref{app:fisher-proof}.
\end{proof}

This is an intrinsic orbit statement.  A circuit-coordinate quantum
Fisher matrix has the same rank only when the circuit differential
surjects onto the orbit tangent space.

The Fisher-rank identity shows that \(\kappa\) is an intrinsic dimension of
the measurement orbit rather than a coordinate count.  Comparing it
with \(\mathsf B\) reveals precisely which behavior directions can be exposed
by infinitesimal readout motion and which remain invisible to the
binary measurement.

\subsection{Visible and dark behavior directions}

Under independent full sector control,
\[
  \mathsf B(\mathbf D)-1
  =
  \sum_\alpha(D_\alpha^2-1)
\]
is the full traceless behavior mass, whereas
\[
  \kappa(\mathbf D,\mathbf r)
  =
  \sum_\alpha2r_\alpha q_\alpha
\]
is visible to infinitesimal accepting-subspace motion.  Their
difference is
\[
  \mathsf B(\mathbf D)-1-\kappa(\mathbf D,\mathbf r)
  =
  \sum_\alpha
  \bigl(r_\alpha^2+q_\alpha^2-1\bigr).
\]

\begin{lemma}[Readout curvature]\label{lem:readout-curvature}
Let $P$ be a rank-$r$ orthogonal projector on $\mathbb C^D$, put
$q=D-r$, and define the off-diagonal normal space
\[
  \mathfrak n_P
  =
  \{N\in\mathfrak{su}(D):PNP=0=(I-P)N(I-P)\}.
\]
Then
\[
  \spanR\{[[P,N],N]:N\in\mathfrak n_P\}
  =
  \{X\in\Herm_0(\mathbb C^D):[X,P]=0\},
\]
whose real dimension is $r^2+q^2-1$.  Moreover,
\[
  \spanR\{[P,X]:X\in\mathfrak{su}(D)\}
  \oplus
  \spanR\{[[P,N],N]:N\in\mathfrak n_P\}
  =
  \Herm_0(\mathbb C^D).
\]
\end{lemma}

\begin{proof}
Relative to
$\mathbb C^D=\operatorname{ran}P\oplus\ker P$, write
\[
  N=
  \begin{pmatrix}0&B\\-B^\dagger&0\end{pmatrix}.
\]
Then
\[
  [P,N]=
  \begin{pmatrix}0&B\\B^\dagger&0\end{pmatrix},
  \qquad
  [[P,N],N]
  =
  \begin{pmatrix}-2BB^\dagger&0\\0&2B^\dagger B\end{pmatrix}.
\]
The second operator is traceless Hermitian and commutes with $P$.
Conversely, taking $B=uv^\dagger$ with unit vectors $u$ and $v$ gives
$\operatorname{diag}(-2uu^\dagger,2vv^\dagger)$.  Differences at fixed
$u$ span $0\oplus\Herm_0(\mathbb C^q)$, and differences at fixed $v$
span $\Herm_0(\mathbb C^r)\oplus0$.  One remaining element supplies the
relative scalar direction
$\operatorname{diag}(-I_r/r,I_q/q)$.  These spaces have total dimension
$(r^2-1)+(q^2-1)+1=r^2+q^2-1$ and exhaust the traceless Hermitian
commutant of $P$.  The first span is the off-diagonal Hermitian space,
so the final direct sum follows from
$2rq+r^2+q^2-1=D^2-1$.
\end{proof}

Thus the visible directions are the accepting--rejecting coherences,
whereas the dark directions are accepting--accepting and
rejecting--rejecting coherences together with the traceless
block-diagonal population direction.  The lemma shows that these dark
directions are precisely the quadratic curvature span of the readout
orbit.  The witness in \cref{thm:curved-witness} converts one common
radial combination of this span into an additional threshold
coordinate.

The two capacities now have distinct roles.  The total capacity
\(\mathsf B\) controls the largest possible Hankel rank, while \(\kappa\) measures
the directions that can be converted directly into threshold tests.
The remaining step is operational: local tangent information must be
turned into finitely many prefix--suffix words whose signs force a
classical state lower bound.

\subsection{Symmetry release along subgroup chains}

A release of symmetry enlarges the commutant.  If $H\subseteq K$ are
compact groups represented on the same Hilbert space, then
\[
  \mathcal C_K\subseteq\mathcal C_H.
\]
Let
\[
  \mathsf B_K=1+\dim_\C[\mathcal C_K,\mathcal C_K]
\]
denote the maximal $K$-equivariant behavior rank.

Suppose the $K$-representation is
\[
  \Hh\cong\bigoplus_\lambda \mathcal M_\lambda\otimes V_\lambda,
  \qquad m_\lambda=\dim\mathcal M_\lambda,
\]
and that restriction to $H$ has branching rule
\[
  V_\lambda\!\downarrow_H
  \cong
  \bigoplus_\alpha \C^{b_{\alpha\lambda}}\otimes W_\alpha.
\]
For each $H$-type that occurs, define
\[
  \mathcal N_\alpha
  =
  \bigoplus_\lambda
  \mathcal M_\lambda\otimes\C^{b_{\alpha\lambda}},
  \qquad
  \widetilde m_\alpha
  =\sum_\lambda b_{\alpha\lambda}m_\lambda.
\]

\begin{theorem}[Branching-rule symmetry release]\label{thm:branching-release}
The noncommutative capacity after restriction from $K$ to $H$ is
\[
  \nu_H
  =
  \sum_{\alpha:\widetilde m_\alpha>0}
  (\widetilde m_\alpha^2-1)
  =
  \sum_{\alpha:\widetilde m_\alpha>0}
  \left[
    \left(\sum_\lambda b_{\alpha\lambda}m_\lambda\right)^2-1
  \right].
\]
Consequently,
\[
  \mathsf B_H-\mathsf B_K
  =
  \sum_{\alpha:\widetilde m_\alpha>0}
  (\widetilde m_\alpha^2-1)
  -
  \sum_\lambda(m_\lambda^2-1).
\]
\end{theorem}

\begin{proof}
Regrouping the restricted representation by $H$-type gives
$\Hh\cong\bigoplus_\alpha\mathcal N_\alpha\otimes W_\alpha$.
Schur's lemma therefore identifies the $H$-commutant with
$\bigoplus_\alpha\Mat_{\widetilde m_\alpha}(\C)$.  Applying
\cref{thm:commutant-law} gives both formulas.
\end{proof}

\begin{corollary}[Basis-free symmetry-release law]\label{cor:release-difference}
For every compact subgroup inclusion $H\subseteq K$,
\[
  \mathsf B_H-\mathsf B_K
  =
  \dim_\C[\mathcal C_H,\mathcal C_H]
  -
  \dim_\C[\mathcal C_K,\mathcal C_K].
\]
Equivalently,
\begin{align*}
  \mathsf B_H-\mathsf B_K
  ={}&
  \dim_\C\mathcal C_H-\dim_\C\mathcal C_K\\
  &-
  \dim_\C Z(\mathcal C_H)
  +
  \dim_\C Z(\mathcal C_K).
\end{align*}
\end{corollary}

\begin{proof}
Apply \cref{thm:commutant-law} to $H$ and $K$, then use
\[
  \dim_\C[\mathcal C,\mathcal C]
  =
  \dim_\C\mathcal C-\dim_\C Z(\mathcal C)
\]
for finite-dimensional semisimple $C^*$-algebras.
\end{proof}

This formula separates two effects of symmetry release.  The commutant
may acquire new matrix directions, while its center may also change.
Only the net growth of the noncommutative part increases the maximal
word-dependent behavior.

The earlier sector-merging calculus is the explicit partition form of
this general law.  Let $\Pi$ be a partition of the original reduced
sectors.  A block
\(B\in\Pi\) represents sectors that may now mix under full control.
Define
\[
  D_B=\sum_{\lambda\in B}D_\lambda,
  \qquad
  r_B=\sum_{\lambda\in B}r_\lambda,
  \qquad
  q_B=D_B-r_B.
\]
The maximal behavior and accepting-orbit capacities are
\[
  \mathsf B(\Pi)
  =
  1+\sum_{B\in\Pi}(D_B^2-1),
\]
\[
  \kappa(\Pi)
  =
  \sum_{B\in\Pi}2r_Bq_B.
\]

\begin{theorem}[Block-merging release identities]
\label{thm:block-merging-release}
Merging \(m\) blocks gives
\[
  \Delta\mathsf B
  =
  2\sum_{a<b}D_aD_b+(m-1),
\]
\[
  \Delta\kappa
  =
  2\sum_{a<b}(r_aq_b+r_bq_a),
\]
and
\[
  \Delta\mathsf B-\Delta\kappa
  =
  2\sum_{a<b}(r_ar_b+q_aq_b)+(m-1).
\]
\end{theorem}

\begin{proof}
Expand the capacity formulas for the merged block and subtract the sum
of the $m$ original contributions.  The cross terms are precisely the
displayed pairwise sums, and replacing $m$ copies of $-1$ by one copy
contributes $m-1$ to $\Delta\mathsf B$ and to the dark difference.
\end{proof}

For two blocks $A,B$, the formulas become
\[
  \Delta\mathsf B=2D_AD_B+1,
  \qquad
  \Delta\kappa=2(r_Aq_B+r_Bq_A),
\]
\[
  \Delta\mathsf B-\Delta\kappa
  =2(r_Ar_B+q_Aq_B)+1.
\]
The term $2D_AD_B$ counts newly available cross-block Hermitian
coherences, and the additional $1$ is the new relative block
population direction.  The visible part crosses the
accepting--rejecting boundary; the remaining same-outcome coherences
and relative population direction are dark to the chosen two-outcome
readout.

Since \(\mathsf B\) and \(\kappa\) are functions of the partition, these increments
are path independent on the partition lattice.

The readout-visible release fraction
\[
  \phi_{\mathrm{vis}}(A,B)
  =
  \frac{\Delta\kappa}{\Delta\mathsf B}
\]
is approximately \(1/2\) for balanced readouts.  It approaches one
only when one block is predominantly accepting and the other
predominantly rejecting.

These identities show that symmetry release is not uniformly useful to
the readout.  New accepting--rejecting coherences increase the visible
capacity, whereas same-outcome coherences and relative population modes
increase the behavior space without helping the chosen two-outcome readout.
The ratio \(\Delta\kappa/\Delta\mathsf B\) is a capacity-level geometric
fraction: it measures how much newly available tangent space crosses
the chosen accepting boundary.  It is not, by itself, a fixed-language
or task-performance gain.  The gain for a particular automaton is
determined by its reachable--observable pairing, and a threshold
advantage additionally requires a finite sign witness.

\section{Fixed-weight memory and the critical transition}

Fixed-weight permutation modules turn the structural law into an exact combinatorial memory count.  Their capacity truncates the two-row Schur--Weyl sum at the filling, so half filling becomes a Catalan law, while the macroscopic and critical regimes describe how that maximal representation-theoretic memory scale is approached.

\subsection{Multiplicity-free sector decomposition}

Let \(\mathcal H_{n,k}\) be the permutation module on \(k\)-subsets of
\([n]\), equivalently the Hamming-weight-\(k\) subspace of
\((\C^2)^{\otimes n}\).  Assume \(0\leq k\leq n/2\).
Young's rule gives the multiplicity-free decomposition
\cite{Sagan}
\[
  \mathcal H_{n,k}
  \cong
  \bigoplus_{j=0}^{k}S^{(n-j,j)}.
\]
The sector dimensions are
\[
  D_{n,j}
  =
  \binom nj-\binom n{j-1}
  =
  \frac{n-2j+1}{n-j+1}\binom nj,
\]
where \(\binom n{-1}=0\), and
\[
  \sum_{j=0}^{k}D_{n,j}=\binom nk.
\]

Define
\[
  N_{n,k}=\binom nk,
  \qquad
  S_{n,k}=\sum_{j=0}^{k}D_{n,j}^2.
\]
The trivial sector \(j=0\) has dimension one.  Taking the sectors
\(j=1,\ldots,k\) as active gives
\[
  \mathsf B_{n,k}
  =
  1+\sum_{j=1}^{k}(D_{n,j}^2-1)
  =
  S_{n,k}-k.
\]

Because the module is multiplicity free, the image of the complex
group algebra is
\[
  \bigoplus_{j=0}^{k}
  \operatorname{End}
  \bigl(S^{(n-j,j)}\bigr).
\]
Thus the abstract binary generators of the preceding sections can be
realized inside the unitary group of this image algebra.  This is a
representation-sector statement.  It is not a claim that the
resulting operators commute with the physical permutation action, nor
a claim of shallow local implementation.

Up to the lower-order subtraction of the number of nontrivial sectors,
\(S_{n,k}\) is the maximal behavior capacity.  Balanced accepting ranks
also make it the correct order of the binary probabilistic simulation
cost.  The asymptotic problem is therefore reduced to understanding how
the squared dimensions are distributed near the endpoint \(j=k\).

\subsection{Away from half filling}

\begin{theorem}[Non-half-filled macroscopic law]\label{thm:fixed-macro}
If
\[
  \frac{k_n}{n}\longrightarrow
  \alpha\in\left[0,\frac12\right),
\]
then
\[
  \frac{S_{n,k_n}}{N_{n,k_n}^2}
  \longrightarrow
  1-2\alpha.
\]
\end{theorem}

\begin{proof}
Write \(j=k_n-\ell\).  For every fixed \(\ell\geq0\),
\[
  \frac{\binom n{k_n-\ell}}{\binom n{k_n}}
  \longrightarrow
  \left(\frac{\alpha}{1-\alpha}\right)^\ell,
\]
and
\[
  \frac{n-2(k_n-\ell)+1}
       {n-(k_n-\ell)+1}
  \longrightarrow
  \frac{1-2\alpha}{1-\alpha}.
\]
Hence
\[
  \frac{D_{n,k_n-\ell}}{N_{n,k_n}}
  \longrightarrow
  \frac{1-2\alpha}{1-\alpha}
  \left(\frac{\alpha}{1-\alpha}\right)^\ell.
\]
Choose $q$ with
\[
  \frac{\alpha}{1-\alpha}<q<1.
\]
For all sufficiently large $n$ and every $0\leq\ell\leq k_n$,
\[
  \frac{D_{n,k_n-\ell}}{N_{n,k_n}}
  \leq
  \frac{\binom n{k_n-\ell}}{\binom n{k_n}}
  \leq q^\ell.
\]
This summable uniform bound justifies dominated convergence and gives
\[
  \lim_{n\to\infty}
  \frac{S_{n,k_n}}{N_{n,k_n}^2}
  =
  \frac{(1-2\alpha)^2}{(1-\alpha)^2}
  \sum_{\ell=0}^{\infty}
  \left(\frac{\alpha}{1-\alpha}\right)^{2\ell}.
\]
The geometric sum simplifies to \(1-2\alpha\).
\end{proof}

For fixed \(k\),
\[
  S_{n,k}
  \sim
  \frac{n^{2k}}{(k!)^2}.
\]
For every fixed \(\alpha<1/2\), balanced binary automata on these
sectors therefore have worst probabilistic state cost
\[
  \Theta(N_{n,k}^2).
\]

For every limiting density strictly below one half, the endpoint terms
remain geometrically dominant and the simulation cost retains the full
quadratic scale in the Hilbert-space dimension.  This mechanism breaks
down at half filling, where a growing window of irreducible sectors
contributes on the same scale.

\subsection{The Catalan law at half filling}

\begin{theorem}[Exact half-filling law]\label{thm:catalan}
For every \(n\),
\[
  S_{n,\lfloor n/2\rfloor}
  =
  C_n
  :=
  \frac{1}{n+1}\binom{2n}{n}.
\]
Consequently,
\[
  \mathsf B_{n,\lfloor n/2\rfloor}
  =
  C_n-\left\lfloor\frac n2\right\rfloor.
\]
\end{theorem}

\begin{proof}
The Robinson--Schensted correspondence identifies
\[
  \sum_{\lambda\vdash n,\ \ell(\lambda)\leq2}
  (f^\lambda)^2
\]
with the number of permutations whose longest decreasing subsequence
has length at most two.  These are the \(321\)-avoiding permutations,
which are counted by the Catalan number \(C_n\)
\cite{StanleyEC2}.  The partitions with at most two rows are exactly
\((n-j,j)\), and
\[
  f^{(n-j,j)}=D_{n,j}.
\]
\end{proof}

At half filling, the squared-sector-dimension sum is exactly Catalan and the structural capacity is $C_n-\lfloor n/2\rfloor$, the Catalan count with its central-sector correction removed.  Both are determined entirely by the representation profile, before any transition family or accepting readout is chosen.  This gives an independently computable test of the center--commutator law rather than a model-specific construction.

Stirling's formula gives
\[
  C_n
  \sim
  \frac{4^n}{\sqrt\pi\,n^{3/2}}.
\]
At half filling,
\[
  N_{n,\lfloor n/2\rfloor}
  \sim
  2^n\sqrt{\frac{2}{\pi n}},
\]
and hence
\[
  \frac{C_n}{N_{n,\lfloor n/2\rfloor}^2}
  \sim
  \frac{\sqrt\pi}{2\sqrt n}.
\]
Balanced binary automata therefore have worst PFA state cost
\[
  \Theta(C_n)
  =
  \Theta\left(\frac{N_{n,\lfloor n/2\rfloor}^2}{\sqrt n}\right).
\]

The half-filled point is therefore smaller than the naive quadratic
scale by a factor of order \(\sqrt n\).  The transition is not abrupt:
when the distance from half filling is itself of order \(\sqrt n\), the
sector sum converges to a nontrivial interpolation function.

\subsection{The critical window}

\begin{theorem}[Critical transition law]\label{thm:critical-window}
Let \(k_n\leq n/2\) and suppose
\[
  c_n
  =
  \frac{n-2k_n}{\sqrt n}
  \longrightarrow c\in[0,\infty).
\]
Then
\[
  \sqrt n\,
  \frac{S_{n,k_n}}{N_{n,k_n}^2}
  \longrightarrow
  F(c),
\]
where
\[
  F(c)
  =
  c+
  \frac{\sqrt\pi}{2}
  e^{c^2}\erfc(c).
\]
\end{theorem}

\begin{proof}
See \cref{app:critical-proof}.
\end{proof}

At \(c=0\),
\[
  F(0)=\frac{\sqrt\pi}{2},
\]
recovering the Catalan ratio.  As \(c\to\infty\),
\[
  F(c)=c+\frac{1}{2c}+O(c^{-3}),
\]
which matches the onset of the macroscopic law
\(1-2k/n\sim c/\sqrt n\).

\begin{corollary}[Permutation-sector binary simulation law]\label{cor:fixed-weight-pfa}
Let $\operatorname{SC}^{(2)}_{\R}(n,k)$ be the worst real-PFA
simulation cost over binary measure-once automata on $\mathcal H_{n,k}$
whose transitions and readout preserve the displayed Specht-sector
decomposition.  Then
\[
  \operatorname{SC}^{(2)}_{\R}(n,k)
  =\Theta(S_{n,k}).
\]
Balanced nontrivial accepting ranks on the sectors with
$D_{n,j}\geq2$ attain the lower order.
Consequently,
\[
  \operatorname{SC}^{(2)}_{\R}(n,k)
  =
  \begin{cases}
    \Theta(n^{2k}), & k \text{ fixed},\\[2mm]
    \Theta\!\left(\binom nk^2\right),
      & k/n\to\alpha<1/2,\\[2mm]
    \Theta(C_n),
      & k=\lfloor n/2\rfloor.
  \end{cases}
\]
\end{corollary}

\begin{proof}
Apply \cref{cor:worst-profile} to the dimensions
$D_{n,1},\ldots,D_{n,k}$.  Its balanced-rank lower bound and universal
upper bound are both of order
$1+\sum_{j=1}^k(D_{n,j}^2-1)=S_{n,k}-k$; in each displayed asymptotic
regime this has the same order as $S_{n,k}$.  The three cases then
follow from the fixed-$k$, macroscopic, and half-filling formulas above.
\end{proof}

The three regimes are governed by one representation-theoretic sum: endpoint domination away from half filling, Catalan accumulation at the symmetric point, and a Gaussian critical window between them.  The finite \(S_4\) calculation in \cref{app:s4-certificate} follows the same structure-to-witness chain through capacity, Fisher rank, and prepare--test margin; it serves as a concrete consistency check, not as evidence for the asymptotic law.

\section{Intermediate halting and nonhalting behavioral memory}

Intermediate halting tests the three-layer framework at a different boundary.  Absorption into accepting and rejecting accumulators changes the instance pairing, while repeated projection makes the nonhalting corners---including their trace coordinates---the relevant structural space.  Finite weak-leak shattering then realizes the resulting profile capacity, but the scalar boundary below shows why the dynamic phase upgrade from measure-once channels is not automatic.  Consider the standard Kondacs--Watrous measure-many model
\cite{KondacsWatrous,LiQiuEquivalence} with
\[
  \Hh
  =
  \Hh_{\mathrm{acc}}
  \oplus
  \Hh_{\mathrm{rej}}
  \oplus
  \Hh_{\mathrm{non}},
\]
and corresponding orthogonal projectors
$P_{\mathrm{acc}},P_{\mathrm{rej}},P_{\mathrm{non}}$.  For every input
symbol and both endmarkers, the machine applies a unitary $U_\sigma$ and
then measures this three-outcome decomposition.  We assume
\[
  U_\sigma,
  P_{\mathrm{acc}},
  P_{\mathrm{rej}},
  P_{\mathrm{non}}
  \in\mathcal C_K
\]
for every $\sigma$.  The initial state is supported on
$\Hh_{\mathrm{non}}$; acceptance is accumulated throughout the
computation; and any weight remaining nonhalting after the right
endmarker is nonaccepting.  These conventions fix the model used in the
following structural statement.

\begin{proposition}[Absorbing-channel representation]
\label{prop:mm-absorption}
For every symbol or endmarker $\sigma$, define
\begin{align*}
  \widehat\Phi_\sigma(X)
  ={}&
  P_{\mathrm{acc}}XP_{\mathrm{acc}}
  +P_{\mathrm{rej}}XP_{\mathrm{rej}}\\
  &+
  \sum_{j\in\{\mathrm{acc},\mathrm{rej},\mathrm{non}\}}
  P_jU_\sigma P_{\mathrm{non}}X
  P_{\mathrm{non}}U_\sigma^\dagger P_j.
\end{align*}
Then $\widehat\Phi_\sigma$ is completely positive and trace preserving,
and all its Kraus operators belong to $\mathcal C_K$.  Replacing every
intermediate measurement by these absorbing channels preserves the
acceptance probability of every word.  Moreover, for every isotypic
central projection $Z_\lambda$,
\[
  \widehat\Phi_\sigma^*(Z_\lambda)=Z_\lambda.
\]
Hence the total weight of every isotypic component is conserved
throughout the absorbed computation.
\end{proposition}

\begin{proof}
A Kraus family is
\[
  P_{\mathrm{acc}},
  \qquad
  P_{\mathrm{rej}},
  \qquad
  P_jU_\sigma P_{\mathrm{non}}
  \quad
  (j\in\{\mathrm{acc},\mathrm{rej},\mathrm{non}\}).
\]
Its completeness relation is
\begin{align*}
  &P_{\mathrm{acc}}+P_{\mathrm{rej}}
  +P_{\mathrm{non}}U_\sigma^\dagger
  (P_{\mathrm{acc}}+P_{\mathrm{rej}}+P_{\mathrm{non}})
  U_\sigma P_{\mathrm{non}}\\
  &\hspace{8em}
  =P_{\mathrm{acc}}+P_{\mathrm{rej}}+P_{\mathrm{non}}
  =I.
\end{align*}
Thus the map is trace preserving, and every displayed Kraus operator
lies in $\mathcal C_K$.  The accepting and rejecting blocks retain the
probability already absorbed there, while only the nonhalting block is
updated.  Induction over the left endmarker, input, and right endmarker
shows that the accepting trace equals the cumulative measure-many
acceptance probability; this is the total-state bookkeeping used for
general one-way quantum automata \cite{LiEtAl1gQFA}.

Since $Z_\lambda$ is central in $\mathcal C_K$, it commutes with every
Kraus operator $A_r$ above.  Therefore
\[
  \widehat\Phi_\sigma^*(Z_\lambda)
  =
  \sum_rA_r^\dagger Z_\lambda A_r
  =
  Z_\lambda\sum_rA_r^\dagger A_r
  =Z_\lambda.
\]
\end{proof}

Applied to the absorbed dynamics, \cref{thm:channel-exact-rank} identifies the measure-many continuation space with its exact reachable--observable pairing.  Absorption changes the operator coordinates that carry the behavior, but not the instance-geometric invariant.

Write the symmetry decomposition as
\[
  \Hh\cong
  \bigoplus_\lambda
  \mathcal M_\lambda\otimes V_\lambda
\]
and resolve the three halting projectors by
\[
  P_{\mathrm{acc}}
  =
  \bigoplus_\lambda A_\lambda\otimes I_{V_\lambda},
  \qquad
  P_{\mathrm{rej}}
  =
  \bigoplus_\lambda R_\lambda\otimes I_{V_\lambda},
  \qquad
  P_{\mathrm{non}}
  =
  \bigoplus_\lambda N_\lambda\otimes I_{V_\lambda}.
\]
Set
\[
  a_\lambda=\rank A_\lambda,
  \qquad
  r_\lambda^{\mathrm{rej}}=\rank R_\lambda,
  \qquad
  n_\lambda=\rank N_\lambda,
  \qquad
  a_\lambda+r_\lambda^{\mathrm{rej}}+n_\lambda=m_\lambda,
\]
and let
\[
  p_\lambda=\Tr(Z_\lambda\rho_0)
\]
be the initial isotypic weight.  For a feasible fixed profile,
\[
  \sum_\lambda p_\lambda=1,
  \qquad
  p_\lambda>0\Longrightarrow n_\lambda>0,
\]
because the initial state is supported on the nonhalting subspace.
The nonhalting corner has dimension
\[
  D_{\mathrm{non}}
  =
  \dim_\C
  \bigl(P_{\mathrm{non}}\mathcal C_KP_{\mathrm{non}}\bigr)
  =
  \sum_\lambda n_\lambda^2.
\]

\begin{corollary}[Measure-many commutant upper bound]
\label{cor:mm-upper}
Every measure-many automaton satisfying the hypotheses of
\cref{prop:mm-absorption} has
\[
  \beta(\A)
  \leq
  \min\left\{
    1+\nu_K,
    1+D_{\mathrm{non}}
  \right\}.
\]
Consequently, its strict-cutpoint language is recognized by a real
probabilistic finite automaton with at most
\[
  \min\left\{
    1+\nu_K,
    1+D_{\mathrm{non}}
  \right\}+1
\]
states.
\end{corollary}

\begin{proof}
The absorbing channels are Kraus-wise charge zero.  Incorporate the
left endmarker into the initial state and regard the right endmarker as
one additional fixed charge-zero symbol.  The input acceptance series
is a right translate of the resulting measure-once series, so its
Hankel rank cannot increase.  The bound
$\beta(\A)\leq1+\nu_K$ therefore follows from
\cref{cor:kraus-zero-capacity}.

For the corner bound, twirl the absorbed initial state.  After each
input symbol it is enough to retain the pair
\[
  (a,X)
  \in
  \R\oplus
  \Herm\bigl(
    P_{\mathrm{non}}\mathcal C_KP_{\mathrm{non}}
  \bigr),
\]
where $a$ is the acceptance probability already absorbed and $X$ is
the nonhalting block.  Reading a symbol $\sigma$ acts linearly by
\[
  (a,X)
  \longmapsto
  \left(
    a+\Tr(P_{\mathrm{acc}}U_\sigma XU_\sigma^\dagger),
    P_{\mathrm{non}}U_\sigma XU_\sigma^\dagger
    P_{\mathrm{non}}
  \right).
\]
The right endmarker supplies a fixed linear functional of this pair.
Hence the acceptance function has a real linear representation of
dimension $1+D_{\mathrm{non}}$, proving the second rank bound.  Both
this accumulator representation and the charge-zero commutant
representation contain a normalized constant mode.  The
probabilistic-state bound therefore follows from
\cref{thm:codimension-one-pfa}.
\end{proof}

The nonhalting weight inside an isotypic component may decrease, but its
loss is absorbed into accepting or rejecting states in the same
component.  Intermediate measurement therefore does not release the
central coordinates of $\mathcal C_K$.  It does, however, make the
surviving weight of each nonhalting block a dynamical coordinate.  This
closes the capacity once the halting profile is fixed.

Call a label $\lambda$ active when
\[
  p_\lambda>0,
  \qquad
  a_\lambda>0,
  \qquad
  n_\lambda>0,
\]
and write $\Lambda_{\mathrm{act}}^{\mathrm{MM}}$ for the active set.  Put
\[
  D_{\mathrm{MM}}
  =
  \sum_{\lambda\in\Lambda_{\mathrm{act}}^{\mathrm{MM}}}
  n_\lambda^2,
  \qquad
  M_{\mathrm{MM}}=1+D_{\mathrm{MM}}.
\]
For $q\geq1$, let
\[
  B_q^{\mathrm{MM}}
  (\mathbf a,\mathbf r^{\mathrm{rej}},\mathbf n,\mathbf p)
\]
be the supremum of the real Hankel rank over $K$-commuting
measure-many automata with the displayed projector ranks and initial
isotypic weights and with at most $q$ input symbols.  The left and right
endmarkers follow the convention fixed above and are not counted in
$q$.

\begin{theorem}[Exact measure-many profile capacity]
\label{thm:mm-profile-capacity}
If $\Lambda_{\mathrm{act}}^{\mathrm{MM}}$ is empty, then
\[
  B_q^{\mathrm{MM}}
  (\mathbf a,\mathbf r^{\mathrm{rej}},\mathbf n,\mathbf p)=0.
\]
Otherwise, for every $q\geq1$,
\[
  B_q^{\mathrm{MM}}
  (\mathbf a,\mathbf r^{\mathrm{rej}},\mathbf n,\mathbf p)
  =
  M_{\mathrm{MM}}
  =
  1+
  \sum_{\lambda\in\Lambda_{\mathrm{act}}^{\mathrm{MM}}}
  n_\lambda^2.
\]
In particular, one input symbol attains the exact profile capacity.
\end{theorem}

\begin{proof}
For a prefix $x$, let $a_x$ be the acceptance probability already
absorbed and let $X_{x,\lambda}$ be the nonhalting multiplicity block in
label $\lambda$.  Every suffix $y$ has a positive nonhalting effect
$0\leq E_{y,\lambda}\leq N_\lambda$ such that
\[
  f_\A(xy)
  =
  a_x+
  \sum_\lambda\Tr(E_{y,\lambda}X_{x,\lambda}).
\]
If $p_\lambda=0$, then $X_{x,\lambda}=0$ for every prefix.  If
$a_\lambda=0$, then $E_{y,\lambda}=0$ for every suffix.  Labels with
$n_\lambda=0$ have no nonhalting corner.  Hence the Hankel matrix
factors through
\[
  \R\oplus
  \bigoplus_{\lambda\in\Lambda_{\mathrm{act}}^{\mathrm{MM}}}
  \Herm(N_\lambda\mathcal M_\lambda),
\]
a real space of dimension $M_{\mathrm{MM}}$.  This proves the upper
bound.  If the active set is empty, the acceptance function is
identically zero.

Assume now that the active set is nonempty.  For each active label,
choose a contraction $V_\lambda$ and a unit vector
$|\psi_\lambda\rangle\in N_\lambda\mathcal M_\lambda$ as in
\cref{lem:mm-defect-modes}.  All products
$z_{\lambda,ij}=\zeta_{\lambda,i}
\overline{\zeta_{\lambda,j}}$ are globally distinct, and
\[
  \|V_\lambda^t\psi_\lambda\|^2
  =
  \sum_{i,j=1}^{n_\lambda}
  c_{\lambda,ij}z_{\lambda,ij}^{\,t}
\]
with every coefficient nonzero.  The same lemma completes
$V_\lambda$ to a unitary on the nonhalting space plus one accepting
direction.  Extend it by the identity on the remaining multiplicity
coordinates and tensor with $I_{V_\lambda}$.  Use these block unitaries
as one common input symbol, take the left endmarker to be the identity,
and let all residual nonhalting weight after the right endmarker be
nonaccepting.

Choose the active initial multiplicity state in label $\lambda$ to be
$|\psi_\lambda\rangle\langle\psi_\lambda|$ with weight $p_\lambda$.
Every defect leaks into acceptance and no active weight leaks into
rejection.  Therefore
\[
  f_\A(a^t)
  =
  p_{\mathrm{act}}
  -
  \sum_{\lambda\in\Lambda_{\mathrm{act}}^{\mathrm{MM}}}
  p_\lambda
  \sum_{i,j=1}^{n_\lambda}
  c_{\lambda,ij}z_{\lambda,ij}^{\,t},
  \qquad
  p_{\mathrm{act}}
  =
  \sum_{\lambda\in\Lambda_{\mathrm{act}}^{\mathrm{MM}}}p_\lambda.
\]
The constant mode and the $D_{\mathrm{MM}}$ nonconstant modes are all
distinct and have nonzero coefficients.  The Vandermonde factorization
in \cref{lem:exponential-hankel} gives Hankel rank
$1+D_{\mathrm{MM}}$.
\end{proof}

The profile law counts the full Hermitian operator space of every
active nonhalting multiplicity block, including its trace direction.
A defect-one compression is already sufficient.  The construction in
\cref{lem:mm-defect-modes} solves a discrete Stein equation for a
diagonal matrix with prescribed eigenvalues and then conjugates by the
Stein Gramian.  The resulting contraction has rank-one defect but is
generically nonnormal, so one survival probability contains all
ordered products
\(\zeta_i\overline{\zeta_j}\), rather than only the diagonal moduli.
These \(n_\lambda^2\) modes account for the full nonhalting corner,
while accumulated acceptance supplies the shared constant mode.

Let
\[
  \operatorname{SC}_{\R}^{\mathrm{MM}}
  (\mathbf a,\mathbf r^{\mathrm{rej}},\mathbf n,\mathbf p)
\]
be the worst, over arbitrary finite input alphabets and the same
measure-many profile, of the minimum number of states in a real PFA
recognizing the same strict-cutpoint language.

\begin{theorem}[Measure-many probabilistic state law]
\label{thm:mm-pfa}
For every nonempty active profile,
\[
  M_{\mathrm{MM}}
  \leq
  \operatorname{SC}_{\R}^{\mathrm{MM}}
  (\mathbf a,\mathbf r^{\mathrm{rej}},\mathbf n,\mathbf p)
  \leq
  M_{\mathrm{MM}}+1.
\]
The lower bound is witnessed by an alphabet of size at most
\[
  |\Lambda_{\mathrm{act}}^{\mathrm{MM}}|+2
  \leq M_{\mathrm{MM}}+1.
\]
If the active set is empty, the exact value is one.
\end{theorem}

\begin{proof}
The pair \((a,X)\) used in the capacity proof gives a
constant-normalized representation of dimension \(M_{\mathrm{MM}}\),
so the upper bound follows from
\cref{thm:codimension-one-pfa,thm:mm-profile-capacity}.

For the lower bound, set \(D=M_{\mathrm{MM}}-1\) and
\(p=\sum_{\lambda\in\Lambda_{\mathrm{act}}^{\mathrm{MM}}}p_\lambda\).
Choose a rational \(s_0\in(0,p)\).  Blockwise rank-one positive operators
\[
  X=\bigoplus_\lambda
  w_\lambda|\phi_\lambda\rangle\langle\phi_\lambda|,
  \qquad
  0<w_\lambda<p_\lambda,
  \qquad
  \sum_\lambda w_\lambda=s_0,
\]
affinely span the trace-\(s_0\) hyperplane in the active nonhalting
corner.  Rational pure multiplicity projectors and rational weights are
dense in this family, while affine independence is an open condition.
Hence one may fix \(D\) members \(X_1,\ldots,X_D\) in its relative
interior that are affinely independent and have rational coordinates in
the frame used below.

Apply the compiled weak-leak construction of
\cref{lem:mm-compiled-frame}.  Two nonleaking control symbols densely
generate the independent nonhalting multiplicity unitaries, and one
weak-leak symbol is assigned to each active label.  Prefix words
prepare pairs within \(O(\delta)\) of the fixed \(X_i\)'s and of one
common accumulated-acceptance baseline; suffix words implement the
dual rank-one tests; and one further prefix gives a common positive
anchor.  With the amplification, leak strength, and control accuracy
chosen in that order, the resulting finite sign matrix is
\(H_D^{(+1)}\).  Hence its sign-rank is
\(D+1=M_{\mathrm{MM}}\), over an alphabet of size
\(|\Lambda_{\mathrm{act}}^{\mathrm{MM}}|+2\).
\Cref{lem:pfa-cutpoint-rank} gives the claimed PFA lower bound.

If the active set is empty, the language is empty or universal
according to the cutpoint and is recognized by one state.
\end{proof}

Thus the standard measure-many model has a closed profile-dependent
numerical capacity and an exact-order strict-cutpoint obstruction over
an alphabet controlled by the number of addressable active nonhalting labels rather than
the nonhalting operator dimension.

The dynamic lift does not automatically raise this interval to its
upper endpoint.  In the compiled weak-leak frame the static test effect
and the accumulated-acceptance baseline satisfy a common affine
identity.  Appending a clock to the existing tests can therefore leave
a nonzero constant offset, including for the all-positive and
all-negative columns, so the required recurrent cutpoint crossing does
not follow from \cref{lem:mm-compiled-frame}.

This is a genuine boundary, not only a proof artifact.  Consider one
active label with \(n_\lambda=a_\lambda=p_\lambda=1\) and
\(r_\lambda^{\mathrm{rej}}=0\).  Its nonhalting survivor is scalar, so
every word has the form
\[
  f(w)=1-c\prod_{j=1}^{|w|}q_{w_j}
  \qquad(0\leq c,q_{w_j}\leq1),
\]
after incorporating the endmarkers.  A two-state PFA reproduces this
series exactly, while \(M_{\mathrm{MM}}=2\) and the lower bound above
is attainable.  Hence this profile has exact state cost two, not three.
A phase-assisted upgrade, if available, must at least require a
nonhalting corner of dimension two and a new transverse test frame;
the present theorem makes no such unproved claim.

\subsection{Covariant general channels between halting measurements}

The previous profile law uses unitaries that commute with \(K\).
Allowing a covariant general channel before each halting measurement
changes the conserved objects from individual isotypic labels to the
mobility components imposed on the channel.

\begin{definition}[Measure-many general-channel model]
\label{def:mmgqfa}
A \(K\)-covariant measure-many one-way general quantum finite automaton
has the same halting decomposition and endmarker convention as above,
but applies a \(K\)-covariant completely positive trace-preserving map
\(\Phi_\sigma\) before each three-outcome measurement
\cite{LiEtAl1gQFA}.  Fix a partition \(\Pi\) of the isotypic labels.
The automaton is \(\Pi\)-conservative when
\[
  \Phi_\sigma^*(Z_C)=Z_C
  \qquad(C\in\Pi)
\]
for every input symbol and endmarker.
\end{definition}

The absorbing form now becomes
\begin{align}
  \widehat\Phi_\sigma(X)
  ={}&
  P_{\mathrm{acc}}XP_{\mathrm{acc}}
  +P_{\mathrm{rej}}XP_{\mathrm{rej}}
  \notag\\
  &+
  \sum_{j\in\{\mathrm{acc},\mathrm{rej},\mathrm{non}\}}
  P_j\Phi_\sigma(P_{\mathrm{non}}XP_{\mathrm{non}})P_j.
  \label{eq:mmg-absorption}
\end{align}
If \(\Phi_\sigma(X)=\sum_rA_rXA_r^\dagger\), a Kraus family for
\(\widehat\Phi_\sigma\) is
\[
  P_{\mathrm{acc}},\quad P_{\mathrm{rej}},\quad
  P_jA_rP_{\mathrm{non}}.
\]
It follows that \(\widehat\Phi_\sigma\) is completely positive and
trace preserving.  Covariance follows from the invariance of the
halting projectors.  Moreover,
\[
  \widehat\Phi_\sigma^*(Z_C)
  =
  P_{\mathrm{acc}}Z_CP_{\mathrm{acc}}
  +P_{\mathrm{rej}}Z_CP_{\mathrm{rej}}
  +P_{\mathrm{non}}\Phi_\sigma^*(Z_C)P_{\mathrm{non}}
  =Z_C,
\]
so absorption preserves the prescribed mobility components.

For \(C\in\Pi\), define the nonhalting invariant corner and its
dimension by
\[
  \mathfrak N_C
  =
  P_{\mathrm{non}}Z_C\mathcal C_KZ_CP_{\mathrm{non}},
  \qquad
  D_C^{\mathrm{non}}
  =
  \dim_\C\mathfrak N_C
  =
  \sum_{\lambda\in C}n_\lambda^2,
\]
and put \(p_C=\sum_{\lambda\in C}p_\lambda\).  Call \(C\) active when
\[
  p_C>0,\qquad
  P_{\mathrm{acc}}Z_C\neq0,\qquad
  D_C^{\mathrm{non}}>0.
\]
Write \(\Pi_{\mathrm{act}}^{\mathrm{MM},\mathrm{1g}}\) for the active
components and set
\[
  D_{\mathrm{MM}}^{\mathrm{1g}}
  =
  \sum_{C\in\Pi_{\mathrm{act}}^{\mathrm{MM},\mathrm{1g}}}
  D_C^{\mathrm{non}},
  \qquad
  M_{\mathrm{MM}}^{\mathrm{1g}}
  =
  1+D_{\mathrm{MM}}^{\mathrm{1g}}.
\]
Also write
\[
  Z_{\mathrm{act}}
  =
  \sum_{C\in\Pi_{\mathrm{act}}^{\mathrm{MM},\mathrm{1g}}}Z_C,
  \qquad
  p_{\mathrm{act}}
  =
  \sum_{C\in\Pi_{\mathrm{act}}^{\mathrm{MM},\mathrm{1g}}}p_C.
\]

\begin{lemma}[Generic nonhalting contraction]
\label{lem:mmg-generic-contraction}
Let \(\mathfrak N_C\neq0\), and let
\(\rho_C\in\Herm(\mathfrak N_C)\) be a nonzero positive operator.
There is a \(K\)-covariant completely positive trace-nonincreasing map
\(\mathcal T_C\) on the nonhalting \(C\)-subspace such that
\[
  t\longmapsto\Tr\!\left(\mathcal T_C^t(\rho_C)\right)
\]
has real Hankel rank \(D_C^{\mathrm{non}}\).  For finitely many
components, the maps can be chosen so that the spectra of their
restrictions to \(\Herm(\mathfrak N_C)\) are simple, nonzero, avoid
one, and are pairwise disjoint.
\end{lemma}

\begin{proof}
Let \(\omega_C^{\mathrm{non}}\) be the normalized identity on the
nonhalting \(C\)-subspace and choose \(0<s<1\).  The map
\[
  \mathcal S_0(X)
  =
  s\Tr(X)\omega_C^{\mathrm{non}}
\]
has positive-definite Choi matrix and satisfies
\(\mathcal S_0^*(P_{\mathrm{non}}Z_C)=sP_{\mathrm{non}}Z_C\), strictly
below the identity effect.  Complete positivity and the strict
trace-nonincreasing inequality therefore persist in a neighborhood of
\(\mathcal S_0\).

Every real linear map \(L\) on \(\Herm(\mathfrak N_C)\) has a
Hermiticity-preserving \(K\)-covariant extension to the physical
nonhalting block: extend \(L\) complex linearly and compose it with the
Haar twirl.  Hence restrictions of admissible perturbations of
\(\mathcal S_0\) contain a Euclidean-open set in
\(\operatorname{End}_{\mathbb R}(\Herm(\mathfrak N_C))\).
Within the full matrix space, simple nonzero spectrum, avoidance of a
finite prescribed set, and nonsingularity of the controllability and
observability matrices for the fixed vector \(\rho_C\) and trace
functional are complements of proper polynomial zero sets.  They are
proper because a matrix with distinct eigenvalues and a generic
eigenbasis makes both fixed boundary vectors have nonzero coordinates
in every eigenmode.  This Zariski-open set meets the admissible
Euclidean neighborhood.  The reachable--observable rank formula then
gives Hankel rank \(D_C^{\mathrm{non}}\).  Choosing the finitely many
maps successively also makes their spectra pairwise disjoint.
\end{proof}

For \(q\geq1\), let
\[
  B_{q,\Pi}^{\mathrm{MM},\mathrm{1g}}
  (\mathbf a,\mathbf r^{\mathrm{rej}},\mathbf n,\mathbf p)
\]
be the supremum of the real Hankel rank over the
\(\Pi\)-conservative model in \cref{def:mmgqfa} with the displayed
halting ranks and initial isotypic weights, and with at most \(q\)
input symbols.

\begin{theorem}[Exact covariant measure-many channel capacity]
\label{thm:mmg-profile-capacity}
If \(\Pi_{\mathrm{act}}^{\mathrm{MM},\mathrm{1g}}\) is empty, then
\[
  B_{q,\Pi}^{\mathrm{MM},\mathrm{1g}}
  (\mathbf a,\mathbf r^{\mathrm{rej}},\mathbf n,\mathbf p)=0.
\]
Otherwise, for every \(q\geq1\),
\[
  B_{q,\Pi}^{\mathrm{MM},\mathrm{1g}}
  (\mathbf a,\mathbf r^{\mathrm{rej}},\mathbf n,\mathbf p)
  =
  M_{\mathrm{MM}}^{\mathrm{1g}}
  =
  1+
  \sum_{C\in\Pi_{\mathrm{act}}^{\mathrm{MM},\mathrm{1g}}}
  D_C^{\mathrm{non}}.
\]
Thus one input symbol attains the exact capacity.
\end{theorem}

\begin{proof}
For a prefix \(x\), retain the accumulated acceptance \(a_x\) and the
nonhalting invariant blocks \(X_{x,C}\).  Every suffix \(y\) supplies
effects \(0\leq E_{y,C}\leq P_{\mathrm{non}}Z_C\) such that
\[
  f_\A(xy)
  =
  a_x+\sum_C\Tr(E_{y,C}X_{x,C}).
\]
If \(p_C=0\), the corresponding reachable block vanishes.  If
\(P_{\mathrm{acc}}Z_C=0\), its suffix effect vanishes.  If
\(D_C^{\mathrm{non}}=0\), there is no nonhalting block.  The Hankel
matrix therefore factors through
\[
  \mathbb R\oplus
  \bigoplus_{C\in\Pi_{\mathrm{act}}^{\mathrm{MM},\mathrm{1g}}}
  \Herm(\mathfrak N_C),
\]
which has real dimension \(M_{\mathrm{MM}}^{\mathrm{1g}}\).  When the
active set is empty, component conservation makes acceptance
identically zero.

For the lower bound, twirl the fixed initial state and write its
nonzero active blocks as \(\rho_C\), with trace \(p_C\).  Choose maps
\(\mathcal T_C\) from \cref{lem:mmg-generic-contraction}.  The positive
effect
\[
  \Delta_C
  =
  P_{\mathrm{non}}Z_C-\mathcal T_C^*(P_{\mathrm{non}}Z_C)
\]
records the lost trace.  Route that trace to the invariant accepting
state
\[
  \alpha_C
  =
  \frac{P_{\mathrm{acc}}Z_C}
       {\Tr(P_{\mathrm{acc}}Z_C)}
\]
and route \(\mathcal T_C\) back to the nonhalting block.  This
instrument is completely positive, trace preserving on the nonhalting
input, \(K\)-covariant, and confined to component \(C\).  Complete it
on the orthogonal input blocks using invariant replacement states in
the same component.  Taking the direct sum over \(C\) gives one
\(\Pi\)-conservative input channel; inactive components may be left
nonhalting and never accepted.

After \(t\) copies of this symbol,
\[
  f_\A(a^t)
  =
  p_{\mathrm{act}}
  -
  \sum_{C\in\Pi_{\mathrm{act}}^{\mathrm{MM},\mathrm{1g}}}
  \Tr\!\left(\mathcal T_C^t(\rho_C)\right),
  \qquad
  p_{\mathrm{act}}=\sum_{C\in\Pi_{\mathrm{act}}^{\mathrm{MM},\mathrm{1g}}}p_C.
\]
The constant mode and all component modes are nonzero and spectrally
disjoint by the lemma.  Their Vandermonde factorization has rank
\(1+\sum_CD_C^{\mathrm{non}}\), proving equality.
\end{proof}

The mobility partition is essential in this formula.  A general
covariant channel may move initial weight from one label to another
inside the same \(C\), so labels with \(p_\lambda=0\) can still
contribute through the full component corner.  For the singleton
partition the formula reduces to the active-label law above; for one
full mobility component it becomes \(1+D_{\mathrm{non}}\).

Let
\[
  \operatorname{SC}_{\Pi,\mathbb R}^{\mathrm{MM},\mathrm{1g}}
  (\mathbf a,\mathbf r^{\mathrm{rej}},\mathbf n,\mathbf p)
\]
be the worst real-PFA state cost over arbitrary finite input alphabets
for the same covariant general-channel profile.

Define the number of addressable nonhalting labels by
\[
  L_{\mathrm{MM}}^{\mathrm{1g}}
  =
  \#\left\{
    \lambda:
    n_\lambda>0,\ 
    \lambda\in C\text{ for some }
    C\in\Pi_{\mathrm{act}}^{\mathrm{MM},\mathrm{1g}}
  \right\}.
\]

\begin{theorem}[Covariant measure-many channel state law]
\label{thm:mmg-pfa}
For every nonempty active component profile,
\[
  M_{\mathrm{MM}}^{\mathrm{1g}}
  \leq
  \operatorname{SC}_{\Pi,\mathbb R}^{\mathrm{MM},\mathrm{1g}}
  (\mathbf a,\mathbf r^{\mathrm{rej}},\mathbf n,\mathbf p)
  \leq
  M_{\mathrm{MM}}^{\mathrm{1g}}+1.
\]
The lower bound is witnessed by an alphabet of size at most
\[
  L_{\mathrm{MM}}^{\mathrm{1g}}+2
  \leq M_{\mathrm{MM}}^{\mathrm{1g}}+1.
\]
If the active set is empty, the exact value is one.
\end{theorem}

\begin{proof}
The accumulator--corner representation in the capacity proof is
constant normalized and has dimension
\(M_{\mathrm{MM}}^{\mathrm{1g}}\).  The upper bound follows from
\cref{thm:codimension-one-pfa,thm:mmg-profile-capacity}.

For the lower bound, set
\(D=M_{\mathrm{MM}}^{\mathrm{1g}}-1\) and choose a rational
\(s_0\in(0,p_{\mathrm{act}})\).  Choose a relative-interior reference
survivor whose trace in each active component \(C\) is strictly below
the available weight \(p_C\).  In a sufficiently small neighborhood,
rational pure multiplicity projectors and rational weights are dense,
and affine independence is open.  We may therefore fix \(D\) rational,
affinely independent blockwise-rank-one survivors \(X_i\) of common
trace \(s_0\), with label weights \(w_{i,\lambda}>0\), and then choose
post-left-endmarker seed weights \(q_\lambda\) satisfying
\[
  \sum_{\lambda\in C}q_\lambda=p_C,
  \qquad
  q_\lambda>\max_i w_{i,\lambda}
\]
for every addressable label in every active component.

Use the covariant part of \cref{lem:mm-compiled-frame}.  Two nonleaking
commutant-unitary controls act on the multiplicity blocks.  For each of
the \(L_{\mathrm{MM}}^{\mathrm{1g}}\) seeded labels, one weak-leak
channel has survivor
\(I+(\sqrt{1-\delta}-1)Q\) and routes the complementary outcome to an
invariant accepting state in the same mobility component.  The full
extension is completely positive, trace preserving, \(K\)-covariant,
and \(\Pi\)-conservative.  Compiled preparation and test words, with
the parameters chosen in the order specified by the lemma, realize
\(H_D^{(+1)}\).  Its sign-rank is
\(D+1=M_{\mathrm{MM}}^{\mathrm{1g}}\), and the alphabet has
\(L_{\mathrm{MM}}^{\mathrm{1g}}+2\) symbols.  The PFA lower bound follows
from \cref{lem:pfa-cutpoint-rank}.  The empty active case is
a one-state threshold language.
\end{proof}

General covariant channels therefore close both measure-many profile
capacity and its strict-cutpoint state scale.  The witness alphabet is
controlled by the number of addressable nonhalting labels, whereas the state
obstruction is controlled by the full nonhalting operator dimension.

\section{Bridges and limits of operational realization}

The three layers of the theory measure different objects, and their separation is part of the result.  Instance geometry is an equality for one acceptance function, structural capacity is a maximum over a dynamical class, and operational state cost concerns only signs relative to a cutpoint.  The stochastic embedding gives the general upper bridge
\[
  \operatorname{sc}_{\mathrm{PFA},\mathbb R}(\A,\tau)
  \leq\beta_\tau(\A)+1
  \leq\beta(\A)+2.
\]
For every nontrivial saturated charge-zero or component-conserving measure-once channel class, and for every nonempty measure-many profile, if \(M\) is the exact structural capacity then
\[
  M\leq\operatorname{SC}\leq M+1.
\]
A recurrent noncommutative phase reaches the upper endpoint for measure-once channels, whereas the scalar measure-many construction reaches the lower endpoint.  For a reversible prescribed readout the corresponding bridge is
\[
  2+\kappa\leq\operatorname{SC}\leq\mathsf B+1,
\]
so task-visible orbit geometry, rather than total structural capacity alone, controls the available lower obstruction.  These inequalities are the points at which the three filters communicate; none identifies the quantities on its two sides.

The remaining gap for reversible fixed readouts is itself structural.  Finite affine-complete shattering is limited by the dimension of the accepting-projector orbit, so a full Hankel-capacity lower bound cannot always be obtained by enlarging the same sign witness.

\begin{proposition}[Ceiling for affine-complete reversible witnesses]
\label{prop:witness-ceiling}
Fix a reversible profile with a prescribed readout $P$ whose orbit has
dimension $\kappa\geq1$.  Suppose reachable states
$\rho_1,\ldots,\rho_d$ are shattered at one cutpoint by effects in the
readout orbit: for every $\eta\in\{\pm1\}^d$, some $E\in G\cdot P$
satisfies
\[
  \operatorname{sign}\bigl(\Tr(E\rho_j)-\tau\bigr)=\eta_j
  \qquad(1\leq j\leq d).
\]
Then
\[
  d\leq4\kappa\log_2(16e\kappa).
\]
\end{proposition}

The sign-pattern proof is given in \cref{app:operational-boundaries}.  Conceptually, the ceiling is a limitation of the operational layer itself: it bounds every affine-complete witness drawn from the readout orbit, not merely one construction.

For one rank-one block of dimension $m$, $\kappa=2(m-1)$ while the
full numerical capacity is $m^2$.  The proposition shows that an
affine-complete shattering witness can certify only $O(m\log m)$
states in this regime; closing the quadratic fixed-readout gap requires
a different sign matrix or a different lower-bound invariant.  For
balanced ranks, where $\kappa=\Theta(m^2)$, this ceiling leaves the
fixed-readout question open.
The recurrent affine lift adds one state to such a finite affine
obstruction but does not change this asymptotic ceiling, so the two
statements are consistent.

Alphabet restrictions create a separate boundary.  The measure-many lower construction uses two nonleaking controls and one weak-leak letter per addressable active label, and therefore establishes the stated finite-alphabet result rather than a binary theorem.  Within the terminal-reject weak-leak architecture, powers of a single test letter make cumulative acceptance monotone and cannot realize the complete sign frame for dimension at least two; \cref{app:operational-boundaries} gives the argument.  This obstruction is to the present compression mechanism, leaving binary measure-many automata as a distinct semigroup problem.

The state-count resource used throughout is the number of PFA states.  Word length, gate depth, and hardware locality are separate resources; the finite witnesses above have explicit positive margins and finite word lengths, while the fixed-weight group-algebra realization is a representation-sector statement.  The universal claims are analytic; the finite \(S_4\) calculations in \cref{app:s4-certificate} audit concrete formulas and certificates.

\section{Conclusion}

The main conclusion is that symmetry does not have a single ``dimension-reduction'' effect on quantum-automaton memory.  It separates operator coordinates into classes with different computational status, and the three-stage theory identifies that status exactly: reachable--observable pairing gives the memory of one instance, center--commutator and mobility structure give the capacity of a dynamical class, and threshold realization determines which surviving dimensions become unavoidable probabilistic states.  This is the sense in which behavioral memory is produced rather than read off from an ambient Hilbert space.

The full-mobility theorem gives the sharpest expression of the principle.  For a fixed nontrivial invariant readout, a commutative invariant algebra makes the symmetry-reduced automaton stochastic on its central populations, so the full-mobility class has worst state cost equal to the algebra dimension.  If the algebra contains a noncommutative multiplicity block, the unrestricted worst-case cost is exactly one state larger.  The quadratic-plus-one law at trivial symmetry is therefore the endpoint of an algebraic dichotomy, not a separate phenomenon.  Reversible dynamics, component conservation, and intermediate halting show why the broader framework is necessary: they can leave part of the commutant read-only, restrict what the readout sees, or prevent structural coordinates from acquiring the recurrent phase needed for the extra state.

The representation-theoretic consequences show that this distinction has scale.  On one tensor-power Hilbert space, the two Schur--Weyl commutants support polynomial and exponential worst-case memory laws under different preserved symmetries; at fixed weight, half filling gives an exact Catalan squared-sector-dimension sum and the corresponding Catalan-minus-center structural capacity.  Readout-orbit Fisher geometry and subgroup branching make the same theory task-sensitive by quantifying which structural directions a chosen measurement can see and how symmetry release restores them.  The remaining open cases now have a precise location: rank-one reversible readouts require lower-bound invariants beyond affine-complete shattering to approach full quadratic capacity, while binary channel endpoints and phase-assisted measure-many constructions require different finite semigroup realizations.  Those are operational-realization problems, rather than ambiguities about which operator space symmetry makes available.

\clearpage
\appendix

\section{Binary Lie generation}\label{app:binary-lie-proof}
The body states the two-generator construction needed for
\cref{thm:binary-dense}.  We record the frequency-separation proof here.

\begin{proof}[Proof of \cref{lem:binary-lie}]
For an adjacent edge \(e=(\alpha,j,j+1)\), let
\[
  \Delta_e
  =
  2^{m_{\alpha,j+1}}-2^{m_{\alpha,j}}.
\]
If
\[
  2^b-2^a=2^d-2^c
\]
with \(b>a\) and \(d>c\), the two-adic valuation gives \(a=c\), and
then \(b=d\).  Thus all \(|\Delta_e|\) are distinct.

Write
\[
  S_e=i(E_{j,j+1}+E_{j+1,j}),
  \qquad
  A_e=E_{j,j+1}-E_{j+1,j}
\]
inside the \(\alpha\)-block.  Direct calculation gives
\[
  \operatorname{ad}_{iH_0}^2(S_e)=-\Delta_e^2S_e,
  \qquad
  [iH_0,S_e]=\Delta_e A_e.
\]
Because the numbers \(-\Delta_e^2\) are pairwise distinct, for each
edge \(e\) there is a real Lagrange polynomial \(p_e\) which is one at
\(-\Delta_e^2\) and zero at all other edge eigenvalues.  Hence
\[
  p_e\!\left(\operatorname{ad}_{iH_0}^2\right)(iH_1)=S_e.
\]
Every term here is an iterated commutator of \(iH_0\) and \(iH_1\),
with the constant term a scalar multiple of \(iH_1\), so each \(S_e\)
belongs to the generated Lie algebra.  The second displayed identity
then gives \(A_e\), since \(\Delta_e\neq0\), and
\[
  [S_e,A_e]=2i(E_{j+1,j+1}-E_{j,j}).
\]
Thus every adjacent edge supplies its two off-diagonal quadratures and
the corresponding traceless diagonal difference.  Commutators such as
\([E_{j,j+1},E_{j+1,j+2}]=E_{j,j+2}\) propagate along the path, so these
elements generate \(\mathfrak{su}(D_\alpha)\) in each block.  The same
global interpolation isolates edges in different blocks, giving the
direct sum over all \(D_\alpha\geq2\).  The trace subtraction in \(H_0\)
and the zero diagonal of \(H_1\) exclude central \(\mathfrak u(1)\)
directions; blocks with \(D_\alpha=1\) contribute
\(\mathfrak{su}(1)=\{0\}\).
\end{proof}

\section{Fisher geometry of the readout orbit}\label{app:fisher-proof}
\begin{proof}[Proof of \cref{thm:fisher-orbit}]
The spectrum of \(\sigma_P\) consists of \(1/R_P\) on the accepting
subspace and \(0\) on its orthogonal complement.  A projector tangent
has only accepting--rejecting off-diagonal blocks.  Substitution into
the spectral SLD formula gives
\[
  g_{\sigma_P}^{\mathrm{SLD}}
  \left(
  \frac{\dot P_X}{R_P},
  \frac{\dot P_Y}{R_P}
  \right)
  =
  \frac{2}{R_P}\Tr(\dot P_X\dot P_Y).
\]
The Hilbert--Schmidt form is positive definite on the orbit tangent
space, whose dimension is \(\kappa_G(P)\).
\end{proof}

\section{Critical-window asymptotics}\label{app:critical-proof}
\begin{proof}[Proof of \cref{thm:critical-window}]
For an index $j$ write
\[
  x_j=\frac{n-2j}{\sqrt n}.
\]
The mesh spacing of the variables $x_j$ is $2/\sqrt n$.  Fix
$0<\varepsilon<L$.  Uniformly for
$x_j\in[\varepsilon,L]$, the local central limit theorem gives
\[
  \binom nj
  =
  2^n\sqrt{\frac{2}{\pi n}}
  e^{-x_j^2/2}\bigl(1+o(1)\bigr),
\]
while
\[
  \frac{n-2j+1}{n-j+1}
  =
  \frac{2x_j}{\sqrt n}\bigl(1+o(1)\bigr).
\]
Consequently,
\[
  D_{n,j}
  =
  2^n\frac{2\sqrt2}{\sqrt\pi\,n}
  x_j e^{-x_j^2/2}\bigl(1+o(1)\bigr)
\]
uniformly on this truncated interval.  Summing squares and using the
mesh size gives
\[
  \sum_{\substack{j\leq k_n\\
                  \varepsilon\leq x_j\leq L}}
  D_{n,j}^2
  \sim
  \frac{4^{n+1}}{\pi n^{3/2}}
  \int_{\max(c,\varepsilon)}^{L}
  x^2e^{-x^2}\,dx.
\]

It remains to justify the two discarded regions.  On every fixed
interval $0\leq x_j\leq L$ one has the endpoint-safe expansion
\[
  \frac{n-2j+1}{n-j+1}
  =
  \frac{2x_j}{\sqrt n}+O(n^{-1})
\]
uniformly.  Together with the local central estimate, this implies
\[
  \sum_{0\leq x_j<\varepsilon}D_{n,j}^2
  \leq
  C\frac{4^n}{n^{3/2}}
  \bigl(\varepsilon^3+o(1)\bigr)
\]
for a constant $C$ independent of $\varepsilon$.  Standard Gaussian
bounds for binomial coefficients similarly give
\[
  \sum_{x_j>L}D_{n,j}^2
  \leq
  C\frac{4^n}{n^{3/2}}
  \int_L^\infty(1+x^2)e^{-x^2}\,dx.
\]
We may therefore first let $n\to\infty$, then $L\to\infty$, and finally
$\varepsilon\downarrow0$.  This yields
\[
  S_{n,k_n}
  \sim
  \frac{4^{n+1}}{\pi n^{3/2}}
  \int_c^\infty x^2e^{-x^2}\,dx.
\]

At the endpoint $k_n$, the local central limit theorem gives
\[
  N_{n,k_n}^2
  \sim
  4^n\frac{2}{\pi n}e^{-c^2}.
\]
Hence
\[
  \sqrt n\,
  \frac{S_{n,k_n}}{N_{n,k_n}^2}
  \longrightarrow
  2e^{c^2}\int_c^\infty x^2e^{-x^2}\,dx.
\]
Integration by parts gives
\[
  \int_c^\infty x^2e^{-x^2}\,dx
  =
  \frac{c}{2}e^{-c^2}
  +
  \frac{\sqrt\pi}{4}\erfc(c),
\]
which is the claimed function.
\end{proof}

\section{Measure-many contraction and test constructions}
\label{app:mm-constructions}

\begin{lemma}[Defect-one contractions with complete product spectrum]
\label{lem:mm-defect-modes}
Given positive integers $n_1,\ldots,n_s$, there exist diagonalizable
contractions $V_\lambda\in\Mat_{n_\lambda}(\C)$ and unit vectors
$|\psi_\lambda\rangle$ such that
\[
  \rank(I-V_\lambda^\dagger V_\lambda)
  =
  \rank(I-V_\lambda V_\lambda^\dagger)
  =1,
\]
and
\[
  \|V_\lambda^t\psi_\lambda\|^2
  =
  \sum_{i,j=1}^{n_\lambda}
  c_{\lambda,ij}
  z_{\lambda,ij}^{\,t}.
\]
Every coefficient $c_{\lambda,ij}$ is nonzero, and all
$\sum_\lambda n_\lambda^2$ numbers $z_{\lambda,ij}$ are nonzero and
pairwise distinct.  Each $V_\lambda$ is the compression of a unitary on
$\C^{n_\lambda+1}$.
\end{lemma}

\begin{proof}
Choose nonzero numbers
$\zeta_{\lambda,1},\ldots,\zeta_{\lambda,n_\lambda}$ in the open unit
disc so that all products
\[
  \zeta_{\lambda,j}\overline{\zeta_{\lambda,i}}
\]
are pairwise distinct across all labels and ordered pairs.  Such a
choice is generic: each forbidden equality is a proper real-algebraic
hypersurface in a finite product of open discs.

Fix one label and suppress $\lambda$.  Let
\[
  A=\operatorname{diag}(\zeta_1,\ldots,\zeta_n),
  \qquad
  b=(1,\ldots,1)^{\mathsf T},
\]
and define the convergent Stein Gramian
\[
  G
  =
  \sum_{t=0}^\infty A^{\dagger t}bb^\dagger A^t,
  \qquad
  G_{ij}
  =
  \frac{1}{1-\overline{\zeta_i}\zeta_j}.
\]
The matrix $G$ is positive definite.  Indeed, if
$c^\dagger Gc=0$, then
$\sum_i c_i\zeta_i^t=0$ for every $t\geq0$; the first $n$ equations
form a nonsingular Vandermonde system.  The Stein identity is
\[
  G-A^\dagger GA=bb^\dagger.
\]
Set
\[
  V=G^{1/2}AG^{-1/2}.
\]
Then
\[
  I-V^\dagger V
  =
  G^{-1/2}bb^\dagger G^{-1/2}\geq0,
\]
so $V$ is a contraction with a rank-one defect.  Since $V$ is square,
$I-VV^\dagger$ has the same rank.  It is similar to $A$ and therefore
has eigenvalues $\zeta_1,\ldots,\zeta_n$.

Take
\[
  |\psi\rangle
  =
  \frac{G^{1/2}b}{\sqrt{b^\dagger Gb}}.
\]
Then
\[
  \|V^t\psi\|^2
  =
  \frac{1}{b^\dagger Gb}
  \sum_{i,j=1}^n
  G_{ij}
  \bigl(\zeta_j\overline{\zeta_i}\bigr)^t.
\]
Every $G_{ij}$ is nonzero, giving the required coefficients and modes.

Finally write a singular-value decomposition
\[
  V=U\operatorname{diag}(1,\ldots,1,s)W^\dagger,
  \qquad 0\leq s<1.
\]
The unitary
\[
  (U\oplus1)
  \left[
    I_{n-1}\oplus
    \begin{pmatrix}
      s&\sqrt{1-s^2}\\
      \sqrt{1-s^2}&-s
    \end{pmatrix}
  \right]
  (W^\dagger\oplus1)
\]
has $V$ as its compression to the first $n$ coordinates.
\end{proof}

\begin{lemma}[Compiled weak-leak frame]
\label{lem:mm-compiled-frame}
Let
\[
  \mathscr X=\bigoplus_{\lambda\in L}\Herm(\mathbb C^{n_\lambda}),
  \qquad
  D=\sum_{\lambda\in L}n_\lambda^2,
\]
and let \(p_\lambda>0\).  Fix \(D\) affinely independent operators,
with rational coordinates in the rank-one frame used below, of common
trace \(s_0>0\) and of the form
\[
  X_i
  =
  \bigoplus_{\lambda\in L}
  w_{i,\lambda}
  |\phi_{i,\lambda}\rangle\langle\phi_{i,\lambda}|,
  \qquad
  0<w_{i,\lambda}<p_\lambda.
\]
If every block has an accepting exit, there is a commuting-unitary
measure-many witness over at most \(|L|+2\) input symbols whose finite
prefix--suffix sign matrix is \(H_D^{(+1)}\).

If the labels are grouped into mobility components with prescribed
available weights \(p_C>0\), suppose a left endmarker seeds every
listed label with weights \(q_\lambda\) such that
\(q_\lambda>\max_i w_{i,\lambda}\) and
\(\sum_{\lambda\in C}q_\lambda=p_C\) in each active component, and
every component has an invariant accepting exit.  Then the same
conclusion holds for
component-conservative covariant channels over at most \(|L|+2\)
symbols.  The two control symbols are nonleaking; the remaining symbol
for each label is a weak leak.  Approximation is allowed only inside
the construction: the resulting threshold signs are exact.
\end{lemma}

\begin{proof}
Take the right endmarker to act as the identity on the survivor and
declare all residual nonhalting mass nonaccepting.  Thus the eventual
test effect of a weak-leak word with nonhalting survivor \(V_y\) is
\(I-V_y^\dagger V_y\).

Choose two unitaries \(g_0,g_1\) that preserve every nonhalting block
and whose positive-word closure is
\(\prod_{\lambda\in L}SU(n_\lambda)\); they act as the identity on the
halting subspaces.  For each label fix a nonhalting multiplicity
direction \(q_{\lambda,0}\) and an accepting multiplicity direction
\(\alpha_\lambda\).  A label leak \(\ell_\lambda\) rotates these two
directions, tensored with \(I_{V_\lambda}\), through an angle with
\(\sin^2\theta=\delta\), and is the identity elsewhere.  Its
nonhalting survivor along \(Q=|q\rangle\langle q|\) is
\[
  S_Q(\delta)
  =
  I+(\sqrt{1-\delta}-1)Q.
\]
Positive control words approximate both \(W\) and \(W^{-1}\), so a
word of the form \(W^{-1}\ell_\lambda^kW\) compiles \(k\) aligned
leaks along any prescribed rank-one \(Q\) in that block.  The control
letters themselves never change accumulated acceptance.

In an \(n\)-dimensional block with basis
\(e_1,\ldots,e_n\), use the rank-one projections
\[
  Q_j=|e_j\rangle\langle e_j|
\]
and, for \(j<k\),
\[
  Q_{jk,\pm}^{R}
  =
  \left|
    \frac{e_j\pm e_k}{\sqrt2}
  \right\rangle
  \left\langle
    \frac{e_j\pm e_k}{\sqrt2}
  \right|,
\]
\[
  Q_{jk,\pm}^{I}
  =
  \left|
    \frac{e_j\pm i e_k}{\sqrt2}
  \right\rangle
  \left\langle
    \frac{e_j\pm i e_k}{\sqrt2}
  \right|.
\]
There are \(2n^2-n\) such projections.  Their real span is
\(\Herm(\mathbb C^n)\), because their pairwise differences give the
real and imaginary off-diagonal matrix units, while the \(Q_j\) give
the diagonal units.  They also form a tight frame in the elementary
sense
\[
  \sum_{Q\in\mathcal Q_n}Q=(2n-1)I_n.
\]

Let \(C_0\) be a common multiple of all \(2n_\lambda-1\), and assign
the positive integer
\[
  r_Q=\frac{C_0}{2n_\lambda-1}
\]
to every frame projection in block \(\lambda\).  On the direct sum,
\[
  \sum_Qr_QQ=C_0I.
\]
Let \(G_1,\ldots,G_D\) be Hilbert--Schmidt dual to
\(X_1,\ldots,X_D\).  All frame Gram entries are rational, so the
rational-coordinate hypothesis gives a common positive integer
\(R\) and integers \(b_{\eta,Q}\) such that
\[
  R\sum_{i=1}^D\eta_iG_i
  =
  \sum_Qb_{\eta,Q}Q
\]
for every sign vector \(\eta\).  Choose an integer \(L_0\) large enough
that
\[
  N_{\eta,Q}=L_0r_Q+b_{\eta,Q}\geq0
\]
for every \(\eta\) and \(Q\).  Then
\[
  F_\eta:=\sum_QN_{\eta,Q}Q
  =L_0C_0I+R\sum_{i=1}^D\eta_iG_i,
\]
and therefore
\[
  \Tr(F_\eta X_i)=L_0C_0s_0+R\eta_i.
\]

Compile each frame projection \(Q\) by conjugating the single leak of
its label, and group repeated uses as
\(W_Q^{-1}\ell_\lambda^{N_{\eta,Q}}W_Q\).  Thus the number of control
approximations is fixed by the frame size, even though the exponents
may be large.  Since
\[
  S_Q(\delta)=I-\frac{\delta}{2}Q+O(\delta^2),
\]
the eventual acceptance effect of the resulting test word satisfies
\[
  E_\eta(\delta)
  =
  \delta F_\eta+O(\delta^2R^2)+O(\zeta),
\]
uniformly over \(\eta\), where \(\zeta\) is the total dense-control
error.  Here \(L_0\), every \(N_{\eta,Q}\), and the norm of \(F_\eta\)
are \(O(R)\) after the initial rational data are fixed.

It remains to compile the preparations.  Let \(r_\lambda=p_\lambda\)
in the commuting-unitary realization and \(r_\lambda=q_\lambda\) in
the covariant realization.  Starting from weight \(r_\lambda\) in a
pure nonhalting direction, aligned leaks produce the grid
\(r_\lambda(1-\delta)^k\).  If \(n_\lambda\geq2\), a final
leak along a direction with continuously variable overlap fills the
interval between two consecutive grid values exactly; a final
nonleaking control sets the survivor direction to
\(\phi_{i,\lambda}\).  If \(n_\lambda=1\), the grid approximates any
fixed target weight with error \(O(\delta)\).  Consequently there are
prefix words producing pairs
\[
  (a_i,X_i')
  \quad\text{with}\quad
  |a_i-a_0|+\|X_i'-X_i\|
  \leq C\delta+O(\zeta),
  \qquad
  a_0=\sum_\lambda r_\lambda-s_0,
\]
for a constant \(C\) independent of \(R\) and \(\delta\).  A long
aligned-leak prefix leaves total survivor below \(s_0/4\) and supplies
the anchor row.

Choose the parameters in the following order.  First scale the integer
dual construction so that \(R>4C\).  Next choose \(\delta>0\) with
\(\delta R\ll1\), making the product error
\(O(\delta^2R^2)\) smaller than the signal \(\delta R\).  Finally take
\(\zeta\ll\delta R\).  At the cutpoint
\[
  \tau=a_0+\delta L_0C_0s_0,
\]
the \(i\)-th preparation and \(\eta\)-th test then have sign
\(\eta_i\).  The anchor is positive because its accumulated acceptance
exceeds \(a_0+3s_0/4\), while
\(\delta L_0C_0s_0<s_0/4\).  Hence the finite sign matrix is
\(H_D^{(+1)}\).

For the covariant realization, the left endmarker first prepares the
specified invariant seed weights \(q_\lambda\) on the listed
nonhalting labels.  For a multiplicity-space rank-one projection
\(Q\) in label \(\lambda\), set on the full physical space
\[
  \widetilde Q=Q\otimes I_{V_\lambda},
  \qquad
  S=I-\bigl(1-\sqrt{1-\delta}\bigr)\widetilde Q,
\]
where \(\widetilde Q\) is zero, and hence \(S\) is the identity, on
all other labels and on the halting subspaces.  Replace the unitary
leak by the full channel
\[
  \Phi_{\lambda,Q}(X)
  =
  SXS^\dagger
  +\Tr\!\bigl((I-S^\dagger S)X\bigr)\alpha_C,
\]
where \(\alpha_C\) is a normalized invariant accepting state in the
same component.  The explicit survivor and the measure--prepare branch
make this map completely positive, trace preserving,
\(K\)-covariant, and \(\Pi\)-conservative.  The preceding preparation,
test, and error analysis is unchanged.  In either realization the
alphabet is precisely the two controls together with one leak per
label.
\end{proof}

\section{Operational boundary arguments}\label{app:operational-boundaries}

\subsection{Affine-complete shattering ceiling}

\begin{proof}[Proof of \cref{prop:witness-ceiling}]
In an active block of dimension \(m_\lambda\) and readout rank \(r_\lambda\), a standard graph chart for the complex Grassmannian uses \(2r_\lambda(m_\lambda-r_\lambda)\) real coordinates.  Its projector entries are rational functions with positive denominator \(\det(I+Z_\lambda^\dagger Z_\lambda)\) of degree \(2r_\lambda\).  After the blockwise denominators are cleared, each function \(\Tr(E\rho_j)-\tau\) has the sign of a real polynomial in \(\kappa\) variables of degree at most
\[
  2\sum_\lambda r_\lambda\leq\kappa.
\]
The standard coordinate-subspace charts cover the orbit with
\[
  J=\prod_\lambda\binom{m_\lambda}{r_\lambda}
  \leq2^{\sum_\lambda m_\lambda}
  \leq2^\kappa
\]
charts, since every active block satisfies \(m_\lambda\leq2r_\lambda(m_\lambda-r_\lambda)\).

If \(d<\kappa\), the claimed bound is immediate.  Otherwise Warren's polynomial sign-pattern theorem bounds the number of strict sign vectors realized in one chart by \((4ed)^\kappa\) \cite{Warren}.  Across all charts the orbit therefore realizes at most \((8ed)^\kappa\) patterns.  Shattering gives \(2^d\leq(8ed)^\kappa\), or
\[
  \frac d\kappa
  \leq
  \log_2\!\left(8e\kappa\frac d\kappa\right).
\]
The function \(x-\log_2x\) is increasing for \(x\geq2\), and evaluating the inequality at \(x=4L\), where \(L=\log_2(16e\kappa)\), gives \(4L>L+1+\log_2L=\log_2(8e\kappa\cdot4L)\).  Hence no \(x\geq4L\) can satisfy it.
\end{proof}

\subsection{Why the compiled weak-leak frame is not directly binary}

In the compiled measure-many witness, the right endmarker acts trivially on the survivor and all residual nonhalting mass is terminally nonaccepting.  Suppose, within this architecture, that every test were a power \(t^N\) of one letter.  For a fixed preparation prefix \(x_i\), each further application of \(t\) can only transfer additional survivor weight into the accepting accumulator, so \(f_\A(x_it^N)\) is nondecreasing in \(N\).  Relative to a fixed cutpoint, the sign in row \(i\) can therefore change at most once.  As \(N\) varies, the column sign vectors form a chain in which each of the \(D\) coordinates flips at most once, and hence contain at most \(D+1\) distinct patterns.  The complete frame \(H_D\) contains \(2^D\) patterns, so no such power family realizes it when \(D\geq2\).  This proves only the architecture-specific obstruction stated in the main text: a right-endmarker survivor effect or another semigroup construction need not obey the same monotonicity.

\section{Common-character Kraus boundary}\label{app:common-character-boundary}

\begin{proof}[Proof of \cref{prop:common-character}]
The intertwining relation sends the \(\lambda\)-isotypic summand into the \(T_\chi(\lambda)\)-summand, so
\[
  Z_\mu A_r=A_rZ_{T_\chi^{-1}(\mu)}.
\]
Trace preservation forbids an occurring source label from having an absent target.  Since tensoring by a character is invertible and the occurring set is finite, \(T_\chi\) permutes that set.  Trace preservation now gives
\[
  \Phi^*(Z_\mu)
  =\sum_rA_r^\dagger Z_\mu A_r
  =\left(\sum_rA_r^\dagger A_r\right)Z_{T_\chi^{-1}(\mu)}
  =Z_{T_\chi^{-1}(\mu)}.
\]
Summing over an orbit proves orbit conservation, and \cref{thm:component-channel-capacity} gives the capacity bound.

To see that equality can fail, take \(K=\mathbb Z_2\times\mathbb Z_2\) with each of its four one-dimensional irreducible representations occurring once, and choose a nontrivial character \(\chi\).  Tensoring by \(\chi^{-1}\) partitions the four labels into two transpositions.  Here \(D_K=4\), so the component bound is \(4-2+1=3\).  Yet each common-character channel sends the population of every one-dimensional label deterministically to its paired label.  Starting from an invariant state and using an invariant readout, every word behavior therefore depends only on the parity of its length and has Hankel rank at most two.  Orbit conservation is thus the exact universal consequence of the common-character assumption, not an automatic saturation theorem.
\end{proof}

\section{Finite \texorpdfstring{$S_4$}{S4} certificate and numerical validation}\label{app:s4-certificate}

The Hamming-weight-two module for four sites is the smallest example in
which two nontrivial sectors contribute simultaneously.  It instantiates
the behavior capacity, readout capacity, Fisher rank, and analytic
prepare--test margin without being used as evidence for the asymptotic
theorems.

Take the Hamming-weight-two permutation module:
\[
  \mathcal H_{4,2}
  \cong
  S^{(4)}
  \oplus
  S^{(3,1)}
  \oplus
  S^{(2,2)},
\]
with dimensions \(1,3,2\).  Use the one-dimensional sector as an
always-accepting spectator and take accepting rank one in each active
sector.

\begin{table}[ht]
\centering
\begin{tabular}{c@{\qquad}c@{\qquad}c@{\qquad}c@{\qquad}c}
\toprule
sector & \(D_\alpha\) & \(r_\alpha\)
& \(D_\alpha^2-1\) & \(2r_\alpha q_\alpha\)\\
\midrule
\((4)\)   & \(1\) & \(1\) & \(0\) & \(0\)\\
\((3,1)\) & \(3\) & \(1\) & \(8\) & \(4\)\\
\((2,2)\) & \(2\) & \(1\) & \(3\) & \(2\)\\
\bottomrule
\end{tabular}
\caption{Sector capacities and rank-one readout dimensions for the
Hamming-weight-two $S_4$ module.}
\label{tab:s4-sector-data}
\end{table}

Hence
\[
  \mathsf B=1+8+3=12,
  \qquad
  \kappa=4+2=6,
\]
and the binary state bounds are
\[
  8
  \leq
  \operatorname{sc}_{\mathrm{PFA},\R}
  \leq
  12+1
  =13.
\]

The active dynamical Lie algebra is
\[
  \mathfrak{su}(3)\oplus\mathfrak{su}(2)
\]
of dimension \(8+3=11\).  The projector stabilizer has dimension
\(4+1=5\), so the orbit dimension is \(11-5=6\), agreeing with \(\kappa\).

The total accepting projector has rank
\[
  R_P=3.
\]
Its intrinsic SLD Fisher metric has rank six and six equal nonzero
eigenvalues
\[
  \frac{2}{R_P}=\frac23.
\]

The explicit first-order certificate contains
\[
  (\kappa+1)2^{\kappa}=7\cdot64=448
\]
prefix--suffix signs, including the affine anchor row.  It gives the
seven-state sub-obstruction with an explicit margin.  The curved
rank-one witness in \cref{thm:curved-witness} adds the radial row and
uses
\[
  (\kappa+2)2^{\kappa+1}=8\cdot128=1024
\]
strict signs to obtain the displayed eight-state lower bound.  For the
first-order margin, assign weight $p=1/3$ to each active sector and
weight $1/3$ to the always-accepting spectator.  The spectator places
the common cutpoint at $1/2$, while the active derivatives retain the
factor $p=1/3$.  Choosing
\[
  t=\frac1{8\kappa}=\frac1{48}
\]
gives the conservative margin
\[
  \gamma=\frac{p}{16\kappa}=\frac1{288}.
\]
If every word is estimated from \(N_{\mathrm{shot}}\) independent
Bernoulli trials, Hoeffding's inequality and a union bound give total
failure probability at most \(\delta\) whenever
\[
  N_{\mathrm{shot}}
  \geq
  \frac{1}{2\gamma^2}
  \log
  \left(
    \frac{2(\kappa+1)2^{\kappa}}{\delta}
  \right).
\]
This finite calculation instantiates the structure-to-witness chain.
It is not evidence for asymptotic experimental efficiency.

This example illustrates the role of finite calculations in the paper:
they check the dimension and margin formulas in a concrete instance,
while the scaling laws themselves remain analytic.

More generally, numerical calculations may reproduce finite Lie closures, branching and release identities, Catalan sums, critical-window scaling, the \(448\)-sign tangent certificate, and the \(1024\)-sign curved certificate above.  These checks diagnose implementations of the formulas; no numerical observation is used to prove a universal capacity or state-cost theorem.

\end{document}